\documentclass[letterpaper, 10 pt, conference]{ieeeconf}  
\IEEEoverridecommandlockouts                              
\makeatletter
\let\NAT@parse\undefined
\makeatother

\usepackage[hidelinks]{hyperref}
\usepackage{graphicx} 
\usepackage{amsmath} 
\usepackage{amssymb}  
\usepackage{cleveref}
\usepackage{booktabs}
\usepackage{algorithm}
\usepackage{algorithmic}
\usepackage[compress, sort]{cite}

\usepackage{xcolor}
\usepackage{cuted}
\usepackage{multirow}
\usepackage{xfrac}

\let\labelindent\relax
\usepackage{enumitem}

\usepackage{tikz}
\usetikzlibrary{arrows.meta}

\newtheorem{remark}{Remark}
\newtheorem{example}{Example}
\newtheorem{definition}{Definition}

\newtheorem{proposition}{Proposition}

\newcommand{\mysec}[2]{\section{#1}}
\newcommand{\mysubsec}[2]{\subsection{#1}}
\newcommand{\miniskip}[0]{\vspace{2pt plus 0.5pt minus 0.5pt}}

\newcommand*{\eigvec}{w}
\newcommand*{\eigval}{\hat{\lambda}}
\newcommand{\innerprod}[2]{\langle #1, #2 \rangle}
\newcommand{\apprx}{\operatorname{approx}}
\newcommand{\aug}{\operatorname{aug}}
\newcommand{\nD}{n}

\newcommand*{\eigfun}{\varphi}
\newcommand*{\eigfunvec}{\eigfun_{[n]}}

\newcommand*{\basisfun}{\psi}
\newcommand*{\basisfunvec}{\Psi}
\newcommand*{\basisfunvecbar}{\bar{\Psi}}
\newcommand*{\basismat}{\Psi}

\newcommand*{\basisfunvariant}{\zeta}
\newcommand*{\basisfunvecvariant}{\basisfunvariant_{[n]}}

\newcommand*{\observable}{g}

\newcommand*{\decoderfun}{h}
\newcommand*{\decoderfunvec}{\decoderfun}

\newcommand{\longthmtitle}[1]{\mbox{}{\textit{(#1):}}}

\title{\LARGE \bf
Koopman operator theory: fundamentals, control, and applications
}

\author{Igor Mezi{\'c}, Jorge Cort\'es, Karl Worthmann, Mircea Lazar, and Armin Lederer
\thanks{${}^*$Authors are listed in the order of their talks during the tutorial session.}%
\thanks{I. Mezi\'c is with the Department of Mechanical Engineering, University of California, Santa Barbara, USA
        {\tt\small mezic@ucsb.edu}}%
\thanks{J. Cort\'es is with the Department of Mechanical and Aerospace Engineering, University of California, San  Diego, USA
        {\tt\small cortes@ucsd.edu}}%
\thanks{K.~Worthmann is with the Optimization-based Control Group, 
TU Ilmenau, 
Germany
        {\tt\small karl.worthmann@tu-ilmenau.de}}%
\thanks{M.~Lazar is with the Constrained Control of Complex Systems Lab, Control Systems Group, Department of Electrical Engineering, 
TU Eindhoven, The Netherlands
        {\tt\small m.lazar@tue.nl}}%
\thanks{A.~Lederer is with the Department of Electrical and Computer Engineering, National University of Singapore, Singapore
        {\tt\small armin.lederer@nus.edu.sg}}%
}

\begin{document}

\maketitle
\thispagestyle{empty}
\pagestyle{empty}

\begin{abstract}
    The Koopman operator has
    gained considerable attention due to its ability to provide a global linear representation of highly complex dynamical systems. 
    The operator describes nonlinear dynamics in a linear way through the lens of real- or complex-valued observable functions. 
    Data-driven techniques, like extended dynamic mode decomposition (EDMD), kernel EDMD, 
    and machine-learning methods, can be used to generate finite-dimensional approximations accompanied by finite-data error bounds.
    In this tutorial paper, we provide a concise introduction into Koopman operator theory and its use in systems and control. 
    A particular focus is put on data-driven surrogate models, their extension to systems with inputs, and controller design using Koopman operator theory. 
    Moreover, we demonstrate the key techniques, i.e., EDMD and Koopman MPC. To this end, we provide simulation studies including source code on GitHub to enable the interested reader to experience the Koopman operator in systems and control step by step.\looseness=-1
\end{abstract}

\section{Motivation and outline}
\label{sec:intro}

Analysis and control of complex nonlinear systems pose fundamental challenges in modern control theory, particularly regarding closed-form controller design~\cite{
krstic95:nonlinear}, the rigorous characterization of asymptotic properties such as stability~\cite{khalil02:nonlinear} and invariant sets~\cite{blanchini08:set}, and the computation of accurate long-term predictions~\cite{mezic2024koopman}. Traditionally, these issues are mitigated either through highly specialized, domain-specific solutions that require costly expert knowledge~\cite{recht19:tour}, or through a linearization around nominal operating points~\cite{shamma02:analysis}. However, these conventional approaches exhibit significant limitations: expert-driven designs lack broad applicability, while locally valid linear approximations often experience severe performance degradation outside their restricted operating envelopes if applicable at all~\cite{MullWort17:quadratic}. Consequently, the development of flexible, global control methodologies for complex nonlinear systems remains a critical open problem.

Koopman operator theory offers a compelling perspective to address these challenges by providing an alternative representation of nonlinear dynamical systems via a potentially infinite-dimensional, yet  
linear composition operator \cite{MeziBana04, mezi05, RowlMezi09}. Because this framework yields a global linear representation, it inherently facilitates the resolution of the aforementioned difficulties: it significantly simplifies controller design by rendering classical control-theoretic techniques applicable~\cite{StraWort26}, enables the analysis of asymptotic properties such as stability and invariant sets using techniques from linear systems theory~\cite{mauroy2016global}, and allows for accurate long-term predictions through a linear propagation of latent states \cite{bevanda23:koopman}.
Coupled with its natural suitability for data-driven formulations, these advantages have catalyzed the widespread adoption of Koopman operator-theoretic approaches within the systems and control community \cite{bevanda21:koopman, brunton22:modern, Mauroy20Susu20:book}. However, practical implementation necessitates the approximation of this infinite-dimensional operator using finite data sets to yield a tractable, finite-dimensional model. 
Deriving approximations that are simultaneously amenable to control design, preserve insightful spectral information, and maintain long-term predictive accuracy demands carefully designed approximation techniques; see, e.g., \cite{StraScha26:SafEDMD} for a comprehensive treatment of approximation errors and closed-loop guarantees, and \cite{shi26:koopman} for a recent survey emphasizing robotic applications.

\paragraph*{Overview}
This tutorial paper accompanies our tutorial session at the $65$th IEEE Conference on Decision and Control, Honolulu, Hawaii. 
We will introduce the foundations of Koopman operator theory, present elementary techniques for determining data-driven approximations of the composition operator, explain dedicated techniques for controller design based on Koopman models, and highlight recently established connections to machine learning. 
On the one hand, we will provide a comprehensive overview of Koopman operator-theoretic methods.
On the other hand, we will \textit{decode} a selection of methods to enable the reader to apply the presented tool chain to complex systems. 
Further, we aim to facilitate the understanding and applicability of the presented methods through demonstrations on representative applications, which are publicly available at
\begin{center}
    \url{https://github.com/KOT-tutorial/CDC26}    
\end{center}

\begin{table}[htb!]
\setlength{\tabcolsep}{1pt}
    \begin{center}
    \caption{Outline and structure of the paper.}
    \label{tab:overview}
    \small
    \begin{tabular}{p{0.25cm} p{7.35cm} c}
    \toprule
        Section & & Page\\
        \midrule
        \multicolumn{2}{p{7.5cm}}{\ref{sec:intro} Motivation and outline} \vspace{0.10cm} & \pageref{sec:intro} 
        \\
        \multicolumn{2}{p{7.5cm}}{\ref{sec:KOT} Koopman operator theory} & \pageref{sec:KOT}
        \\
        & \ref{subsec:eigenfunctions} Koopman eigenfunctions and spectral analysis & \pageref{subsec:eigenfunctions}
        \\
        & \ref{subsec:Koopman-invariance} Koopman invariance \vspace{0.10cm}  & \pageref{subsec:Koopman-invariance}
        \\
        \multicolumn{2}{p{7.5cm}}{\ref{sec:EDMD} Extended dynamic mode decomposition} & \pageref{sec:EDMD}\\
        & \ref{subsec:EDMD} EDMD algorithm & \pageref{subsec:EDMD}\\
        & \ref{subsec:EDMD_spectral} Eigenfunctions and spectral approximations & \pageref{subsec:EDMD_spectral}\\
        & \ref{subsec:kEDMD} Kernel EDMD & \pageref{subsec:kEDMD} \\
        & \ref{subsec:EDMD_error} EDMD error analysis: projection and estimation & \pageref{subsec:EDMD_error}\\
        & \ref{subsec:EDMD_error_full} Kernel EDMD: full approximation error bounds \vspace{0.10cm} & \pageref{subsec:EDMD_error_full}\\
        \multicolumn{2}{p{7.5cm}}{\ref{sec:control} Extension to control systems} &  \pageref{sec:control}\\
        & \ref{subsec:LTI} On exact LTI embeddings & \pageref{subsec:LTI}\\
        & \ref{subsec:control:family} Koopman control family & \pageref{subsec:control:family}\\
        & \ref{subsec:product:Hilbert} Koopman operator for control via tensor products & \pageref{subsec:product:Hilbert}\\
        & \ref{subsec:EDMD:control} EDMD with control & \pageref{subsec:EDMD:control}\\
        & \ref{subsec:EDMDc:implementation} Implementation and numerical comparison & \pageref{subsec:EDMDc:implementation}\\
        & \ref{subsec:observer} Koopman observers \vspace{0.10cm}  & 
        \pageref{subsec:observer}\\
        \multicolumn{2}{p{7.5cm}}{\ref{sec:controller} Controller design} & \pageref{sec:controller} \\
        & \ref{subsec:closed-form_control} Closed-form control laws & \pageref{subsec:closed-form_control} \\
        & \ref{subsec:MPC:theory} Model predictive control & \pageref{subsec:MPC:theory} \\
        & \ref{subsec:MPC:tutorial} A hands-on tutorial on Koopman MPC \vspace{0.10cm}  & \pageref{subsec:MPC:tutorial} \\
        \multicolumn{2}{p{7.5cm}}{\ref{sec:ML} Koopman meets machine learning} & \pageref{sec:ML}\\
        & \ref{subsec:stat_learn} Statistical learning theory for Koopman  models & \pageref{subsec:stat_learn}\\
        & \ref{subsec:Bayesian} Bayesian perspective on kernel methods & \pageref{subsec:Bayesian} \\
        & \ref{subsec:deep_koopman} Deep Learning-based Koopman models  & \pageref{subsec:deep_koopman} \\
        & \ref{subsec:ml_beyond} Exploiting Koopmanism in machine learning \vspace{0.10cm}  &  \pageref{subsec:ml_beyond}\\
        \multicolumn{2}{p{7.5cm}}{\ref{sec:outlook} Outlook} \vspace{0.10cm}  & \pageref{sec:outlook}\\
        \multicolumn{2}{p{7.5cm}}{References} & \pageref{sec:references}\\
        \bottomrule
    \end{tabular}
    \end{center}
    \vspace{-0.5cm}
\end{table}

\miniskip
\noindent\Cref{sec:KOT}: \textbf{Koopman operator theory}. 
We provide a concise introduction to Koopman operator theory 
emphasizing two topics of particular importance for prediction and stability analysis of complex dynamical systems. 
In \Cref{subsec:eigenfunctions}, we briefly consider the spectrum of the Koopman operator with a focus on the discrete part and its relation to Koopman eigenfunctions. \textit{Koopman eigenfunctions} are pivotal
\begin{itemize}
    \item to uncover inherent structures of nonlinear dynamical systems, e.g., invariant sets or symmetries
    \item for reliable long-term predictions of quantities of interest along the flow of complex dynamical systems
    \item to conduct a stability analysis
\end{itemize}
Then, we deal with \textit{Koopman invariance} in \Cref{subsec:Koopman-invariance}. First, we briefly recap known results on Koopman invariance of function spaces and spectral decompositions. 
We explain the role that projections w.r.t.\ a given inner product play in 
generating approximations of the operator over finite-dimensional spaces and  how the recently proposed concept of \textit{invariance proximity} can help in building (approximately) invariant representations.
\miniskip

\noindent\Cref{sec:EDMD}: \textbf{Extended Dynamic Mode Decomposition (EDMD)}.
We recap EDMD as a data-driven approach to compute a surrogate model of the Koopman operator. 
After presenting the foundations of EDMD for discrete-time 
systems, we demonstrate its application to compute a linear predictor in a hands-on example in \Cref{subsec:EDMD}.
We further discuss how EDMD 
can be employed to approximate the Koopman generator for continuous-time systems and demonstrate how the previously discussed invariance proximity can be computed in closed-form. Due to the significance of eigenvalues and eigenfunctions, we introduce an EDMD-based formulation for their approximation, which we illustrate in a numerical example in \Cref{subsec:EDMD_spectral}. 
In \Cref{subsec:kEDMD}, we present kernel EDMD as an important variant of EDMD. Further, we provide intuitive explanations for kernels and their induced function spaces. Finally, 
we provide an overview on existing results on approximation errors for EDMD in \Cref{subsec:EDMD_error} and kernel EDMD in \Cref{subsec:EDMD_error_full}, emphasizing the role of suitably chosen function spaces for Koopman invariance and
the error analysis.
\miniskip

\noindent\Cref{sec:control}: \textbf{Koopman representations for nonlinear systems with inputs}. 
We extend Koopman operator theory and (kernel) EDMD to dynamical systems with control inputs. 
To this end, we first introduce an infinite-dimensional linear time-invariant (LTI) representation, provide an in-depth treatment of such exact LTI Koopman embeddings, 
and point out their limitations in \Cref{subsec:LTI}. 
Exact LTI Koopman models are particularly attractive, offering the potential to fully leverage results from linear systems theory for nonlinear control systems.
Then, we present the Koopman control family as a recently proposed framework to unify existing results on linear and bilinear Koopman models and briefly discuss its connection to
nonlinear representations in \Cref{subsec:control:family}. A product Hilbert space perspective on systems with control inputs is provided in \Cref{subsec:product:Hilbert}.
In \Cref{subsec:EDMD:control}, we present EDMD extensions to control systems and illustrate all three approaches by numerical simulations.
Finally, in \Cref{subsec:observer}, we provide an introduction to Koopman observer design.

\noindent\Cref{sec:controller}: \textbf{Controller design}.
We leverage the presented Koopman operator theory, the data-driven surrogate models generated by (kernel) EDMD, and the respective extensions to systems with inputs to systematically construct stabilizing controllers for nonlinear dynamical systems. 
To this end, we illuminate key challenges in data-driven controller design with an emphasis on robustness to model-plant mismatch. 
In \Cref{subsec:closed-form_control}, we briefly recap recently proposed controller designs such as LQR control using linear Koopman models and robust controller design using bilinear representations. Koopman model predictive control (MPC) with and without (stabilizing) terminal conditions is presented in \Cref{subsec:MPC:theory}. 
Leveraging the previously presented error analysis, we show that Koopman operator theory allows to rigorously verify all assumptions on the data-driven surrogate model required by this framework for data-driven MPC.
Moreover, we provide a step-by-step tutorial on Koopman MPC including source code on GitHub, a detailed description of its implementation and illustrating examples to highlight pros and cons of the different Koopman MPC designs in \Cref{subsec:MPC:tutorial}.
        
\noindent\Cref{sec:ML}: \textbf{Koopman operator meets machine learning}. 
Due to the data-driven nature of methods for approximating Koopman operators and the prevalence of dynamical systems in machine learning, there has been a significant mutual influence between the fields in recent years. We first illustrate the benefits of statistical learning theory to analyze Koopman operator learning approaches without assumptions on the data distribution. To this end, \Cref{subsec:stat_learn} provides a brief introduction to statistical learning theory exemplified for EDMD. As alternative to the frequentist setting of statistical learning theory, we introduce a Bayesian formulation for Koopman operator learning in \Cref{subsec:Bayesian}. Practical neural network approaches for learning deep Koopman operator models and their relationship to modern \textit{world-model} AI architectures are discussed in \Cref{subsec:deep_koopman}. Finally, we highlight the benefits that Koopman-based 
models can provide in machine learning algorithms such as diffusion models, neural network pruning, and reinforcement learning in \Cref{subsec:ml_beyond}.\looseness=-1

\section{Koopman operator theory} 
\label{sec:KOT}

We consider the discrete-time dynamical system
\begin{equation}\label{eq:dynamics:DT}
    x^+ = F(x),
\end{equation}
where the successor state~$x^+$ is given by the image of the continuous map $F: \Omega \rightarrow  \Xi$ applied to the current state~$x$. 
While $\Omega, \Xi$ might, in general, be  topological spaces, we choose $\Omega \subseteq \mathbb{R}^{d}$ and $\Xi := \{ y \in \mathbb{R}^d \mid \exists\,x \in \Omega: F(x) = y \}$ to keep the presentation technically simple and refer to \cite[Appendix~A]{kohne2025error} for regularity conditions on the boundary of the open domain~$\Omega$.

In the 1930s, Bernard Osgood Koopman\footnote{See P.M.~Morse's article~\cite{Morse82:Koopman} for biographical information.} proposed to consider the dynamical system~\eqref{eq:dynamics:DT} through the lens of observables, i.e., measurable functions $\observable: \Xi \rightarrow \mathbb{K}$ contained in some function space~$\mathcal{F}(\Xi)$, where the field~$\mathbb{K}$ denotes either the real or complex numbers, i.e., $\mathbb{K} = \mathbb{C}$ or $\mathbb{K} = \mathbb{R}$, respectively. 
Then, the Koopman operator~$\mathcal{K}$ is defined by 
\begin{equation}\label{eq:Koopman:identity}
    \mathcal{K} \observable = \observable \circ F, 
\end{equation}
see the original works~\cite{Koop31} and~\cite{KoopNeum32}. 
The Koopman identity~\eqref{eq:Koopman:identity} states that, for every state~$x \in \Omega$, the observable~$\observable \in \mathcal{F}(\Xi)$ evaluated at the successor state~$F(x)$ exhibits the very same function value as the propagated observable $\mathcal{K}\observable$ at~$x$.\footnote{Formally, measurability is defined w.r.t.\ some measure~$\mu$. Hence, we only require equality in the Koopman identity~\eqref{eq:Koopman:identity} $\mu$-almost everywhere.}
Hence, the Koopman operator $\mathcal{K}$ maps functions $\observable \in \mathcal{F}(\Xi)$ to functions defined on the set~$\Omega$, 
which also explains the terminology \textit{Kolmogorov backward operator}, which is used for stochastic systems, see Remark~\ref{rem:stochastic-dynamics}.
The Koopman operator is linear since we have
\begin{equation}\nonumber 
    \mathcal{K}(\alpha \observable + \beta h) = (\alpha \observable + \beta h) \circ F = \alpha (\observable \circ F) + \beta (h \circ F) 
\end{equation}
for all $\alpha, \beta \in \mathbb{K}$ and $\observable,h \in \mathcal{F}(\Xi)$. 
Koopman operator theory allows for an equivalent \textit{linear} description of the nonlinear map~$F$, which is, however, infinite dimensional.
Further, the linear Koopman operator is bounded if the function space~$\mathcal{F}$ is Koopman invariant, i.e., $\mathcal{K}\observable \in \mathcal{F}$ holds for all $\observable \in \mathcal{F}$, e.g., on the space of bounded continuous functions or $\mathcal{F} = L^p$, $p \in [1,\infty]$, due to positivity, see \cite[Section~3]{ColbMaur26}, where also characterizations of Koopman invariance can be found.

Analogously, one may begin with a continuous-time dynamical system governed by the ordinary differential equation\looseness=-1
\begin{equation}\label{eq:dynamics:CT}\tag{ODE}
    \dot{x}(t) = f(x(t)),
\end{equation}
where $f: \Omega \rightarrow \mathbb{R}^d$ is a locally Lipschitz continuous map. We assume forward invariance of the open and sufficiently regular 
set $\Omega \subseteq \mathbb{R}^d$ to ensure global existence and uniqueness of solutions~$x(t;\hat{x})$ on $[0,\infty)$ for the initial value problem consisting of~\eqref{eq:dynamics:CT} and the initial condition $x(0;\hat{x}) = \hat{x}$ for all $\hat{x} \in \Omega$.
Then, the strongly continuous semigroup $(\mathcal{K}^t)_{t \geq 0}$ of bounded linear operators is defined by $\mathcal{K}^t \observable = \observable \circ x(t;\cdot)$ for each $t \geq 0$.
Further, on the domain $\mathcal{D}(\mathcal{L})$ of the Koopman generator defined by $\{ \observable \in \mathcal{F}(\Omega) \mid \exists\, \lim_{t \searrow 0} \frac 1t \mathcal{K}^t \observable \}$, the generator of the Koopman semigroup is given by the Lie derivative
\begin{equation}\label{eq:generator:identity}
   (\mathcal{L} \observable)(x) = \langle \nabla \observable(x), f(x) \rangle \qquad\forall\,x \in \Omega.
\end{equation}
We emphasize that the Koopman generator is, in general, an unbounded operator.

\begin{remark}\label{rem:stochastic-dynamics}\longthmtitle{Stochastic dynamics}
    We may define the Koopman operator also for stochastic dynamics, see, e.g., \cite{mezic2000comparison,vcrnjaric2020koopman,sinha2020robust,wanner2022robust,HertPhil25}. Let us consider the stochastic differential equation
    \begin{equation}\label{eq:SDE}\tag{SDE}
        \mathrm{d}X(t) = f(X(t))\,\mathrm{d}t + \sigma(X(t))\,\mathrm{d}W_t
    \end{equation}
    with drift~$f$, diffusion~$\sigma$, and Brownian motion~$W_t$. 
    Then, the Koopman semigroup $(\mathcal{K}^t)_{t \geq 0}$ corresponding to~\eqref{eq:SDE} is defined by\looseness=-1
    \begin{equation}\label{eq:Koopman:stochastic}
        \mathcal{K}^t \observable(x) = \mathbb{E}[ \observable(X_t) \mid X_0 = x],
    \end{equation}
    i.e., the expectation value of the stochastic process~$X_t$ conditioned on the initial value $X_0 = x$ at time $t = 0$. 
    The associated Koopman generator~\eqref{eq:generator:identity} can be represented as $\mathcal{L} = f \cdot \nabla + \frac 12 \sigma \sigma^T : \nabla^2$ with $A:B := \sum_{i,j=1}^d a_{ij}b_{ij}$ being the standard Frobenius inner product for matrices.
    Stochastic dynamics are of particular interest, e.g., w.r.t.\ molecular dynamics~\cite{wu2017variational,luzzatto2026data} and due to parabolic regularity, which may alleviate the verification of required smoothness properties of stochastic expectation operators and, thus, the Koopman operator defined by~\eqref{eq:Koopman:stochastic}, see \cite[Section~4.3]{hertel2026verifiable}.
\end{remark}

In the remainder of this section, we 
recap two fundamental concepts from
Koopman operator theory: eigenfunctions and invariance.
Eigenfunctions correspond to dominant patterns of dynamical systems characterizing, among others, the stability behaviour~\cite{mezic2020spectrum}. 
Koopman invariance is another essential property. 
On the one hand, invariance of function spaces is key to enable a spectral analysis and, in addition, to rigorously derive bounds on the approximation error. 
On the other hand, Koopman invariance of finite-dimensional (sub-) spaces is of particular importance since it 
provides the theoretical foundation of (exact) finite-dimensional linear representations of the action of the Koopman operator and, as we show later, will inform Koopman-based controller design.

\subsection{Koopman eigenfunctions and spectral analysis}
\label{subsec:eigenfunctions}

While finite-dimensional operators on $n$-dimensional complex or real vector spaces always have $n$ eigenvalues (counting multiplicities) by the Fundamental Theorem of Algebra, linear operators on infinite-dimensional (complex) Banach spaces~$\mathcal{F}$ have a more complicated spectral structure \cite{yosida2012functional}. 
Let $\mathcal K:\mathcal{F}\to\mathcal{F}$ be the Koopman operator. 
Then, the complex plane decomposes as the disjoint union\looseness=-1
\begin{equation*}
    \mathbb C = \rho(\mathcal K)\,\dot\cup\,\sigma_p(\mathcal K)\,\dot\cup\,\sigma_c(\mathcal K)\,\dot\cup\,\sigma_r(\mathcal K)
\end{equation*}
with the following components: The \textit{resolvent set}~$\rho(\mathcal K)$ consists of all $\lambda\in\mathbb C$ such that
$\mathcal K-\lambda I$ is bijective and its inverse
\begin{equation*}
    R(\lambda) = (\mathcal K-\lambda I)^{-1}:\mathcal{F}\to\mathcal{F}
\end{equation*}
exists and is bounded. 
The \textit{point spectrum}~$\sigma_p(\mathcal K)$ consists of all $\lambda\in\mathbb C$ such that
$\mathcal K-\lambda I$ is not injective.
Equivalently, $\lambda\in\sigma_p(\mathcal K)$ if there exists a nonzero
$f\in\mathcal{F}$ satisfying $\mathcal Kf=\lambda f$. 
The \textit{continuous spectrum}~$\sigma_c(\mathcal K)$ consists of all $\lambda\in\mathbb C$ such that
$\mathcal K-\lambda I$ is injective,
$\operatorname{Ran}(\mathcal K-\lambda I)$ is dense in $\mathcal{F}$,
but $(\mathcal K-\lambda I)^{-1}$ does not exist as a bounded operator.
And, finally, the \textit{residual spectrum}~$\sigma_r(\mathcal K)$ consists of all $\lambda\in\mathbb C$ such that
$\mathcal K-\lambda I$ is injective, but
$\operatorname{Ran}(\mathcal K-\lambda I)$ is not dense in $\mathcal{F}$.
\looseness=-1

In practice, we are 
particularly interested in the point spectrum of the Koopman operator since it contains eigenvalues $\lambda\in\mathbb C$.
Clearly, $(\mathcal K-\lambda I)^{-1}$, or equivalently $(\mathcal K^t-e^{\lambda t}I)^{-1}$, is unbounded for these eigenvalues, which illustrates the similarity to eigenvalues of matrices.
Therefore, we can analogously assign a `direction' to the eigenvalue, which corresponds to an eigenfunction. 
\begin{definition}\label{def:Koopman:eigenfunction}\longthmtitle{Koopman eigenfunction; point spectrum}
    An eigenfunction $\eigfun\in\mathcal{F}$ is an observable that, given an eigenvalue $\lambda\in\mathcal{C}$, satisfies 
    \begin{equation}\label{eq:eigfun_discrete}
        \mathcal K\eigfun=\lambda\eigfun
    \end{equation}
    in the discrete-time case and
    \begin{equation}
        \mathcal K^t\eigfun=e^{\lambda t}\eigfun \qquad\forall\,t \in [0,\infty)
    \end{equation}
    in the continuous-time case.  
\end{definition}
Note that we retain the notation $\lambda$ for eigenvalues rather than $e^{\lambda t}$ in the continuous-time case.

In discrete time, we have the following result on the algebraic structure of eigenfunctions under products. 
We note that an analogous result holds for continuous-time systems.
\begin{proposition}\label{prop:efunc-semigroup}
    Let $\mathcal K$ be the Koopman operator associated with the dynamics~\eqref{eq:dynamics:DT}.
    Assume $\mathcal F$ is a subset of all $\mathbb C$-valued functions such that $\mathcal F$ is a commutative algebra, i.e., it forms a vector space closed under pointwise products  
    and contains the constant function equal to one. 
    Then the set $\mathcal E$ of eigenfunctions of $\mathcal K$ 
    is an Abelian monoid under pointwise multiplication of functions.
    In particular, if $\eigfun_1,\eigfun_2\in\mathcal{F}$ are eigenfunctions of $\mathcal K$ with eigenvalues $\beta_1$ and $\beta_2$, then $\eigfun_1\eigfun_2$ is an eigenfunction of $\mathcal K$ with eigenvalue $\beta_1\beta_2$.
\end{proposition}
\begin{proof}
    First note that the constant function equal to one
    is 
    an eigenfunction of the Koopman operator with eigenvalue~$1$. This is also the identity element of the monoid.

    Let the function~$\observable$ be defined by $\observable(x)=\eigfun_1(x)\eigfun_2(x)$ for the eigenfunctions $\mathcal K\eigfun_1 =\beta_1\eigfun_1$ and $\mathcal K\eigfun_2=\beta_2\eigfun_2$. Then, we have
    \begin{align}
        \mathcal K\observable( x) & \stackrel{\eqref{eq:Koopman:identity}}{=} (\observable \circ F) (x) = \eigfun_1(F(x))\eigfun_2(F(x)) \nonumber \\
    &=\mathcal K\eigfun_1(x)\,
      \mathcal K\eigfun_2(x) \nonumber = \beta_1\beta_2
      \eigfun_1(x)\eigfun_2(x) \\
    & = \beta_1\beta_2\,\observable(x) , \nonumber
    \end{align}
    showing that 
    the set of eigenfunctions is closed under pointwise multiplication.
\end{proof}
\miniskip

\noindent
\textbf{Koopman Mode Decomposition (KMD)}. 
Koopman Mode Decomposition, first formulated in~\cite{mezi05} for on-attractor measure-preserving dynamics, holds for a large class of dynamical systems, including systems that are not necessarily measure preserving. 
Hence,
Koopman eigenvalues may lie in the complex plane and not only on the unit circle.

Let $\mathcal K$ be the Koopman operator acting on a suitable space of observables such that observables~$\observable$
admit a decomposition into a part associated with the point spectrum and a part associated with the continuous spectrum, see, e.g., \cite{mezic2020spectrum} for the construction of such spaces for dissipative systems. 
Then, the evolution of $\observable$ 
under $\mathcal K^k$ may be written formally as\looseness=-1
\begin{equation}\label{eq:KMD}
    \begin{split}
        \mathcal K^k \observable(x)
        =
        P_0 \observable(x)
        +
        \sum\nolimits_i \lambda_i^k
        \langle \observable,\eigfun_i^*\rangle\,\eigfun_i
        (x) \\
        +
        \int_{\sigma_c(\mathcal K)}
        \lambda^k\,\observable(\lambda)\,\eigfun
        (x,\lambda)\,\mathrm{d}\mu_{\observable}(\lambda),
    \end{split}
\end{equation}%
where $P_0$ denotes the projection onto the eigenspace associated with the eigenvalue $1$, the functions $\eigfun_i$ are Koopman eigenfunctions satisfying \eqref{eq:eigfun_discrete} and $\eigfun_i^*$ 
denotes the corresponding dual eigenfunction or spectral projection functional. 
The remaining integral 
term represents the contribution from the continuous spectrum of $\mathcal K$, with spectral parameter $\lambda\in\sigma_c(\mathcal K)\subset\mathbb C$ and spectral measure $\mathrm{d}\mu_{\observable}(\lambda)$ associated with the observable~$\observable$. 
The function $\varphi(x,\lambda)$, referred to as eigenmeasure, generalizes the concept of eigenfunctions, see~\cite{mezic2020spectrum} for details. 
Note that the eigenvalues and eigenfunctions are properties of the operator, 
while the spectral projection terms $\langle \observable,\eigfun_i^*\rangle$, the so-called Koopman modes, depend on the observable.
The point spectrum reflects the almost periodic component of the system dynamics in KMD, 
i.e., it describes the behavior towards equilibria. In contrast, the continuous spectrum captures chaotic behavior that is not related to fixed points of the dynamics and is empty for many types of dynamical systems, see~\cite{mezi05}. The 
remainder of this tutorial paper will focus mainly on the point spectrum as common in the literature; see also~\cite{Mezi16} for an in-depth discussion on spectral properties of the Koopman operator.

While the exact Koopman Mode Decomposition~\eqref{eq:KMD} provides a linear representation of the dynamics, it contains a potentially infinite dimensional sum. Therefore, we are often interested in the Koopman operator restricted to a finite-dimensional subspace~$\mathbb{V}$ of the function space~$\mathcal{F}$, e.g., the span of finitely many 
eigenfunctions $\mathbb{V}=\operatorname{span}(\{ \eigfun_i \mid i \in [n] \})$.
Concatenating these eigenfunctions into a vector~$\eigfunvec$, i.e., $[\eigfunvec]_i=\eigfun_i$, we can then define the finite-dimensional restricted Koopman operator $\mathcal{K}|_{\mathbb{V}}$ by minimizing the error in the identity~\eqref{eq:Koopman:identity}, i.e., 
\begin{align}\label{eq:LS_EV}
    \min\nolimits_{\mathcal{K}|_{\mathbb{V}}} \|\eigfunvec\circ F - \mathcal{K}|_{\mathbb{V}} \eigfunvec\|,
\end{align}
where the action of the operator on a vector-valued map is defined in a componentwise 
manner.
Using Definition~\ref{def:Koopman:eigenfunction} of Koopman eigenfunctions, it is straightforward to see that the solution to this optimization problem 
can be represented through a finite-dimensional matrix $K=\mathrm{diag}(\lambda_1,\ldots,\lambda_n)\in\mathbb{C}^{n\times n}$ with eigenvalues $\lambda_i$, $i\in [n]$. These properties can be generalized to arbitrary finite-dimensional Koopman-invariant subspaces. 
Indeed, every observable function~$\observable$ contained in the subspace~$\mathbb{V}$ can be represented as a linear combination of basis functions in~$\mathbb{V}$
enabling fast evaluation of the action of the Koopman operator. 
For application to arbitrary observables~$\observable$, $\observable\in\mathcal{F}$, this matrix defines the action of the Koopman operator on the observable~$\observable$ projected on the subspace~$\mathbb{V}$, i.e., 
$\mathcal{K} P_{\mathbb{V}}$, 
via 
\begin{equation}\label{eq:EDMD_operator_eigenfunction}
    \mathcal{K} P_{\mathbb{V}} = \sum_{i=1}^n \lambda_i \eigfun_i \left<\cdot,\eigfun_{i}^*\right> = \underbrace{\eigfunvec^\top}_{\text{lifting}}\hspace*{0.5mm}\underbrace{K^\top_{\phantom{[n]}} }_{\text{propagation}} \underbrace{T^{-1} \langle  \eigfunvec, \cdot \rangle}_{\text{reconstruction}},
\end{equation}
where $[T]_{ij}= \langle \eigfun_i, \eigfun_j \rangle$ 
and $\langle \eigfunvec, \cdot \rangle$ is the vector of componentwise 
scalar products.\footnote{We assume complex inner products to be linear in their first argument and sesquilinear in their second argument throughout this paper.
} 
In this expression, the reconstruction vector $v= T^{-1} \langle \eigfunvec, \observable \rangle$ effectively describes the observable~$\observable$ in the coordinate system induced by the eigenfunctions~$\eigfun_i$, i.e., it is the vector of Koopman modes. 
Therefore, \eqref{eq:EDMD_operator_eigenfunction} simplifies to $\mathcal{K}P_{\mathbb{V}} \eigfunvec=K\eigfunvec$ when applied to the eigenfunction vector $\eigfunvec$. Crucially, the usage of the span of eigenfunctions~$\mathbb{V}$ 
ensures that the linearity of~\eqref{eq:KMD} is preserved, resulting in $(\mathcal{K} P_{\mathbb{V}})^k = \mathcal{K}^k P_{\mathbb{V}}$ and, thus,
\begin{align}\label{eq:multistep_pred}
    \!(\mathcal{K} P_{\mathbb{V}})^k  \!=\! \sum_{i=1}^n \lambda_i^k \eigfun_i \left<\cdot,\eigfun_i^*\right> \!=\! (K^k\eigfunvec)^{\hspace*{-0.5mm}\top} T^{-1} \langle \eigfunvec, \cdot \rangle\!
\end{align}%
Similar definitions are possible for the Koopman generator.
\miniskip

\noindent \textbf{Stability analysis using Koopman eigenfunctions}. 
A useful application of eigenfunctions in the context of control theory lies in the global analysis of stability. 
The core of this approach lies in the observation that an eigenvalue of the Koopman generator of a stable dynamical system has a negative real part if and only if its corresponding eigenfunction equals zero at the attractor of the dynamical system~\cite{mauroy2016global}. 
If a non-constant eigenfunction exhibits a different value at these points, the corresponding eigenvalue must be zero. 
Since an eigenfunction of the Koopman generator satisfies
\begin{align*}
    \mathcal{L}\eigfun \stackrel{\eqref{eq:generator:identity}}{=} \langle \nabla \eigfun, f \rangle = \lambda \eigfun,
\end{align*}
the negative real part of the eigenvalue immediately implies that the corresponding system modes vanish asymptotically, see, e.g., \cite{mauroy2020koopman}. Hence, the intersection of the zero level sets of these eigenfunctions describes the global 
attractor. Importantly, this approach not only allows the analysis of equilibria, but straightforwardly extends to limit cycles as well. Moreover, it intuitively induces potentially non-smooth Lyapunov functions
\begin{align*}
    V(x) = \left( \sum\nolimits_{i=1}^{d} c_i|\eigfun_i(x)|^p \right)^{1/p}
\end{align*}
with $c \in \mathbb{R}_{>0}$ and $p\geq 1$, see~\cite{mauroy2013spectral} for details.

\subsection{Koopman invariance}
\label{subsec:Koopman-invariance}

A subspace $\mathbb{V} \subset \mathcal{F}$ is invariant under the Koopman operator
if the following inclusion holds:
\begin{align*}
  \mathcal{K} \observable \in \mathbb{V} \qquad\forall\,\observable \in \mathbb{V}.
\end{align*}
As an example, the subspace generated by a collection of eigenfunctions is automatically invariant. 
However, 
finite-dimensional Koopman-invariant subspaces capturing complete
information about the system are often unavailable analytically and can be difficult to compute with high accuracy, particularly when the original system is not known. 
Nevertheless, their study is key because the concept of subspace invariance determines the general form of finite-dimensional models and provides a bedrock for approximations on non-invariant subspaces.\looseness=-1

Koopman invariance and, then, boundedness of the Koopman operator can be rigorously shown for infinite-dimensional spaces, see~\cite{ColbMaur26}. For example, the space of bounded continuous functions is Koopman invariant provided that the dynamics~$F$ in~\eqref{eq:dynamics:DT} is continuous. Koopman invariance has also been shown on suitably chosen reproducing kernel Hilbert spaces leveraging Sobolev regularity; both for deterministic~\cite{kohne2025error} and stochastic dynamics~\cite{HertPhil25,hertel2026verifiable}. 
However, finite-dimensional subspaces are, in general, not invariant. 
Hence, the action of the Koopman operator on them cannot be equivalently represented by a matrix, which necessitates 
approximations of the Koopman operator on non-invariant subspaces accompanied by an in-depth analysis, as we explain next and further elaborate in Subsections~\ref{subsec:EDMD_error} and~\ref{subsec:EDMD_error_full} w.r.t.\ the error analysis.

Assume that the space~$\mathcal{F}$ is equipped with a complex inner
product $\innerprod{\cdot}{\cdot}: \mathcal{F} \times \mathcal{F} \to \mathbb{K}$.
For instance, one can think of $\mathcal{F} = L^2_\mu(\Omega)$  
equipped 
with
$\innerprod{\observable_1}{\observable_2}= \int_\Omega \observable_1(x) \overline{\observable_2(x)}\,\mathrm{d}\mu (x)$,    
where $\bar{\cdot}$ denotes the complex conjugate.
The inner product 
induces a norm $\| \cdot \|: \mathcal{F} \to \mathbb{R}_{\geq 0}$ by $\|\observable\| = \sqrt{\innerprod{\observable}{\observable}}$. To approximate the action
of the Koopman operator on a finite-dimensional space
$\mathbb{V} \subset \mathcal{F}$ which is not Koopman invariant, 
consider the orthogonal projection operator $\mathcal{P}_{\mathbb{V}}: \mathcal{F} \to \mathbb{V}$, mapping a function in $\mathcal{F}$ to the closest function in~$\mathbb{V}$. 
We then approximate the Koopman operator~$\mathcal{K}$
by
\begin{align*}
    \mathcal{K}_{\apprx} := \mathcal{P}_{\mathbb{V}} \mathcal{K}: \mathcal{F} \to \mathbb{V} \subset \mathcal{F} .
\end{align*}
Note that, even though the subspace $\mathbb{V}$ might not be invariant under $\mathcal{K}$, it is invariant under $\mathcal{K}_{\apprx}$. This means that, with a similar exposition as the one presented in Subsection~\ref{subsec:eigenfunctions}, one can describe the action of $\mathcal{K}_{\apprx}$ on $\mathbb{V}$ with a matrix representation. 
Given a basis $\basisfunvec = (\basisfun_1,\dots,\basisfun_{\nD})$, $\nD=\dim{\mathbb{V}}$, for $\mathbb{V}$, we 
obtain
\begin{align}\label{eq:restriction-on-noninvariant-basis}
	\mathcal{K} \basisfunvec \approx \mathcal{K}_{\apprx} \basisfunvec
	=  K_{\apprx} \basisfunvec, 
\end{align}
where $K_{\apprx} \in \mathbb{K}^{\nD \times \nD}$. Moreover,
similarly to~\eqref{eq:EDMD_operator_eigenfunction}, for any
function $\observable \in \mathbb{V}$ with representation $\observable = v^\top \basisfunvec$, one can
approximate the action of the Koopman operator on $\observable$ by
\begin{align}\label{eq:restriction-on-noninvariant-function}
	\mathcal{K} 
    \observable \approx \mathcal{K}_{\apprx}
    \, \observable =
	v^\top K_{\apprx} \basisfunvec .
\end{align}
The quality of the approximation
in~\eqref{eq:restriction-on-noninvariant-basis} 
directly depends on the invariance properties of the subspace~$\mathbb{V}$. 
If $\mathbb{V}$ is
invariant under the Koopman operator,
equations~\eqref{eq:restriction-on-noninvariant-basis}-\eqref{eq:restriction-on-noninvariant-function}
reduce to~\eqref{eq:EDMD_operator_eigenfunction} and there is no approximation error. 
Otherwise, the projection
in~\eqref{eq:restriction-on-noninvariant-basis} 
leads to information loss and an approximation error as illustrated in \Cref{fig:invariance_illustration}. 
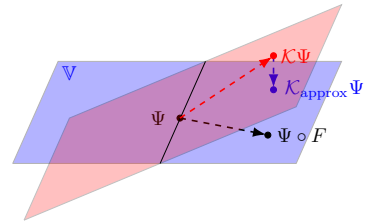
\begin{figure}[htb]
    \centering







\begin{center}
\scalebox{0.75}
{\begin{tikzpicture}[]

    \fill[blue, draw=black, opacity=0.3] plot coordinates{(1,2) (6,2) (6.8,3.8) (1.8,3.8) (1,2)}{};
    \node[blue] at (2.,3.6) {$\mathbb{V} $};

    \draw[thin, fill=black] (3.95,2.8) circle (1.5pt) node[label=left:$\basisfunvec$]{};

    \draw[thick,blue,dashed,-{Latex}] plot coordinates{(6.3,4.0) (6.3,3.02)};
    \draw[thin,blue,fill=blue] (6.3,3.0) circle (1.5pt);
    \node at (7.15,3.0) {\textcolor{blue}{$ \mathcal{K}_{\mathrm{approx}}\basisfunvec$}};
    \node at (6.65,3.5) {\textcolor{blue}{$\mathcal{P}_{\mathbb{V}}$}};

    \fill[red, draw=black, opacity=0.25] plot coordinates{(1.2,1) (6,3) (6.8,4.8) (2.0,2.8) (1.2,1)}{};
    \draw[thin] plot coordinates{(3.6,2) (4.4,3.8)};

    \draw[thin,red,fill=red] (6.3,4.0) circle (1.5pt);
    \node (a) at (6.75,4.0) {\textcolor{red}{$\mathcal{K}\basisfunvec$}};
    \draw[thick,red,dashed,-{Latex}] plot coordinates{(3.95,2.8) (6.27,3.98)};

\end{tikzpicture}}
\end{center}    
    \vspace{-0.5cm}
    \caption{Illustration of the lack of Koopman invariance of the subspace~$\mathbb{V}$.
    }
    \label{fig:invariance_illustration}
\end{figure}

Therefore, it is critical to quantify how close to invariant a
given subspace is, and the impact that this has in the approximation accuracy. The notion that captures this is called \emph{invariance proximity}~\cite{HaseCort26:control:family,MH-JC:24-arxiv}. 
Invariance proximity is formally given by\looseness=-1
\begin{align}\label{def:invariance-proximity}
	\mathcal{I}_{\mathcal{K}} (\mathbb{V})
	:=   \sup_{\observable \in \mathbb{V}, \| \mathcal{K} \observable \| \neq 0} \frac{\| \mathcal{K} \observable -
		\mathcal{K}_{\apprx} \observable \|}{ \| \mathcal{K} \observable \|}.
\end{align}
This measures the worst-case relative error of the
approximation~\eqref{eq:restriction-on-noninvariant-basis} of the
operator's action.
It only depends on the Koopman operator~$\mathcal{K}$ and the subspace $\mathbb{V}$ (since
$\mathcal{K}_{\apprx}$ only depends on $\mathcal{K}$ and $\mathbb{V}$), and
\emph{does not} depend on the choice of basis for~$\mathbb{V}$.  Even though computing invariance proximity might seem like a daunting task based on its definition, it turns out that it has a simple geometric interpretation that 
makes its computation straightforward. To describe this formally, we need to introduce the concept of principal angles. Given the subspace $\mathbb{V}$, consider its image  under the Koopman operator, $\mathcal{K} \mathbb{V}$. The principal angles $0 \leq \theta_1 \leq \cdots \leq \theta_{\dim(\mathcal{K} \mathbb{V})} \leq \frac{\pi}{2}$ between the subspaces $U := \mathbb{V}$ and $V := \mathcal{K} \mathbb{V}$ are 
	defined iteratively by:
	\begin{align*}
		\cos(\theta_i) :=
		&\max_{u \in U} \max_{v \in V} |\innerprod{u}{v} | =: \innerprod{u_i}{v_i}
		\\
		&\text{subject to:} \;
		\innerprod{u}{u_k} = 0, \innerprod{v}{v_k} =0, \; \forall\, k \in [{i-1}]
		\nonumber
		\\
		&\|u\| = 1, \|v\| = 1.
	\end{align*}
Principal angles~\cite{CJ:1875,HH:92} measure the spatial relationship between the two linear subspaces, and  generalize the concept of angle between two vectors. One can show~\cite[Theorem~10]{MH-JC:24-arxiv} that 
\begin{align*}
	\mathcal{I}_{\mathcal{K}} (\mathbb{V})
	=   \sin (\theta_{\dim(\mathcal{K} \mathbb{V})}) .
\end{align*}
This expressions makes clear how the measure $\mathcal{I}_{\mathcal{K}}$ captures the invariance properties of the subspace $\mathbb{V}$. The more aligned the vector spaces $\mathbb{V}$ and $\mathcal{K} \mathbb{V}$ are, i.e., the closer $\mathbb{V}$ is to being invariant, the smaller the maximum principal angle is, and therefore the smaller the invariance proximity is.
\begin{example}
    To illustrate the relationship between  Koopman invariance and prediction accuracy, and the role that invariance proximity plays in quantifying them, we introduce a simple example problem adopted from~\cite{MH-JC:24-arxiv}. Consider the system with state $x = [x_1,x_2]^T$ on $\Omega = [-1,1]^2$,
    \begin{align*}
    	x_1^+ = 0.9 \, x_1,  &&
    	x_2^+ = 0.4 \, \big(\sin(x_2) + x_1^2 \big) + 0.01 \, x_2^2.
    \end{align*}
    Consider the subspaces
    \begin{align*}
        \mathbb{V}_1 &= \operatorname{span}\{1, x_1, x_1^2\},
        \\
        \mathbb{V}_2 &= \operatorname{span}\{1, x_1, x_2, x_1^2\} ,
        \\
        \mathbb{V}_3 &= \operatorname{span} \{1, x_1, x_2, x_1^2, x_2^2\}.
    \end{align*}
    Note that $\mathbb{V}_1 \subset \mathbb{V}_2 \subset \mathbb{V}_3$. 
    Clearly, $\mathbb{V}_1$ is Koopman invariant, since its functions are monomials of the first state variable~$x_1$ and the evolution of $x_1$ abides by a linear dynamics. 
    This is consistent with the fact that $\mathcal{I}_{\mathcal{K}} (\mathbb{V}_1)=0$. The invariance proximity of $\mathbb{V}_2$ is rather small, $\mathcal{I}_{\mathcal{K}} (\mathbb{V}_2)=0.048$, indicating that the worst-case relative function prediction error is $4.8 \%$. On the
    other hand, $\mathcal{I}_{\mathcal{K}} (\mathbb{V}_3)=0.823$, indicating that the worst-case relative function prediction error for $\mathbb{V}_3$ is $82.3 \%$, rendering this model unreliable. 
    This also illustrates the fact that a larger subspace does not lead to a better model if this comes at the cost of negatively impacting the invariance properties.\looseness=-1 \hfill$\blacksquare$
\end{example}

To conclude this section, we briefly mention a consequence of the subspace~$\mathbb{V}$ not being invariant under the Koopman operator for multi-step predictions, where the approximated predictor~$\mathcal{K}_{\apprx}$ is applied more than once. 
Then, nonlinear relationships between the observables contained in~$\mathbb{V}$, e.g., $g_1(x)^2 = g_2(x)$ for $g_1(x) = x_1$ and $g_2(x) = x_1^2$, are not preserved rendering multi-step predictions unreliable due to a missing consistency between the successor state and inherent properties of the subspace~$\mathbb{V}$.
On the one hand, this pronounces the importance of finding (approximately) invariant subspaces using the concept of invariance proximity~\eqref{def:invariance-proximity}. 
On the other hand, reprojection steps may be employed for injective functions $\basisfunvec$ satisfying any of the conditions of \cite[Proposition~4.22]{lee2012} to \textit{restore} consistency between the predicted state and the (nonlinear) manifold $\{ y \in \mathbb{R}^{\nD} \mid \exists\,x \in \mathbb{R}^d: \basisfunvec(x) = y \} \subsetneq \mathbb{R}^n$, $\nD > d$, 
by projecting back to that manifold as illustrated in \Cref{fig:reprojection_illustration} and originally proposed in~\cite{mauroy2019koopman} and further elaborated in~\cite{van23:reprojection,van2025maximum}, where the latter also contains an event-triggering condition.
\begin{figure}[htb]
    \centering
    \begin{center}
\scalebox{0.75}
{\begin{tikzpicture}[]
    \draw[smooth, tension=0.5] plot coordinates{(1,2) (2,2.4) (4,1.78) (6.3,2.8)}{};
    \draw[smooth, tension=0.5] plot coordinates{(6.3,3.8) (4.5,3) (2,3.5) (1.1,3)}{};
    \draw[smooth, tension=1] plot coordinates{(1,2) (1.1,2.5) (1.1,3)}{};
    \draw[smooth, tension=1] plot coordinates{(6.3,2.8) (6.4,3.4) (6.3,3.8)}{};
    \node at (0.55,2.6) {$\basisfunvec(\mathbb{R}^d) $};
    \draw[line width=.08em] (0.5,0.3)--(1,1.3)--(7,1.3)--(6.5,0.3) -- cycle;

    \draw[smooth, tension=0.5] plot coordinates{(2,0.5) (2.5,0.8) (3.5,0.5) (4,0.7)} node[above right=-.05cm] {$F(x)$};
    \draw[thin, fill=black] (2,0.5) circle (1.5pt) node[label=left:$x$]{};
    \draw[thin, fill=black] (4,0.7) circle (1.5pt);

    \draw[smooth, tension=0.5] plot coordinates{(2,2.8) (2.5,2.9) (3.5,2.3) (4,2.5)};
    \node at (4,2.75) {$\basisfunvec(F(x))$};
    \draw[thin, fill=black] (2,2.8) circle (1.5pt); 
    \node at (2,3.15) {$\basisfunvec(x)$};
    \draw[thin, fill=black] (4,2.5) circle (1.5pt);
    \draw [dotted] (2,2.8) -- (2,0.5);
    \draw [dotted] (4,0.7) -- (4,2.5);
    
    \draw[thin,red,fill=red] (4.8,1.5) circle (1.5pt);
    \node (a) at (6.0,1.5) {\textcolor{red}{$K_{\mathrm{approx}}\basisfunvec(x)$}};
    \node (a) at (7,0.5) {$\mathbb{R}^{d}$};

    \draw[red,dashed,-{Latex}] plot coordinates{(4.8,1.5) (4.8,2.36)};
    \draw[thin,red,fill=red] (4.8,2.4) circle (1.5pt);
    \node at (5.9,2.0) {\textcolor{red}{\sffamily\small reprojection}};
    \node at (7.1,2.6) {\textcolor{red}{$\arg\min_{z} \|z-K_{\mathrm{approx}}\basisfunvec(x)\|$}};

\end{tikzpicture}}
\end{center}    
    \vspace{-0.5cm}
    \caption{Illustration of the state space reprojections commonly used to compensate for a lack of invariance. Illustration adapted from \cite{van23:reprojection}.
    }
    \label{fig:reprojection_illustration}
\end{figure}
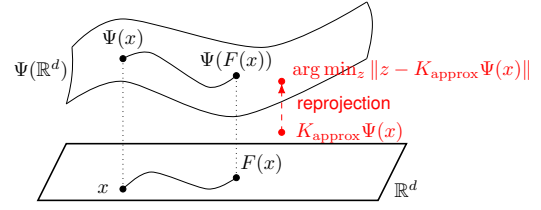

\mysec{Extended dynamic mode decomposition}{\secivpages}
\label{sec:EDMD}

This section introduces 
extended dynamic mode decomposition~(EDMD) as a prototypical data-driven method to approximate the Koopman operator and generator, see also~\cite{Colb24:multiverse} for a comprehensive overview of
sophisticated EDMD variants avoiding spurious eigenvalues and enhancing robustness. 

In \Cref{subsec:EDMD}, we concisely recap the basic EDMD algorithm and demonstrate its applicability by a simulation example showcasing its dependence on the chosen ansatz space~$\mathbb{V}$. 
In \Cref{subsec:EDMD_spectral}, we put our focus on the approximation of eigenfunctions as a key analysis tool for nonlinear dynamical systems as shown in \Cref{subsec:eigenfunctions}, before the approximation error is briefly discussed in \Cref{subsec:EDMD_error}.
In \Cref{subsec:kEDMD}, we present an extension to kernel EDMD (kEDMD). 
Kernel EDMD is attractive due to its close link to machine learning and its potential to leverage recent advances in kernel-based methods.  
Furthermore, kEDMD allows for finite-data bounds on the full approximation error as sketched in \Cref{subsec:EDMD_error_full}.

\mysubsec{EDMD algorithm}{\secivapages}
\label{subsec:EDMD}

In practice, we usually do not know the eigenfunctions of the Koopman operator, such that we resort to using an arbitrary set of basis functions $\{\basisfun_i \mid i \in [n]\}$, usually referred to as dictionary, for estimating a matrix representation of the Koopman operator resembling~\eqref{eq:LS_EV} or, to be more precise, the compression $\mathcal{P}_{\mathbb{V}} \mathcal{K}|_{\mathbb{V}}$ for $\mathbb{V} = \operatorname{span}( \basisfunvec)$. 
To avoid the computation of the function norm, we approximate it using a data set $\{(x^{(j)},y^{(j)}=F(x^{(j)})) \mid j \in [N]\}$. 
Concatenating the basis 
functions evaluated at data points into matrices $[\basismat_{X}]_{i,j}=\basisfun_i(x^{(j)})$ and $[\basismat_{Y}]_{i,j}=\basisfun_i(y^{(j)})$, this allows us to approximate \eqref{eq:LS_EV} via the least squares problem
\begin{align}\label{eq:LS_dictionary}
    \min\nolimits_{K} \|\basismat_{Y}-K \basismat_{X}\|_{\mathrm{Fr}}^2.
\end{align}
The closed-form solution is the EDMD estimate $K_{\basismat}= \basismat_{Y}\basismat_{X}^\top(\basismat_{X} \basismat_{X}^\top)^{\dagger} $ assuming full row rank of the data matrix~$\basismat_X$, which is typically satisfied if the number~$N$ of samples is larger than the dictionary size~$n$. Note that the pseudo-inverse $(\basismat_{X} \basismat_{X}^\top)^{\dagger}$ can be efficiently determined via the singular value decomposition of $\basismat_{X}$, see~\cite{TuRowl14}, which may also be computed using reduced-rank approximations~\cite{rowley17:model}. 
Based on the estimated matrix~$K_{\basismat}$, we define a data-driven estimate of the Koopman operator $\mathcal{K}$ analogously to \eqref{eq:EDMD_operator_eigenfunction} as
\begin{align}\label{eq:EDMD_operator}
     \hat{\mathcal{K}}= \basisfunvec^\top K_{\basismat}^\top T_\basismat^{-1} \langle \basisfunvec, \cdot \rangle, 
\end{align}
where $[T_\basismat]_{ij}\!=\! \langle \basisfun_i, \basisfun_j \rangle$ for notational convenience.
Note that there exists a representation using dual basis functions $\basisfunvec^*\!=\!\overline{(T_{\basismat}^{-1})}\bar{\basisfunvec}$ similar to \eqref{eq:EDMD_operator_eigenfunction}.
This approach for obtaining an EDMD-based linear predictor for a given observable function is outlined in lines 1-6 of \Cref{alg:EDMD}.\looseness-1

In contrast to the matrix representation~$K$ computed in~\eqref{eq:LS_EV} using a dictionary of eigenfunctions, which can be seen as the exact solution to $\Phi_Y=K\Phi_X$ with $\Phi_Y,\Phi_X$ defined similar as $\basismat_Y,\basismat_X$, $\basismat_Y=K_{\basismat}\basismat_X$ does not have an exact solution in general as arbitrary functions $\basisfun_i$, $i \in [n]$, do not generate a Koopman-invariant subspace~$\mathbb{V}$, c.f., \Cref{subsec:Koopman-invariance}. 
Thus, $K_\basismat$ is generally a projection, such that $\hat{\mathcal{K}}$ approximates 
the action of the Koopman operator on an observable~$\observable$ projected on the subspace~$\mathbb{V}$
and projected back onto $\mathbb{V},$ i.e., the compression $ \mathcal{K}_{\apprx}  = \mathcal{P}_{\mathbb{V}} \mathcal{K}|_{\mathbb{V}}$. Here, the projection is taken with respect to the norm associated to the empirical inner product
\begin{align}\label{eq:empirical}
    \innerprod{\observable_1}{\observable_2} = \frac{1}{N} \sum\nolimits_{j=1}^{N} \observable_1(x_j) \overline{\observable_2(x_j)}.
\end{align}
Observe that this projection exhibits the following trade-off.
At the one hand, we need a dictionary that enables the accurate representation of a wide range of observables to mitigate reconstruction errors $\|\basisfunvec^\top T_{\basismat}^{-1} \langle \basisfunvec, \observable \rangle -\observable \|$. On the other hand, the dictionary needs to be invariant under the restricted Koopman operator to reduce  projection errors. 
Thus, a larger dictionary is not necessarily better as it may cause a decrease in reconstruction errors at the cost of higher projection errors \cite{haseli23:generalizing,MH-JC:25-access}, which renders the choice of the dictionary challenging. In practice, the dictionary is often chosen from common parameterizations such as polynomials \cite{haseli23:generalizing}, Fourier features \cite{DeGennaro19:scalable, wormell25:orthogonal}, radial basis functions (RBFs) \cite{WillKevr15:EDMD},
in combination with dictionary learning~\cite{li17:extended,lusch18:deep}. For further discussions on the dictionary see, e.g.,  \cite{WillKevr15:EDMD}.
\begin{algorithm}[t]
	\caption{EDMD predictor and spectral approximation}
	\begin{algorithmic}[1]
		\label{alg:EDMD}
        \STATE{Obtain data set $\{(x^{(j)},y^{(j)}=F(x^{(j)})),j\in[N]\}$}
        \STATE Specify dictionary $\mathbb{V}=\operatorname{span}(\{\basisfun_i\mid i\in[n]\})$
        \STATE Compute $[\basismat_{X}]_{i,j}=\basisfun_i(x^{(j)})$ and $[\basismat_{Y}]_{i,j}=\basisfun_i(y^{(j)})$
        \STATE Set 
        $K_\basismat\gets \basismat_{Y}\basismat_{X}^\top(\basismat_{X} \basismat_{X}^\top)^{\dagger}$
        \STATE Define $c= T^{-1}_\basismat \langle \basisfunvec, \observable \rangle$, $[T_\basismat]_{ij}= \langle \basisfun_i, \basisfun_j \rangle$
        \STATE Approximate predictor $ \basisfunvec^\top K_{\basismat}^\top c$
        \STATE Determine eigenvalues $\eigval_i$ and eigenvectors $\eigvec_i$ of $K_{\basismat}$
        \STATE Define $R\gets [\eigvec_1,\ldots,\eigvec_n] $
        \STATE Approximate eigenfunctions $\hat{\varphi}_{[n]}\gets R^{-1}\basisfunvec$
        \STATE Approximate Koopman modes $[\hat{v}_1,\ldots,\hat{v}_n]^\top\gets R^\top c$
	\end{algorithmic}
\end{algorithm}

\begin{example}
    To illustrate the influence of different dictionaries on the resulting EDMD estimate, we apply it to a system identification problem adopted from \cite{WillKevr15:EDMD}. 
    As system dynamics, this problem considers the Duffing oscillator
    \begin{align}\label{eq:duffing}
        \dot{x}_1 = x_2, && \dot{x}_2= -\delta x_2-x_1(\alpha + \beta x_1^2)
    \end{align}
    with  $\delta \!=\!0.5$, $\alpha\!=\!-1$, and $\beta\!=\!1$. We generate 25 training pairs by simulating a single time step with $T\!=\!0.25$ and initial states on a uniform grid over $[-2,\!2]^2$. Comparative results for trajectory data can be generated via the demo code in the GitHub repository. Based on this training data, we determine the EDMD estimate $K_{\psi}$ using three different dictionaries: \looseness=-1
    \begin{enumerate}
        \item polynomial basis with degree less or equal four
        \item $100$ random Fourier features $[\sin(\omega^\top x), \cos(\omega^\top x)]^\top$ with $\omega\sim\mathcal{N}(0,1)$ \cite{Rahimi07:random}
        \item $100$ RBFs with k-means clustering center \cite{WillKevr15:EDMD}
    \end{enumerate}
    Since the goal is the identification of the system dynamics, the considered observables are $\observable_1(x)=x_1$ and $\observable_2(x)=x_2$. 
    Additionally, we include $x$ in every dictionary to enable the exact reconstruction via $c_1=[1, 0, \ldots, 0]^\top$ and $c_2=[0,1,0,\ldots, 0]^\top$, see~\cite{Mauroy21:koopman}.
    Note that we could alternatively determine suitable vectors $c_k$, $k \in \{1,2\}$, via the least-squares problems\looseness=-1
    \[
        \min\nolimits_{c \in \mathbb{R}^n} \sum\nolimits_{j=1}^N \|\observable_k(x^{(j)})-c^\top \basisfunvec(x^{(j)})\|^2,
    \]
    which can also be solved in closed form~\cite{bevanda21:koopman}. 
    As illustrated in \Cref{fig:edmd_error}, the different dictionaries lead to fundamentally different one-step prediction errors.
    \begin{figure}[htb]
        \centering
        \includegraphics[scale=0.83]{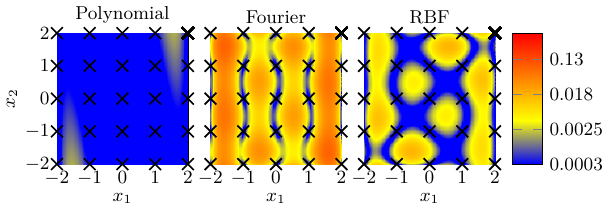}
        \vspace{-0.7cm}
        \caption{One-step prediction errors for different dictionaries. 
        }
        \label{fig:edmd_error}
    \end{figure}
    The polynomial dictionary exhibits the largest errors in the upper right and lower left corners, where the fourth order terms are maximal. 
    In contrast, the errors are maximal in points which are equidistant to the closest training data and decay radially for the RBF dictionary. Despite the low number of points, all dictionaries allow us to effectively capture the fundamental behavior of the Duffing oscillator. 
    We illustrate this by employing the Koopman-operator estimate~\eqref{eq:EDMD_operator} for predicting trajectories analogously to \eqref{eq:multistep_pred} as depicted in \Cref{fig:edmd_traj}, but reproject after every four steps to avoid the excessive accumulation of projection errors~\cite{van23:reprojection}. That means that we lift the predicted state using $\psi$ after every four steps and predict the subsequent four states based on this vector.
    \begin{figure}[htb]
        \centering
        \includegraphics[scale=0.9]{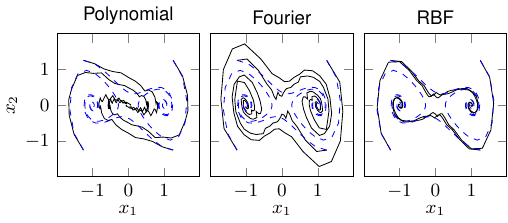}
        \vspace{-0.3cm}
        \caption{Example trajectories (black) for long-term prediction using different Koopman-operator estimates with reprojection every four
        steps. Ground-truth trajectories are depicted as dashed blue lines.}
        \label{fig:edmd_traj}
    \end{figure}    
    
    While the level of accuracy varies between the dictionaries, every model manages to roughly identify the equilibria $[\pm 1, 0]^\top$.
    We want to highlight that the polynomial dictionary yields the lowest average single-step prediction error, but exhibits a large training loss. This effect indicates poor Koopman invariance, which is likely due to the inclusion of monomials of order higher than two, and explains the poor long-term prediction accuracy.
    We further investigate the effect of growing dictionary size. 
    For this, we increase the training data to a uniform grid over $[-2,2]^2$ with 225 samples to ensure $N>n$. We evaluate the performance of each method using the mean squared errors (MSE) for one-step prediction as well as the training loss, i.e., $\|\basismat_Y-K_\basismat \basismat_X\|_{\mathrm{Fr}}^2$. 
    \begin{figure}[htb]
        \centering
        \includegraphics[scale=0.9]{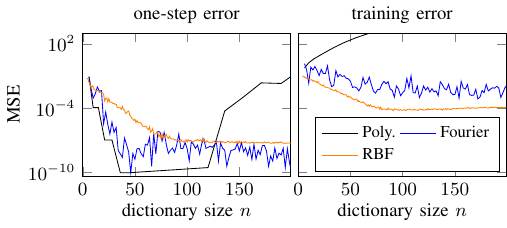}
        \vspace{-0.3cm}
        \caption{Average one-step prediction and training errors depending on the dictionary size. The polynomial dictionary spans subspaces which increasingly lack Koopman invariance, explaining the growing training errors with larger polynomial degrees.}
        \label{fig:edmd_dict}
    \end{figure}    
    As illustrated in \Cref{fig:edmd_dict}, the Fourier and RBF dictionaries yield decreasing errors, which eventually stagnate at a low level. Note that the fluctuations with the Fourier dictionary results from the randomness of its construction. 
    In contrast, the polynomial dictionary suffers from continuously growing training errors despite growing degree of the polynomials, which eventually also causes negative effects on the one-step prediction accuracy. 
    This effect can be explained through an increasing lack of Koopman invariance of the  
    ansatz space~$\mathbb{V}$,  
    which highlights the 
    importance of employing approximately Koopman-invariant dictionaries. Another key point is that, for a given subspace~$\mathbb{V}$, the training loss
    may change with the dictionary selected to generate the subspace~\cite{HaseCort23:Forward-backward_consistency}. Said another way, two different dictionaries that generate the same subspace may have different training losses. This highlights the importance of employing error metrics that do not depend on the specific representation of the subspace, like invariance proximity~\eqref{def:invariance-proximity}.
    \hfill$\blacksquare$
\end{example}

\begin{remark}[Consistency index]
\label{rem:consistency_index}
    For the particular case of the function space w.r.t.\ 
    the empirical measure on the data set, on which EDMD operates, one can obtain a closed-form expression for the invariance proximity defined in Section~\ref{subsec:Koopman-invariance} based on the application of EDMD forward and backward in time~\cite{HaseCort23:Forward-backward_consistency}. This is also known as the \emph{consistency index}. To see this, note that 
    the data set $\{(x^{(j)},y^{(j)}=F(x^{(j)})) \mid j\in[N]\}$ corresponds to the evolution forward in time.
    Likewise, one can interpret the data set $\{(y^{(j)},x^{(j)}) \mid j\in[N]\}$ as the  evolution backward in time. 
    In this way, performing the EDMD approximation yields the backward EDMD matrix $\tilde{K}_{\basismat}= \basismat_{X}\basismat_{Y}^\top(\basismat_{Y} \basismat_{Y}^\top)^{\dagger} $. 
    Note that, if the dictionary $\{\basisfun_i\mid i\in[n]\}$ spans a Koopman-invariant subspace $\mathbb{V}$, then $K_{\basismat} \tilde{K}_{\basismat} = I$. In other words, moving one step forward in time and then one step backward in time should get you back to the same place. 
    Otherwise, if the spanned subspace is not invariant, the forward and backward EDMD matrices will not be the inverse of each other, which motivates the definition of the consistency matrix $M_C = I -  K_{\basismat} \tilde{K}_{\basismat}$. 
    Interestingly, this matrix is similar to a symmetric matrix. In particular, it is diagonalizable with a complete set of eigenvectors, and its spectrum belongs to the interval $[0,1]$. The consistency index is defined as its \textbf{sp}ectral \textbf{rad}ius (abbreviated as $\operatorname{sprad}$). 
  
    Now, if we consider the inner product~\eqref{eq:empirical}
    and the associated induced norm, and compute the invariance proximity of the subspace $\mathbb{V}$, one can show~\cite[Theorem~5.1]{HaseCort23:Forward-backward_consistency} the equality 
    \begin{align*}
        \mathcal{I}_{\mathcal{K}} (\mathbb{V}) = \sqrt{\operatorname{sprad}(M_C)} ,
    \end{align*}
    i.e., invariance proximity corresponds in this case to the square root of the consistency index.
    Notably, this error metric does not depend on the specific dictionary selected to represent the subspace~$\mathbb{V}$.
 \end{remark}

\miniskip
\noindent\textbf{Generator EDMD (gEDMD)}. We follow a similar approach to learn an approximation for the Koopman generator~\eqref{eq:generator:identity} of a continuous-time model \eqref{eq:dynamics:CT}.  
By changing the target value of the least-squares problem to $\dot{\basismat}_X$ with $[\dot{\basismat}_X]_{ij} = \dot{\basisfun}_i(x^{(j)}) := \langle f(x^{(j)}), \nabla \basisfun_i(x^{(j)}) \rangle$ for a sufficiently smooth dictionary, we obtain an estimate $L_{\basismat}= \dot{\basismat}_{X}\basismat_{X}^\top(\basismat_{X} \basismat_{X}^\top)^{\dagger} $ of the matrix representation  of the Koopman generator, see~\cite{klus:nuske:peitz:niemann:clementi:schutte:2020}. 
Note that $\dot{\basisfun}_i(x^{(j)}) = \langle \left.\frac {\mathrm{d}x(t; x^{(j)})}{\mathrm{d}t}\right|_{t=0}, \nabla \basisfun_i(x^{(j)}) \rangle$, such that temporal derivatives along the flow 
are necessary to determine $L_{\basismat}$, which can be obtained, e.g., via temporal differences. Given $L_{\basismat}$, the Koopman generator can be approximated analogously to \eqref{eq:EDMD_operator} via $\hat{\mathcal{L}}= \basisfunvec^\top L_{\basismat}^\top T_\basismat^{-1} \langle \basisfunvec, \cdot \rangle$.
If we want to avoid the usage of explicit numerical differentiation, we can alternatively use the EDMD estimate $K_{\basismat}$ obtained from data of the form $y^{(i)}=x(\Delta T,x^{(i)})$ with small $\Delta T$ to approximate the matrix representation of the Koopman generator via $\hat{L}_\basismat = \tfrac{1}{\Delta T}\log(K_{\basismat})$ \cite{Mauroy16:linear}.

\miniskip
\noindent\textbf{Extensions and further perspectives on (g)EDMD.}
While we present here an introduction of (g)EDMD that comes from a linear regression perspective, these methods are frequently introduced as a finite-data approximation of a Galerkin projection \cite{boyd2013:chebyshev,zhang:zuazua:2023}, and (g)EDMD indeed converge to the Galerkin projection in the infinite-data limit \cite{WillKevr15:EDMD, klus:nuske:peitz:niemann:clementi:schutte:2020}. 
In addition to the infinite-data limit, the behavior of (g)EDMD in the limit of an infinite dictionary has been studied, see e.g., \cite{KordMezi18:convergence}. 
Note that variations of (g)EDMD for stochastic systems also exist \cite{WillKevr15:EDMD, klus:nuske:peitz:niemann:clementi:schutte:2020} with corresponding 
extensions of these theoretical findings. 
Moreover, (g)EDMD is also frequently applied to PDEs via spatial discretization and an interpretation as finite set of coupled ODEs  \cite{Kutz18:applied}. Finally, the requirement of state knowledge is often circumvented through delay embeddings \cite{Arbabi17:ergodic,Kamb20:time-delay}, i.e., replacing states through measurement sequences of observables, an approach inspired by Takens theorem \cite{noakes91:takens}.\looseness=-1

\mysubsec{Eigenfunctions and spectral approximation}{\secivbpages}
\label{subsec:EDMD_spectral}

Given the approximation \eqref{eq:EDMD_operator}, we can straightforwardly approximate the eigenvalues and eigenfunctions of the Koopman operator $\mathcal{K}$. To see this, assume that the EDMD estimate admits an eigendecomposition $K_{\basismat}=R \Lambda^\top R^{-1}$ with matrices $R=[w_1,\ldots, w_n]$ and $\hat{\Lambda}=\mathrm{diag}(\hat{\lambda}_1,\ldots,\hat{\lambda}_n)$ defined using eigenvectors $w_i$ and eigenvalues $\hat{\lambda}_i$, $i\in[n]$, of $K_{\basismat}$. This allows us to reformulate \eqref{eq:EDMD_operator} into
\begin{align*}
    \hat{\mathcal{K}}&= \basisfunvec^\top R^{-\top} \hat{\Lambda} R^\top T_\basismat^{-1} R R^{-1}\langle \basisfunvec, \cdot \rangle\\
    &=\hat{\eigfun}_{[n]}^\top \hat{\Lambda} \hat{T}^{-1} \langle \hat{\eigfun}_{[n]}, \cdot\rangle
\end{align*}
where $\hat{\eigfun}_{[n]}=R^{-1} \basisfunvec$ and $\hat{T}^{-1}=R^\top T_\basismat^{-1} R$. By comparison with \eqref{eq:EDMD_operator_eigenfunction}, where $K$ is a diagonal matrix containing eigenvalues of the Koopman operator, it becomes apparent that the eigenvalues
$\hat{\lambda}_i$ of the EDMD estimate $K_{\basismat}$ can be used as an approximation of the Koopman operator eigenvalues $\lambda_i$, $i\in[n]$. Similarly, we can see that $\hat{\eigfun}_{[n]}$ approximates the vector of eigenfunctions $\eigfun_{[n]}$ in \eqref{eq:EDMD_operator_eigenfunction}. Observe that the inversion of $R$ can be avoided by employing the eigenvectors of $K_{\basismat}^\top$ instead \cite{WillKevr15:EDMD}. Finally, we recover the approximate vector of Koopman modes by comparison, i.e., 
$[\hat{v}_1,\ldots,\hat{v}_n]=\hat{T}^{-1} \langle \hat{\eigfun}_{[n]}, g\rangle = R^\top c$ with $c=T_\basismat^{-1}\langle \basisfunvec, \observable \rangle$. 
Note that the approximation of eigenvalues and eigenfunctions is analogous for the Koopman generator using gEDMD.  
The computation of the approximate eigenvalues, eigenfunctions and Koopman modes using EDMD is summarized in lines 7-10 of \Cref{alg:EDMD}.

\begin{example}
    We reconsider the Duffing oscillator~\eqref{eq:duffing} discretized with time step $T=0.25$ as introduced in \Cref{subsec:EDMD} to illustrate the EDMD-based computation of eigenfunctions and eigenvalues.
    Based on a $100\times 100$ grid to generate training data, we compute eigenvalues and eigenfunctions using EDMD with polynomial (degree~7), Fourier (300~features) and RBF dictionaries (600~centers) following \Cref{alg:EDMD}. The resulting normalized approximate eigenfunctions $\hat{\varphi}_2$ corresponding to the approximate eigenvalue that is second closest to $1$ are illustrated in \Cref{fig:edmd_eigen}. 
    While 
    these approximate eigenfunctions $\hat{\varphi}_2$ are 
    real-valued,  
    most approximate eigenfunctions $\hat{\varphi}_i$ corresponding to other approximate eigenvalues $\hat{\lambda}_i$ are complex-valued. As discussed in \Cref{subsec:eigenfunctions}, eigenfunctions corresponding to the eigenvalue $\lambda=1$ provide a special insight for discrete-time systems as the corresponding modes remain constant. In particular, the space of true eigenfunctions corresponding to the eigenvalue $\lambda=1$ is 3 dimensional: constant functions and two 1-dimensional function subspaces corresponding to indicator functions on each basin of attraction. We empirically find that
the first approximate eigenfunction corresponds 
to the eigenvalue $\lambda=1$ representing 
a global constant offset, i.e., a constant function. The second approximate eigenfunction estimates one of the indicator eigenfunctions in our example, which is a discontinuous function \cite{mauroy13:isostables}.
Even though none of the chosen dictionaries allows to represent discontinuous function, the sign of the continuous approximations $\hat{\varphi}_2$ can still be used to approximate the regions of attraction \cite{WillKevr15:EDMD}.
As highlighted in \Cref{fig:edmd_eigen}, the estimates obtained this way are highly similar between dictionaries despite clearly visible differences in the approximate eigenfunctions themselves.
\begin{figure}[htb]
    \centering
    \includegraphics[scale=0.9]{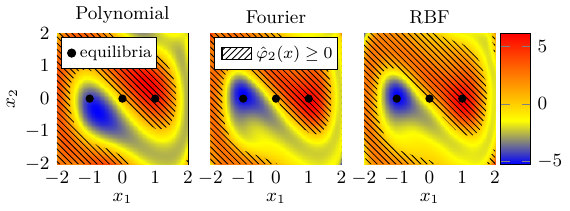}
    \vspace{-0.8cm}
    \caption{Approximate eigenfunctions $\hat{\varphi}_2$ for the approximate eigenvalue that is second closest to $1$ obtained via different dictionaries.
    The sign of this function serves as indicator for
    the region of attraction of
    the stable equilibria as illustrated by the patterned area for the one
    at $x=[1,0]^\top$.}
    \label{fig:edmd_eigen}
\end{figure}

It should be noted, that we could only obtain these results reliably using a large number of data ($10,000$ samples) and the large dictionaries, which comes at the cost of computational complexity. This can be explained by the fact that only a comparatively small fraction of the approximate eigenvalues $\hat{\lambda}_i$ are meaningful. For example, we can estimate $600$ eigenvalues in our example with the RBF dictionary, but only $75$ are larger than $0.01$.
More generally, the non-smoothness of these eigenfunctions is an issue for the EDMD method. Note that alternative techniques to EDMD do not necessarily suffer from non-smoothness in the same way, see e.g., \cite{mezi05, mohr2014construction, Mezi22:approximations}.
\hfill$\blacksquare$
\end{example}

In practice, the focus lies in the accurate identification of the leading eigenvalues and eigenfunctions, i.e., the most slowly decaying tuples for discrete-time systems \cite{WillKevr15:EDMD}. Recall that these leading eigenvalues and eigenfunctions are sufficient to precisely represent the long-term dynamics since the modes corresponding to
eigenvalues with real part close to zero will disappear after very few time steps. 
Furthermore, the full spectrum of systems  with fixed points, but more generally any quasi-periodic attractor, is of the so-called lattice type and depends only on the spectral properties  of the linearization around the attractor \cite{mezic2020spectrum}. These insights
can be directly used to construct a reduced-order model that captures the core evolution of the system via a subset of the approximated eigenfunctions \cite{ChenTu12:DMD} as an alternative to the SVD-based model reduction discussed in \Cref{subsec:EDMD}. Moreover, it can be used for a dedicated separation of different time scales in the dynamics \cite{froyland:14computational}.

\begin{remark}\longthmtitle{Spectral pollution}
    Since (g)EDMD can be viewed as a Galerkin method, it also suffers from the existence of spurious eigenvalues that is common with this class of techniques \cite{lewin10:spectral}. A spurious eigenvalue is a value that is not associated with the actual spectrum of the Koopman operator, but a sequence of eigenvalue estimators with increasing dictionary size converges to this value. This spectral pollution is closely tied to the choice of dictionary, its lack of invariance, and the resulting projection error which depends crucially on the measure $\mu$.
    Thereby, it introduces an additional dependency on the data distribution \cite{Mezi22:approximations} implying, e.g., that single trajectory data only reveals on-attractor eigenvalues. As a consequence, dealing with spectral pollution has been a significant effort in research leading to approaches for the detection of spurious eigenvalues \cite{levitin04:spectral}, error control \cite{colbrook19:compute, colbrook24:rigorous}, and variants of (g)EDMD without spectral pollution \cite{Drmac18:data, mauroy24:analytic}. 
    A full categorization of such issues, alongside an algorithm that guarantees spectral convergence for a broad class of autonomous dynamical systems can be found in \cite{ColbMezi2024:limits}.
    For a detailed introduction to more advanced methods for the numerical approximation of the operator spectra, we refer to \cite{Colbrook26:infinite}.\looseness=-1 
\end{remark}

\begin{remark}\longthmtitle{Discovering coherent structures}
    This section focuses on Koopman operator eigenvalue and eigenfunctions approximations, but modes estimated using EDMD are commonly employed for discovering coherent structures in the context of partial differential equations, e.g., vortices and eddies in flow physics \cite{RowlMezi09, ChenTu12:DMD}. For an exemplary discussion of EDMD-based Koopman modes in fluid flow problems we refer to \cite{mezic13:analysis}.
\end{remark}

\mysubsec{Kernel EDMD}{\secivcpages}
\label{subsec:kEDMD}

A special case of EDMD employs the canonical features of kernels as dictionary, such that the dictionary size~$n$ corresponds to the size~$N$ of the data set~\cite{WillRowl16:kEDMD}. 
A kernel $k:\mathbb{R}^d\times\mathbb{R}^d\rightarrow \mathbb{R}$ is 
a symmetric function of two arguments, i.e., $k(x,x')=k(x',x)$. 
In the context of machine learning and interpolation theory, positive definiteness of the kernel~$k$ is additionally assumed, which is equivalent to the positive semi-definiteness of its Gram matrices $k_X$ with $[k_X]_{ij}=k(x^{(i)},x^{(j)})$ for all sets of 
points $\{x^{(i)} \in \mathbb{R}^d \mid i \in [N] \}$. 
A key feature of kernels is that they uniquely correspond to
Hilbert spaces~$\mathbb{H}$ with beneficial properties -- so called reproducing kernel Hilbert spaces (RKHSs). A RKHS is the completion of the pre-Hilbert space of functions
\begin{align*}
    \mathbb{H}_0 = \Big\{ \psi \in L^2_\mu(\Omega) \mid  \psi=\sum\nolimits
    _{i=1}^{N} w_i k(\cdot,x^{(i)}), x^{(i)}\in\Omega \Big\}
\end{align*}

with the norm $\|\psi\|_{\mathbb{H}}=\sqrt{\sum_{i,j=1}^{N} w_iw_j k(x^{(i)},x^{(j)})}$, which is equipped with an inner product that satisfies the reproducing property 
$\langle \observable, k(\cdot,x) \rangle_{\mathbb{H}} = \observable(x)$, see \cite{scholkopf02:learning, berlinet11:reproducing, manton15:a} for a thorough introduction to kernels and RKHS theory.

Given a strictly positive-definite kernel, we define our dictionary via the kernel's canonical features as $\basisfun_i=k(x^{(i)},\cdot)$, $i\in[N]$, for the set $\mathcal{X} := \{x^{(i)} \in \Omega \mid i \in [N]\}$ of pairwise distinct data points. Then, the term \textit{strict} ensures positive definiteness of the matrix~$k_X$.
We exploit that $\basismat_X^{\dagger} = \basismat_X^\top (\basismat_X\basismat_X^\top)^{\dagger}$ holds~\cite{WillKevr15:EDMD} together with $\basismat_X^{\dagger} = k_X^{\dagger}$ implying $\basismat_X^{\dagger} = k_X^{-1}$ due to pairwise distinct data points and strict 
positive definiteness of the kernel~$k$. 
Similarly, we have $\basismat_Y = k_{XY}$ for $[k_{XY}]_{ij}=k(x^{(i)},y^{(j)})$. 
Hence, we obtain the kernel EDMD (kEDMD) estimate $K_k=k_{XY}k_X^{-1}$ as approximate matrix representation for the Koopman operator~$\mathcal{K}$. 
By approximating~$\mathcal{K}$  
on the RKHS $\mathbb{H}$
, we further define $[T_\basismat]_{ij}= \langle \basisfun_i, \basisfun_j \rangle_{\mathbb{H}} = k(x^{(i)},x^{(j)})$ with the equality following from the reproducing property of the RKHS \cite{scholkopf02:learning}. Thus, we can approximate the Koopman operator via\looseness=-1
\begin{align*}
    \hat{\mathcal{K}}_k = k_{\mathcal{X}}^\top k_X^{-1}k_{YX} k_X^{-1} \langle k_{\mathcal{X}}, \cdot \rangle_{\mathbb{H}}
\end{align*}
with column vector function $[k_{\mathcal{X}}]_i=k(x^{(i)},\cdot)$. Note that the usage of the canonic kernel features allows the straightforward computation of the scalar product for arbitrary observables $\observable\in\mathbb{H}$ as $\langle k_X, \observable \rangle_{\mathbb{H}}=\psi_X$ due to the reproducing property of the RKHS.  
Note that a similar approach for gEDMD yields 
a kernel-based estimator of the Koopman generator~\cite{klus20:kernel}.\looseness=-1
\begin{table}[h!tb]
    \begin{center}
    \caption{Expressions for the commonly used Mat\'ern class, squared exponential (SE), rational quadratic (RQ) and Wendland kernels. The lengthscale parameter is denoted by $l$.}
    \footnotesize
    \begin{tabular}{p{1.7cm} p{1.2cm}|p{4.3cm} }
    \toprule
     kernel & parameter & expression \\
     \midrule \midrule
     \multirow{3}{*}{Mat\'ern} & $\nu=1/2$ & $\exp(-\tfrac{|r|}{l})$\\[2pt]
      & $\nu=3/2$ & $(1 + \tfrac{\sqrt{3}|r|}{l}) \exp(-\tfrac{\sqrt{3}|r|}{l})$\\[2pt] 
      & $\nu=5/2$ & $(1 + \tfrac{\sqrt{5}|r|}{l} + \tfrac{5r^2}{3l^2}) \exp(-\tfrac{\sqrt{5}|r|}{l})$ \\[2pt] \midrule
    SE & & $\exp(-\tfrac{r^{2}}{2l^2})$\\[2pt] \midrule
    RQ & $\alpha$ & $(\tfrac{1+r^2}{2\alpha l^2})^{-\alpha}$\\[2pt] \midrule
    \multirow{3}{*}{\shortstack[l]{
    $\phantom{m}$\\ Wendland\\ $\nu \!=\! \left\lfloor \!\tfrac{d}{2}\! \right\rfloor\!\!+\!2q\!+\!1$}} & $q=0$ & $\max(1-|r|,0)^{\nu}$
    \\[2pt]
    & $q=1$ & $\max(1-|r|,0)^{\nu
    } (\nu
    r+1)$ \\[2pt]
    & $q=2$ & $\tfrac{1}{3} \cdot \max(1-|r|,0)^\nu 
    (\nu 
    r^2+3 \nu 
    r+3)$
    \\[2pt]
    \bottomrule \label{tab:common_kernels}
    \end{tabular}
    \end{center}
    \vspace{-0.6cm}
\end{table}

Due to the strong connection to EDMD, the choice of the kernel with its induced RKHS essentially takes over the role of the dictionary selection in kEDMD. While the impact of this choice might seem less obvious, the properties of kernels are generally well-studied. For example, it is well-known that Mat\'ern class \cite{matern60:spatial}, squared exponential (SE) (also known as radial basis function), and Wendland kernels \cite{Wend04} are universal approximators \cite{micchelli06:universal}. 
Exemplary expressions and illustrations of these kernels are provided in~\Cref{tab:common_kernels} and \Cref{fig:kernel_visual}. 
Additionally, the RKHS of these kernels is well understood: the RKHS of SE kernels is a subset of the analytic functions \cite{vandervaart11:information}, while the RKHS of Mat\'ern and Wendland kernels is isomorphic to Sobolev spaces with a smoothness parameter specifying the degree of weak differentiability \cite{Wend04,kanagawa18:gaussian}. 
This study of RKHSs has recently also been extended to specifically analyze their impact on Koopman operator estimates~\cite{kohne2025error,gonzalez24:the,HertPhil25}.
For example, the close relation (unitary equivalence) between the RKHS of SE kernels and Fock spaces indicates a lack of invariance of this RKHS under Koopman operators, except for linear dynamics \cite{ColbMezi2024:limits}.\looseness=-1
\begin{figure}[htb]
    \centering
    \includegraphics[scale=0.9]{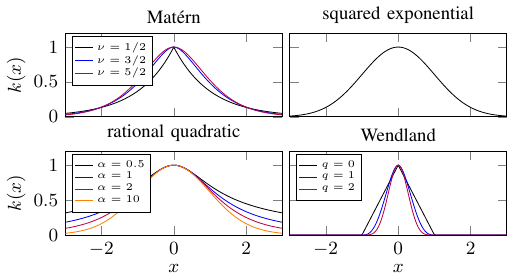}
    \vspace{-0.3cm}
    \caption{Illustration of different kernel functions $k(\cdot,0)$ with different smoothness parameters and lengthscale $l=1$.}
    \label{fig:kernel_visual}
\end{figure}

\begin{remark}\longthmtitle{Computational complexity}
    Since kEDMD requires the inversion of a $(N\times N)$ matrix, its naïve computational complexity grows cubically. To mitigate this issue and ensure scalability to large data set sizes, one can resort to scalable approximations commonly used for kernel methods such as the Nystr\"om method~\cite{williams00:using}, random features~\cite{Rahimi07:random}, or sketching~\cite{woodruff14:sketching}. 
    Advanced techniques relying on these foundations and additionally exploiting GPUs can scale kernel methods to millions of data points as shown in~\cite{meanti20:kernel}.\looseness=-1
\end{remark}

\mysubsec{EDMD error analysis: projection and estimation}{\secivdpages}
\label{subsec:EDMD_error}

In Section~\ref{subsec:Koopman-invariance}, we discussed 
Koopman invariance of the finite-dimensional subspace
\[
    \mathbb{V} := \operatorname{span}\{ \basisfun_i \mid i \in [n] \}
\]
w.r.t.\ the action of the Koopman operator. In particular, we introduced the concept of approximate Koopman invariance, essentially by a normalized weighting of the projection error, cf.\ the definition of invariance proximity $\mathcal{I}_{\mathcal{K}}(\mathbb{V})$ in~\eqref{def:invariance-proximity}. 
When using data, the second source of error is estimation. While only a limited number of data samples is required for a Koopman-invariant subspace, the approximation given by the inner product~\eqref{eq:empirical}, which corresponds to the one employed by EDMD, only converges to the inner product $\langle \basisfun_i(x), \basisfun_j(x) \rangle$ in the infinite-data limit, see~\cite{KordMezi18:convergence}. 
The paper~\cite{ColbMezi2024:limits} shows that, e.g., for $\mathcal{F} = L_\mu^2 (\Omega)$ and the deterministic dynamics~\eqref{eq:dynamics:DT}, more than one limit is necessary ($n \rightarrow \infty$ and infinite data) based on the concept of the solvability complexity index. 
Hence, for EDMD, one has to analyze the estimation and the projection error, which was --~to the best of our knowledge~-- first done in~\cite{Mezi22:approximations} (estimation) and~\cite{zhang:zuazua:2023} (projection) leading overall to probabilistic $L^2$-bounds (based on the Lebesgue measure):
For a desired accuracy $\varepsilon > 0$ and a probabilistic tolerance $\delta > 0$, we have the probabilistic bound\looseness=-1
\begin{equation*}
    \mathbb{P}( \| \hat{K} - \mathcal{P}_{\mathbb{V}}\mathcal{K}|_{\mathbb{V}} \|_{\mathrm{Fr}}^2 \leq \varepsilon) \geq 1 - \delta
\end{equation*}
on the estimation error assuming that the number~$N$ of i.i.d.\ drawn samples is of order $\mathcal{O}(n^2/(\delta \varepsilon^2))$ as shown in~\cite{nuske2023finite} for the Koopman generator and operator. 
Based on these findings, various extensions towards stochastic dynamics and systems with input for i.i.d.\ and ergodic sampling were proposed~\cite{nuske2023finite}, see also~\cite{philipp2026variance} for a variance representation formula and an overview on recent results on the estimation error for deterministic and stochastic dynamics. 
In addition, the projection error has to be analyzed. To this end, it is important to note that the Koopman operator is of transport type as easily inferrable from identity~\eqref{eq:generator:identity} for the Koopman generator, see also \cite[Subsection~2.1]{StraWort26}. 
Hence, using tools originally developed for the numerical analysis of approximation schemes tailored to \textit{hyperbolic} partial differential equations, $L^2$-bounds on the projection error can be derived; also for systems with inputs~\cite{schaller2023towards}.
If, however, Koopman invariance holds (see \Cref{subsec:Koopman-invariance}), there is no projection error resulting in $L^\infty$, i.e., pointwise error bounds. To this end, dictionary learning based on, e.g., invertible neural networks can be leveraged as proposed in~\cite{bevanda22:diffeomorphically}. 
Furthermore, techniques developed for optimization on manifolds can be employed, see, e.g., \cite{schurig2025shaping} for an approach to shape the dictionary on the Grassmanian. This might, e.g., be combined with error bounds tailored to particular observables, which allows to tighten the bounds on approximants of the Koopman operator, which are designed for arbitrary observables contained in the underlying function space.

\mysubsec{Kernel EDMD: full approximation error bounds}{\secivepages}
\label{subsec:EDMD_error_full}

Kernel EDMD 
as introduced in \Cref{subsec:kEDMD} provides a remedy to rigorously derive $L^\infty$-bounds, i.e., uniform pointwise bounds on the approximation error. 
We demonstrate this using 
Wendland~\cite{Wend04} or Matérn kernels, see \Cref{tab:common_kernels}. \looseness=-1

Following the approach proposed in~\cite{kohne2025error}, we provide some insights into the derivation of the respective uniform error bounds using the notation of \Cref{sec:KOT}.
Since we do not assume invariance of the set~$\Omega \subset \mathbb{R}^d$ w.r.t.\ the dynamics~$F$, we define the respective natives spaces~$\mathcal{N}(\Xi)$ and $\mathcal{N}(\Omega)$ corresponding to the RKHSs generated by Wendland kernels on~$\Xi$ and~$\Omega$, see~\cite{HertPhil25} for similar results using Matérn kernels as well as stochastic dynamics. 
Then, we split the approximation error into two parts, where $\| \cdot \| = \| \cdot \|_{\mathcal{N}(\Xi) \rightarrow L^\infty(\Omega)}$ is used for ease of notation and show the inequality
\begin{align*}
    \| \widehat{\mathcal{K}} - \mathcal{K} \| & =  \| P_{\mathcal{X}} \mathcal{K} ( P_{\mathcal{Y}} - I) + (P_{\mathcal{X}} - I) \mathcal{K} \| \\
    & \leq  \| P_{\mathcal{X}} \mathcal{K} ( P_{\mathcal{Y}} - I) \| + \| (P_{\mathcal{X}} - I) \mathcal{K} \| \\
    & \leq  C_1 h_{\mathcal{Y}}^{k+1/2} + C_2 h_{\mathcal{X}}^{k+1/2}
\end{align*}
with fill distances~$h_{\mathcal{X}}$, $h_{\mathcal{Y}}$ on the sets~$\mathcal{X} = \{x^{(i)} \mid i \in [N]\} \subset \Omega$ and $\mathcal{Y} = \{y^{(i)} \mid i \in [N]\} \subset \Xi$, respectively.\footnote{The fill distance $h_{\mathcal{X}}$ on the set $\mathcal{X} = \{x^{(i)} \mid i \in [N]\} \subset \Omega$ is defined as $\sup_{x \in\Omega} d(x,\mathcal{X})$, where $d(x,\mathcal{X}) := \min\{ \| x - x^{(i)} \| \mid i \in [N]\}$.} The operators~$P_\mathcal{X}$ and~$P_\mathcal{Y}$ project on the spaces 
spanned by the canonical features $\{ \basisfun_i = k(x^{(i)},\cdot) \mid i \in [N] \}$ and $\{ \basisfun_i = k(y^{(i)},\cdot) \mid i \in [N] \}$, respectively, meaning that $P_\mathcal{X}$ corresponds to~$P_{\mathbb{V}}$.
For the last inequality, we briefly explain the estimates on the two terms. 
The first follows directly from approximation-theoretic results for the orthogonal projection~$P_{\mathcal{Y}}$ from the native space~$\mathcal{N}(\Xi)$ on~$L^\infty(\Omega)$ as well as $\| P_{\mathcal{X}} \|_{L^\infty(\Omega) \rightarrow L^\infty(\Omega)} = 1$ (projection) and $\| \mathcal{K} \|_{L^\infty(\Omega) \rightarrow L^\infty(\Omega)} = 1$.
The second, however, requires Koopman invariance of the native space, i.e., $\mathcal{K}(\mathcal{N}(Y)) \subseteq \mathcal{N}(X) = H^{\sigma_{d,k}}(X)$ with $\sigma_{d,k} := (d+1)/2 + k$, where $k$ denotes the smoothness parameter of the Wendland kernel and $F \in \mathcal{C}_b^m$, $m > \sigma$, is assumed for the dynamics~\eqref{eq:dynamics:DT}.
Koopman invariance us to infer boundedness of the Koopman operator~$\mathcal{K}$, i.e., $\| \mathcal{K} \|_{\mathcal{N}(Y) \rightarrow \mathcal{N}(X)} < \infty$, see \cite[Theorem~C.1]{kohne2025error} for an estimate of the bound on Sobolev spaces for integer order (for non-integer orders, interpolation theory is additionally required), see also~\cite{ColbMaur26} for characterizations of Koopman invariance on various function spaces for deterministic and stochastic dynamics.
Then, again, approximation theory is used resulting in the fill distance~$h_{\mathcal{X}}$ to the $(k+\frac 12)$-power.

We point out, that supposing sufficiently smooth dynamics~$F$ in~\eqref{eq:dynamics:DT}, the \textit{proportional} bound
\begin{equation*}
    |(\widehat{\mathcal{K}}f)(x) - (\mathcal{K}f)(x)| \leq C h^{k - 1/2}_{\mathcal{X}} \operatorname{dist}(x, \mathcal{X}) \| f \|_{\mathcal{N}_{\basisfun_{d,k}}(\Omega)}     
\end{equation*}
can be established, see~\cite{BoldPhil25}, where also extensions for control systems were derived (including the integration of a regularization term in the RKHS regression problem).\footnote{$\| \cdot \|_{\mathcal{N}_{\basisfun_{d,k}}}(\Omega)$ denotes the norm of the RKHS generated by Wendland kernels with smoothness parameter~$k$ on~$\mathbb{R}^{d}$.} 
The terminology \textit{proportional} stems from the fact, that the error (bound) decays to zero if the state approaches the origin, which is key for, e.g., set-point stabilization as shown later in \Cref{sec:controller}.
For control systems, \cite{schmitz2025excitation} provides (optimality) conditions analysing the excitation of the inputs using the concept of subspace angles to further tighten the control-related bounds and alleviate the data requirements. Note that these techniques are applicable to both (generator) EDMD and kernel EDMD.\looseness=-1

We point out that kernel EDMD may also be analyzed using so-called Mercer features as proposed in~\cite{PhilScha24:kernel}, see also~\cite{PhilSchaWort2025} and~\cite{bevanda2026nonparametric} for an extension to control, where in the latter a Koopman-inspired kernel regression problem is set up to allow for an extremely efficient incorporation of inputs for particular observables of interest. 
Further, using a symmetrization approach, Koopman invariant features over trajectory data can be constructed to ensure a bounded generalization gap to unseen data \cite{bevanda23:koopman}. This approach has been proven to increase the convergence rate w.r.t.\ learning linear predictors~\cite{bevanda25:koopman}. 
Further error bounds were, e.g., proposed in~\cite{llamazares2024data}. An alternative approach to approximate the Koopman operator is proposed in~\cite{yadav2025approximation} using Bernstein polynomials, which is, however, not (directly) related to (kernel) EDMD.\looseness=-1

\mysec{Extension to control systems}{\secvpages}
\label{sec:control}

We consider the discrete-time dynamical system
\begin{equation}\label{eq:controlled-dynamics:DT}
    x^+ = F(x,u),
\end{equation}
where the successor state~$x^+$ is given by the image of the continuous map $F: \Omega \times \mathbb{U} \rightarrow  \Xi$, $\mathbb{U} \subset \mathbb{R}^m$, applied to the current state~$x$ and the input~$u$. Note that the dependence of $F$ on $u$ is arbitrary, in particular, not necessarily affine.

\mysubsec{On exact LTI embeddings}{\secvapages}
\label{subsec:LTI}

The nonlinear system \eqref{eq:controlled-dynamics:DT} admits an exact Linear Time-Invariant (LTI) embedding if there exists a set of linearly independent functions $\basisfun_1, \dots, \basisfun_{\nD}: \Omega \rightarrow \mathbb{K}$  such that
\begin{equation}
\label{eqn:lifted-state}
\basisfunvec(x^+) = A\basisfunvec(x) + Bu, \quad x = C\basisfunvec(x),
\end{equation} 
holds for all $u \in \mathbb{U}$, where $A \in \mathbb{R}^{\nD \times \nD}, B \in \mathbb{R}^{\nD \times m}, C \in \mathbb{R}^{d \times \nD}$ are constant matrices. This means that the value of the functions $\basisfunvec$ evolves linearly along system trajectories for any input~$u$. 
For convenience, we also require that the original system state~$x$ belongs to the subspace~$\mathbb{V}$ generated by the functions $\basisfun_1, \dots, \basisfun_{\nD}$, i.e., $x = C \basisfunvec(x)$ with a constant matrix $C \in \mathbb{R}^{d \times \nD}$. 
In this case, we define $z := \basisfunvec(x)$ as a lifted state, and the nonlinear system \eqref{eq:controlled-dynamics:DT} admits the linear model representation\looseness=-1
\begin{equation}\label{eqn:Koopman-linear}
z^+ = Az + Bu, \quad x = Cz.
\end{equation}
Note that the Koopman linear model \eqref{eq:EDMD_operator_eigenfunction} is a special case of the LTI embedding~\eqref{eqn:lifted-state} (when the state belongs to the subspace $\mathbb{V}$) with zero input. 
Even when an exact LTI embedding does not exist, the approximate Koopman linear representation \eqref{eqn:Koopman-linear} has been widely used in predictive control applications to model nonlinear dynamics \eqref{eq:controlled-dynamics:DT}, e.g., \cite{KordMezi18,StraWort26,shi26:koopman}. \looseness=-1

\begin{example}
Consider the following two-dimensional nonlinear system~\cite{brunton2016koopman,ShanHasi26:LTI}: 
\begin{equation*}
\begin{aligned}
x_1^+  \!=\! x_2^2 + x_1 + u, \qquad 
x_2^+  \!=\! 0.9 x_2.
\end{aligned}
\end{equation*}
The choice of functions  $\basisfunvec(x) = (x_1, x_2, x_2^2)$ reveals that the system admits an exact LTI embedding. Indeed,
\begin{align*}
\basisfunvec(x^+) &= \begin{bmatrix}
        1 & 0 & 1 \\ 0 &  0.9 & 0 \\ 0 & 0 & 0.81
    \end{bmatrix}\basisfunvec(x) + \begin{bmatrix}
        1 \\ 0 \\ 0
    \end{bmatrix} u,
    \\
     x &= \begin{bmatrix}
        1 & 0 & 0 \\ 0 & 1 & 0
    \end{bmatrix} \basisfunvec(x). 
\end{align*}
\hfill$\blacksquare$
\end{example}

The availability of an exact LTI embedding provides a key strategic advantage in dealing with the complexity of nonlinear systems because it allows to bring the full range of well-established tools from linear control systems into 
analysis and control design.  This raises the fundamental question of which classes of nonlinear systems admit an exact LTI embedding. 
The recent paper~\cite{XS-JC-YZ:25-csl} explores the extent to which such systems are amenable to powerful data-driven linear control techniques, like Willems' fundamental lemma.\looseness=-1

A precise and complete characterization of the  class of systems that admits an exact LTI embedding is as follows, cf.~\cite[Theorem~2]{ShanHasi26:LTI}: 
the nonlinear system \eqref{eq:controlled-dynamics:DT} admits an exact LTI  embedding if and only if\looseness=-1
\begin{enumerate}
    \item there exists an invertible matrix $T \!\in\! \mathbb{R}^{d ~\!\!\!\times~\!\!\! d}$ such that, with the coordinate transformation $\tilde{x} \!=\! T x \!:=\! (\tilde{x}_1, \tilde{x}_2)$, the dynamics take the \emph{control-affine preserved (CAP)} form:\looseness=-1
    \begin{subequations}
    \begin{equation}\label{eq:CAP+Koopman-invariant}
    \begin{bmatrix}
       \tilde{x}_1^{+} \\ \tilde{x}_2^{+}
    \end{bmatrix} = \begin{bmatrix}
        h_1(\tilde{x}_2) + A_1 \tilde{x}_1 \\h_2(\tilde{x}_2)
    \end{bmatrix} + \begin{bmatrix}
       D \\  0
    \end{bmatrix} u ,
    \end{equation}
    for $\tilde{x} \in \tilde{\Omega}$ and $u\in \mathbb{U}$,  
    where $\tilde{\Omega} := \{T x \in \mathbb{R}^d \mid x \in \Omega \}$, and $h_1$, $h_2$, $A_1$, and $D$ are functions and matrices of compatible 
    dimensions, and
    \item there exist $\basisfunvecbar \!=\! (\bar{\basisfun}_1,\dots,\bar{\basisfun}_{\bar{\nD}})\!:\! \mathrm{proj}_{\tilde{x}_2}(\tilde{\Omega}) \!\to\! \mathbb{K}^{\bar{\nD}}$ and 
     constant matrices \(A_2, \bar{C}\) of compatible dimensions such that\looseness=-1
        \begin{equation}\label{eq:Koopman-closed-autonomous}
        \!\basisfunvecbar(h_2(\tilde{x}_2)) \!=\! A_2 \basisfunvecbar(\tilde{x}_2), \; (\tilde{x}_2, h_1(\tilde{x}_2)) \!=\!\bar{C} \basisfunvecbar(\tilde{x}_2).\!
        \end{equation}
        \end{subequations}
\end{enumerate}
The condition~\eqref{eq:Koopman-closed-autonomous} yields a Koopman-invariant subspace under the autonomous system $\tilde{x}_2^+ = h_2(\tilde{x}_2)$ that contains its own state $\tilde{x}_2$ and all nonlinear terms $h_1(\tilde{x}_2)$. The condition~\eqref{eq:CAP+Koopman-invariant} presents all required structural properties for a nonlinear system to admit an LTI embedding. The ``only if" implication of this statement is easy to establish. In fact, if $\tilde{x}_1 \in  \mathbb{R}^{d_1}
, \tilde{x}_2 \in \mathbb{R}^{d_2}
$, then 
the choice of functions $\basisfunvec(x):= (\tilde{x}_1, \basisfunvecbar(\tilde{x}_2))$ satisfies~\eqref{eqn:lifted-state} with matrices 
\begin{align*}
A &\!=\! \begin{bmatrix}
    A_1 & [{0} \ I_{d_1}
    ] \bar{C}  \\ {0} & A_2
\end{bmatrix},
~~\!\!
 B \!=\!  \begin{bmatrix}
    D \\ {0}
\end{bmatrix}, ~~\!\!
C \!=\! T^{-1}\!
\begin{bmatrix}
    I_{d_1}
    & {0} \\ {0} & [I_{d_2}
    \ {0}] \bar{C}
\end{bmatrix}\!.
\end{align*}
Instead, the ``if" implication is harder, cf.~\cite{ShanHasi26:LTI}.

\begin{example}
 This example, taken from~\cite{ShanHasi26:LTI}, illustrates the characterization of nonlinear systems that admit exact LTI embeddings. 
Consider the nonlinear system 
\[
    \!\begin{bmatrix}
        x_1^+ \\ x_2^+ \\ x_3^+
    \end{bmatrix}\! =\! 
    \begin{bmatrix}
        (x_2\!+\!x_3)^2 \!+\! (x_1\!+\!x_2)\!+\!u\\
        (x_2\!+\!x_3)^2\!\cdot\!(x_2\!+\!x_3-1)\!+\!x_1 \!-\!2u\\
        (x_2\!+\!x_3)^2\!\cdot\!(1\!-\!x_2\!-\!x_3)\!-\!x_1\!+\!0.5x_2\!+\!0.5x_3\! +\! 2u
    \end{bmatrix}\!\!,
\]
where $x := (x_1, x_2, x_3) \in \mathbb{R}^3$. This system  admits an exact LTI embedding by choosing 
\[
\basisfunvec(x) \!=\! (x_1, -\!2x_1\!-\!x_2, x_2\!+\!x_3, (x_2\!+\!x_3)^2, (x_2\!+\!x_3)^3) .
\]
In fact, if we set $z = \basisfunvec(x)$, we can verify that 
\begin{equation} \label{eq:Koopman-model-example}
\begin{aligned}
z^+ =& \begin{bmatrix}
    -1 & -1 & 0 & 1 & 0 \\ 1 & 2 & 0 & -1 & -1\\ 0 & 0 & 0.5 & 0 & 0\\ 0 & 0 & 0 & 0.25 & 0 \\ 0 & 0 & 0 & 0 & 0.125
\end{bmatrix} z + \begin{bmatrix}
    1 \\ 0 \\ 0 \\ 0 \\ 0
\end{bmatrix} u, \\
x =& \begin{bmatrix}
    1 & 0 & 0 & 0 & 0 \\ 
    -2 & -1 & 0 & 0 & 0\\
    2 & 1 & 1 & 0 & 0
\end{bmatrix} z.
\end{aligned}
\end{equation}
This embedding may not seem immediate. 
Indeed, we can find the invertible matrix
\begin{align*}
    T = \begin{bmatrix}
        1 & 0 & 0
        \\
        -2 & -1 & 0
        \\
        0 & 1 & 1
    \end{bmatrix}
\end{align*}
to transform the system into the CAP structure~\eqref{eq:CAP+Koopman-invariant}. In the new coordinates $\tilde{x} := (\tilde{x}_1, \tilde{x}_2, \tilde{x}_3) = Tx$, we have
\begin{equation}\label{eq:CAP-example}
\begin{bmatrix}
    \tilde{x}_1^+ \\ \tilde{x}_2^+ \\ \tilde{x}_3^+
\end{bmatrix} = \begin{bmatrix}
    \tilde{x}_3^{2} - \tilde{x}_1- \tilde{x}_2 + u \\ -\tilde{x}_3^{3}-\tilde{x}_3^{2}+\tilde{x}_1+2\tilde{x}_2 \\ 0.5 \tilde{x}_3
\end{bmatrix}.
\end{equation}
Then, we can identify the associated lifting function for $\tilde{x}_3^{+} = h_2(\tilde{x}_3) := 0.5\tilde{x}_3$ as $\basisfunvecbar := (\tilde{x}_3, \tilde{x}_3^{2}, \tilde{x}_3^{3})$. The function $\basisfunvecbar$ includes all nonlinearities of the system (i.e., $\tilde{x}_3^{2}, \tilde{x}_3^{3}$) and evolves linearly as $h_2$ is a linear function.  The linear embedding~\eqref{eq:Koopman-model-example} can be obtained directly from~\eqref{eq:CAP-example}.
 \hfill$\blacksquare$
\end{example}

As the above characterization reveals, only systems with a particular structure admit exact LTI embeddings. Therefore,
in general, LTI modeling induces a structural bias~\cite{heeg2026limitations} such that only approximate embeddings can be obtained. 
This raises a key question:
while it is clear from our discussion on Koopman invariance that we approximate the Koopman operator when we deal with autonomous systems, cf. Section~\ref{subsec:Koopman-invariance}, 
what is the appropriate mathematical object that we are approximating when we deal with systems with inputs? We introduce this object in the next section.

\mysubsec{Koopman control family}{\secvbpages}
\label{subsec:control:family}

The challenge for extending Koopman operator theory 
to systems with inputs is that, unlike system~\eqref{eq:dynamics:DT}, the behavior of the control system~\eqref{eq:controlled-dynamics:DT} cannot be determined without knowledge of the input sequence. In what follows, we describe the key role that the notion of invariance plays in articulating a consistent extension of the notion of Koopman operator to control systems.

We start from the observation that, if we fix the input as a constant, we get a system in the
form of~\eqref{eq:dynamics:DT}, which admits a well-defined Koopman operator, see, e.g., \cite{WillHema16:control,ProcBrun18:control}.
Motivated by this idea, one can model the system~\eqref{eq:controlled-dynamics:DT}  by
switching between constant input systems. Formally, given a constant $u \in \mathbb{U}$, let $F_u: \Omega \rightarrow \Xi$ be defined by
\begin{align*}
    x \mapsto F_u(x) = F(x,u). 
\end{align*}
Note that any trajectory $\{x_k\}_{k=0}^L \subset{\Omega}$ of system~\eqref{eq:controlled-dynamics:DT} generated with input sequence
$\{u_k\}_{k=0}^{L-1} \subset\mathbb{U}
$ can be described as $x_k = F_{u_{k-1}} \circ \cdots \circ F_{u_0}  (x_0)$.

Consider the discrete-dynamical system
\begin{equation*}
    x^+ = F_u(x) .
\end{equation*}
Being this now a dynamics of the same form as~\eqref{eq:dynamics:DT}, we can consider the associated Koopman operator, that we denote as~$\mathcal{K}_u$. To capture the entirety of possible evolutions, we really need to consider all possible choices of inputs. This then naturally gives rise to the notion of \emph{Koopman control family (KCF)}, defined as $\{\mathcal{K}_{u} : \mathcal{F} \to \mathcal{F} \}_{u \in \mathbb{U}}$, see~\cite{HaseCort26:control:family}. This is the object that encodes, from a Koopman perspective, the realm of possibilities encoded by all possible control inputs, which is what we have at our disposal when we consider systems of the form~\eqref{eq:controlled-dynamics:DT}.

The fact that, when dealing with control systems, we move from one Koopman operator to an infinite family of them seems daunting. However, we know these operators are  closely related to each other as they all arise from the same controlled dynamics~\eqref{eq:controlled-dynamics:DT}. 
As we show next, the notion of Koopman invariance plays a key role in dealing with this complexity. \looseness=-1

Extending the notion of Koopman invariance introduced in \Cref{subsec:Koopman-invariance}, we call a subspace $\mathbb{V} \subset \mathcal{F}$ invariant under the Koopman control family
if, for all $u \in \mathbb{U}$, one has
\begin{align*}
  \mathcal{K}_u \observable \in \mathbb{V} \qquad\forall\,\observable \in \mathbb{V} .
\end{align*}
This means that $\mathbb{V}$ is a \emph{common} invariant subspace for all the Koopman operators in the family.

One can show, cf.~\cite[Theorem~4.3]{HaseCort26:control:family}, that the Koopman control family has a finite-dimensional common invariant subspace $\mathbb{V}$
if and only if it admits a basis $\basisfunvec = (\basisfun_1,\dots,\basisfun_{\nD}): \Omega \to \mathbb{K}^{\nD}$, with $\nD=\dim{\mathbb{V}}$, and a function $\mathcal{A}: \mathbb{U} \to \mathbb{K}^{\nD \times \nD}$ such that
\begin{align}\label{eq:composition-to-separation}
    \basisfunvec(x^+) = \mathcal{A}(u) \basisfunvec(x), \quad \forall\,(x,u) \in \Omega \times \mathbb{U}.
\end{align}
This means that all the observables in $\mathbb{V}$ evolve linearly under all the members of the Koopman control family. We refer to the general form for the evolution of the common invariant subspace described in~\eqref{eq:composition-to-separation} as the \emph{input-state separable} form.

Interestingly, the input-state separable form encompasses commonly employed Koopman-inspired approximation models when dealing with control systems, like linear and bilinear ones. For instance, a linear finite-dimensional lifted representation of the form
\begin{align*}
\basisfunvecvariant(x^+) = A \basisfunvecvariant(x) + B u,
\end{align*}
where $\basisfunvecvariant\!:\! \Omega \!\to\! \mathbb{K}^{n_{\basisfunvariant}}$, $A \!\in\! \mathbb{R}^{n_\basisfunvariant ~\!\!\!\times ~\!\!\!n_\basisfunvariant}$ and $B\! \in\! \mathbb{R}^{n_\basisfunvariant ~\!\!\!\times~\!\!\! m}$ corresponds to the common invariant subspace $\mathbb{V} \!=\! \mathrm{span} (\basisfunvecvariant,1)$ under the Koopman control family, with input-state separable representation
\begin{align*}
\begin{bmatrix}
\basisfunvecvariant(x^+)
\\
1
\end{bmatrix}
= \begin{bmatrix}
A
& B u
\\
0
&
1
\end{bmatrix}
\begin{bmatrix}
\basisfunvecvariant(x)
\\
1
\end{bmatrix}.
\end{align*}
In particular, an exact LTI embedding~\eqref{eqn:lifted-state} is a
special case of the input-state separable form where the subspace $\mathbb{V}$ contains the state variables.

Likewise,  a bilinear finite-dimensional lifted representation of the form
\begin{align*}
\basisfunvecvariant(x^+) = A \basisfunvecvariant(x) + \sum\nolimits_{i=1}^m  D_i \basisfunvecvariant(x) u_i + B u,
\end{align*}
where $D_i \in \mathbb{R}^{n_\basisfunvariant \times n_\basisfunvariant}$, $i \in [m]$, corresponds to the common invariant subspace $\mathbb{V} = \mathrm{span} (\basisfunvecvariant,1)$ under the Koopman control family, with input-state separable representation
\begin{align*}
\begin{bmatrix}
\basisfunvecvariant(x^+)
\\
1
\end{bmatrix}
= \begin{bmatrix}
A +\sum_{i=1}^m  u_i D_i
& Bu
\\
0
&
1
\end{bmatrix}
\begin{bmatrix}
\basisfunvecvariant(x)
\\
1
\end{bmatrix}.
\end{align*}
This observation means that, if one seeks to find a linear or bilinear representation of the evolution of the control system~\eqref{eq:controlled-dynamics:DT}, what one is really doing is presuming the existence of a particular type of common invariant subspace of the infinite family of Koopman operators associated to the dynamics. The converse is not necessarily true, as corroborated by the fact that a common invariant subspace under the KCF exists if and only if the evolution can be written in input-state separable form. The following example, taken from~\cite{HaseCort26:control:family}, illustrates this point. 
\begin{example}
    Consider  the system with state $x = [x_1,x_2]^T \in \mathbb{R}^2$ and input $u \in \mathbb{R}$
    \begin{align*}
        x_1^+ & = a x_1 +b u \nonumber \\
        x_2^+ & = c x_2 + d x_1^2  + e x_1 u + fu + g \sin(u) + h
    \end{align*}
    where $a,b,c,d,e,f,g,h \in \mathbb{R}$.
    The system admits the input-state separable form
    \begin{align*}
        \begin{bmatrix}
            x_1 \\
            x_2 \\
            x_1^2 \\
            1
        \end{bmatrix}^+ \!\! = \begin{bmatrix}
            a & 0 & 0 & b \,u \\
            e \, u & c & d & fu + g \sin(u) + h \\
            2ab \, u & 0 & a^2 & b^2 u^2 \\
            0 & 0 & 0 & 1
        \end{bmatrix} \begin{bmatrix}
            x_1 \\
            x_2 \\
            x_1^2 \\
            1
        \end{bmatrix}.
    \end{align*}
    If additionally $b \!=\! g \!=\! 0$, the system admits the
    bilinear form 
    \begin{align*}
        \begin{bmatrix}
        x_1 \\ x_2 \\ x_1^2 \\ 1
        \end{bmatrix}^+
        \!\!=  
        \begin{bmatrix}
        a & 0 & 0 & 0 
        \\
        0 & c & d & h 
        \\
        0 & 0 & a^2 & 0
        \\
        0 & 0 & 0 & 1
        \end{bmatrix}
        \begin{bmatrix}x_1 \\ x_2 \\ x_1^2 \\ 1 \end{bmatrix}
        +
        \begin{bmatrix}
        0 & 0 & 0 & 0 
        \\
        e & 0 & 0 & f 
        \\
        0 & 0 & 0 & 0
        \\
        0 & 0 & 0 & 0
        \end{bmatrix}
        \begin{bmatrix}x_1 \\ x_2 \\ x_1^2 \\ 1 \end{bmatrix} u.
    \end{align*}
    If in addition $e = 0$, the system admits the linear form 
    \begin{align*}
        \begin{bmatrix}
            x_1 \\ x_2 \\ x_1^2 \\ 1
        \end{bmatrix}^+ = \begin{bmatrix}
            a & 0 & 0 & 0 \\
            0 & c & d & h \\
            0 & 0 & a^2 & 0 \\
            0 & 0 & 0 & 1
        \end{bmatrix} \begin{bmatrix}
            x_1 \\ x_2 \\ x_1^2 \\ 1 
        \end{bmatrix} + \begin{bmatrix}
            0 \\
            f \\
            0 \\
            0
        \end{bmatrix}u
    \end{align*}
     \hfill$\blacksquare$
\end{example}

\noindent\textbf{Koopman invariance for control systems}.
In general, finding a subspace $\mathbb{V}$ that is invariant under the KCF is challenging. In a similar fashion to what we have described when discussing Koopman invariance for autonomous systems, cf. Section~\ref{subsec:Koopman-invariance}, the ability of a model like~\eqref{eq:composition-to-separation} to accurately capture the evolution of the original system~\eqref{eq:controlled-dynamics:DT} hinges upon how close $\mathbb{V}$ is to being invariant under the KCF. 
Again, this can be precisely quantified by resorting to the notion of invariance proximity, as we explain next. To do so, we find it convenient to consider the augmented system
\begin{align}\label{eq:augmented-system}
(x^+,u^+) =  F^{\aug}(x,u) := (F(x,u), u), 
\end{align}
for $(x,\!u) \!\in\! \Omega \!\times\! \mathbb{U}$.  Note that
in~\eqref{eq:augmented-system}, $u$ is a part of the state vector and
not an input. To define the Koopman operator associated with this augmented system, we first specify the function space. This is given by $\mathcal{F}^{\aug}$, a linear space of $\mathbb{K}$-valued
functions with domain $\Omega ~\!\!\!\times~\!\!\! \mathbb{U}$, with the following properties\looseness=-1
\begin{enumerate}
\item for all  $\upsilon \in \mathcal{F}^{\aug}$, $ \upsilon \circ F^{\aug}$ belongs to  $\mathcal{F}^{\aug}$;
\item for all $\observable \in \mathcal{F}$, $\observable \circ p$ belongs to  $\mathcal{F}^{\aug}$, where $p:\Omega \times \mathbb{U} \rightarrow \Omega$ is the projection onto the first component;
\item for all  $\upsilon \in \mathcal{F}^{\aug}$ and all $u \in \mathbb{U}$, the map $x \mapsto \upsilon (x,u)$ belongs to $\mathcal{F}$.
\end{enumerate}
These properties allow us to relate the observables defined on the state space $\Omega$, suitable for reasoning about~\eqref{eq:controlled-dynamics:DT}, and the ones defined on $\Omega \times \mathbb{U}$, suitable for reasoning about~\eqref{eq:augmented-system}. In fact, we can go from $\mathcal{F}$ to $\mathcal{F}^{\aug}$ thanks to 2), describe the evolution of the resulting observable thanks to 1), and then go back from  $\mathcal{F}^{\aug}$ to $\mathcal{F}$ thanks to 3).

We refer to the Koopman operator corresponding to the system~\eqref{eq:augmented-system} as the \emph{augmented Koopman operator}
$\mathcal{K}^{\aug}: \mathcal{F}^{\aug} \to \mathcal{F}^{\aug}$, defined by
$   \mathcal{K}^{\aug} \phi = \phi \circ F^{\aug}$ for $\phi \in \mathcal{F}^{\aug}$. As we show next, this operator is particularly useful to identify invariant subspaces under the KCF. In doing so, it is convenient to define the notion of normal space. A subspace $\mathbb{W} \subset \mathcal{F}^{\aug}$ of dimension $s=\dim (\mathbb{W})$ is \emph{normal} if it admits a basis in the following form:
\begin{subequations}
     \begin{align}\label{eq:normal}
 \Upsilon(x,u) = \begin{bmatrix} I_{n \times n} \\ G(u) \end{bmatrix} \basisfunvec(x) ,
 \end{align}
where $\basisfunvec: \Omega \rightarrow \mathbb{K}^\nD$, $G: \mathbb{U} \rightarrow \mathbb{K}^{(s-\nD) \times \nD}$, and $\nD \le s$.

One can show~\cite[Theorem 7.3]{HaseCort26:control:family} that if a normal subspace $\mathbb{W}$ is invariant under $\mathcal{K}^{\aug}$, then
\begin{align}\label{eq:assoc-common-inv}
  \mathbb{V} = \operatorname{span} (\basisfunvec)  
\end{align}
\end{subequations}
is a common invariant subspace under the KCF. In fact, since  $\mathbb{W}$ is invariant, we have
\begin{align*}
\Upsilon(x^+,u^+) = \begin{bmatrix}	A_{11} & A_{12} \\ A_{21} & A_{22}\end{bmatrix} \Upsilon(x,u) .
\end{align*}
Using now the fact that $\mathbb{W}$ is normal, we deduce
\begin{align*}
\basisfunvec(x^+) = \big(A_{11} + A_{12} G(u)\big)  \basisfunvec(x) ,
\end{align*}
which is exactly in the input-state separable form~\eqref{eq:composition-to-separation}. 

This discussion sets the basis for dealing with Koopman-based approximations on subspaces $\mathbb{V}$ that are not invariant under the KCF. In fact, assuming $\mathcal{F}^{\aug}$ is equipped with an inner product, we can reproduce our discussion in Section~\ref{subsec:Koopman-invariance}, now for the operator $\mathcal{K}^{\aug}$ instead of $\mathcal{K}$, to precisely quantify the error 
incurred by approximations on non-invariant subspaces. Formally, given a subspace $\mathbb{W} \subset \mathcal{F}^{\aug}$, consider 
\begin{align*}
    \mathcal{K}^{\aug}_{\apprx} := \mathcal{P}_{\mathbb{W}} \mathcal{K}^{\aug}: \mathcal{F}^{\aug} \to \mathbb{W} \subset \mathcal{F}^{\aug} .
\end{align*}
Note that $\mathbb{W}$ is invariant under the operator $ \mathcal{K}^{\aug}_{\apprx}$. The invariance proximity of 
$\mathbb{W}$ is then
\begin{align*}
I_{\mathcal{K}^{\aug}} (\mathbb{W}) = \sup_{\upsilon \in \mathbb{W}, \|\mathcal{K}^{\aug} \upsilon \| \neq 0} \frac{\|\mathcal{K}^{\aug} \upsilon -  \mathcal{K}^{\aug}_{\apprx} \upsilon \|}{\| \mathcal{K}^{\aug} \upsilon \|}.  
\end{align*}
By definition, we have that $\mathbb{W}$ is invariant under $\mathcal{K}^{\aug}$ if and only if $I_{\mathcal{K}^{\aug}} (\mathbb{W})=0$. Invariance proximity measures the worst-case relative error of approximation by projecting the action of $\mathcal{K}^{\aug}$ on~$\mathbb{W}$. Therefore, to obtain good approximate representations of~\eqref{eq:controlled-dynamics:DT}, one should look for normal subspaces $\mathbb{W}$ of $\mathcal{F}^{\aug}$, cf.~\eqref{eq:normal}, with low values of invariance proximity, which yield approximately invariant subspaces under the KCF, cf.~\eqref{eq:assoc-common-inv}, with rigorous accuracy guarantees.

\mysubsec{Koopman operator for control via tensor products}{\secvbpages}
\label{subsec:product:Hilbert}
A third route to handling control inputs, complementary to the
exact LTI embeddings of \Cref{subsec:LTI} and the Koopman control
family of \Cref{subsec:control:family}, has been recently proposed
in \cite{lazar2026product}. The idea is to treat the input in
the same way as the state, namely through observable functions of
its own, instead of regarding it as an exogenous signal or as a
parameter of a family of autonomous Koopman operators. Let
$\mathcal{F}_x$ and $\mathcal{F}_u$ denote separable Hilbert spaces
of state and input observables, with Riesz bases
$\{\basisfun_i^x\}_{i \in \mathbb{N}}$ and
$\{\basisfun_j^u\}_{j \in \mathbb{N}}$, respectively. The lifted
observable space is then taken as the tensor product
$\mathcal{F} := \mathcal{F}_x \otimes \mathcal{F}_u$, the
Hilbert space spanned by all products of a state observable and an
input observable. It is itself a separable Hilbert space admitting
the product Riesz basis
$\{\basisfun_i^x \otimes \basisfun_j^u\}_{i,j \in \mathbb{N}}$ with
$(\basisfun_i^x \otimes \basisfun_j^u)(x,u) = \basisfun_i^x(x)\basisfun_j^u(u)$.
On this product space, the action of the controlled dynamics is
captured by the generalized Koopman (GeKo) operator
$\mathcal{K} : \mathcal{F}_x \to \mathcal{F}$, defined by
$(\mathcal{K}\observable^x)(x,u) = \observable^x(F(x,u))$ for
$\observable^x \in \mathcal{F}_x$. Since $F(x,u)$ depends on
both $x$ and $u$, the propagated observable is a function of the
pair $(x,u)$ and therefore belongs to $\mathcal{F}$ rather than to
$\mathcal{F}_x$, which is why $\mathcal{K}$ maps $\mathcal{F}_x$
into the product space. Expanding each propagated basis
function in the product basis,
\begin{equation}\label{eq:GeKo:expansion}
    (\mathcal{K}\basisfun_i^x)(x,u)
    = \sum_{j,l \in \mathbb{N}} K_{i,(j,l)}\,
      \basisfun_j^x(x)\,\basisfun_l^u(u),
\end{equation}
and collecting the coefficients $K_{i,(j,l)}$ into a matrix $K$
yields the \emph{exact infinite-dimensional bilinear Koopman system}
\begin{equation}\label{eq:GeKo:bilinear}
    z^+ = K \, (z \otimes v),
\end{equation}
where $z = (\basisfun_i^x(x))_{i \in \mathbb{N}}$ is the lifted state,
$v = (\basisfun_j^u(u))_{j \in \mathbb{N}}$ is the lifted input, and
$K$ is the matrix representation on $\ell^2$ of the GeKo operator
$\mathcal{K}$ in the product basis. Hence the lifted state
evolves linearly in the Kronecker product $z \otimes v$, and every
nonlinear dependence on the input is absorbed in the input basis.
Compared with other Koopman representations
for systems with input, the GeKo approach yields an exact
infinite-dimensional bilinear Koopman model in the lifted state and
input space, which can be further exploited to derive a
\emph{nonlinear fundamental lemma} and to analyze operator
convergence via Galerkin projections, see \cite{lazar2026product}
for details.

\begin{remark}\longthmtitle{(Bi)Linear forms as finite-$\mathcal{F}_u$ truncations}\label{rem:GeKo:linbilin}
The linear and bilinear Koopman models are recovered from
\eqref{eq:GeKo:bilinear} by truncating $\mathcal{F}_u$ to the
finite input basis $v = [\,1,\ u_1,\ \dots,\ u_m\,]^\top$,
so that
\begin{equation*}
    z \otimes v
    = \big[\, z^\top,\ u_1 z^\top,\ \dots,\ u_m z^\top \,\big]^\top .
\end{equation*}
Partitioning the GeKo matrix as
$K \!=\! [A \!\mid\! B_1 \!\mid \dots \mid\! B_m]$ gives\looseness=-1
\begin{equation}\label{eq:GeKo:bilinear:rec}
    z^+ = A z + \sum\nolimits_{i=1}^{m} B_i\, z\, u_i .
\end{equation}
When the lifted state contains the constant (unit) observable
$\basisfun_0^x \equiv 1$, i.e., \ $z = [\,1,\ \tilde z^\top\,]^\top$, we obtain
$B_i z u_i = b_i u_i + N_i \tilde z\, u_i$, with $b_i$ the column of $B_i$ acting on the unit observable and
$N_i$ acting on $\tilde z$. Then \eqref{eq:GeKo:bilinear:rec}
becomes the \emph{bilinear} lifted model
\begin{equation*}
    z^+ = A z + B u + \sum\nolimits_{i=1}^{m} N_i\, \tilde z\, u_i ,
    \qquad B := [\,b_1,\ \dots,\ b_m\,] .
\end{equation*}
The \emph{linear} model is the special case in which the input
modes couple only to the constant state observable, i.e.
$N_i = 0$ for all $i$,\looseness=-1
\begin{equation*}
    z^+ = A z + B u .
\end{equation*}
Endowing $\mathcal{F}_u$ with a richer basis, e.g. the Chebyshev
family $\{T_0, \dots, T_{d_u}\}$ of degree $d_u$, yields input
couplings beyond the bilinear case.
\end{remark}

\mysubsec{EDMD with control}{\secvcpages}
\label{subsec:EDMD:control}

Following \cite{KordMezi18}, we consider the linear Koopman
predictor with a state dictionary $\{\basisfun_i\}_{i \in [n]}$ and
lifted state $z = \basisfunvec(x) \in \mathbb{R}^n$,
\begin{equation}\label{eq:EDMDc:predictor}
    z^+ = A \, z + B \, u,
    \qquad A \in \mathbb{R}^{n \times n},
    \quad B \in \mathbb{R}^{n \times m}.
\end{equation}
Given data
$\{(x^{(i)}, u^{(i)}, y^{(i)} \!=\! F(x^{(i)}, u^{(i)})) \!\mid\! i \!\in\! [N]\}$, define
\begin{equation*}
    [\basismat_X]_{i,j} = \basisfun_i(x^{(j)}), \quad
    [\basismat_Y]_{i,j} = \basisfun_i(y^{(j)}), \quad
    [U]_{:,j} = u^{(j)}.
\end{equation*}
Analogously to \eqref{eq:LS_dictionary}, the matrices $(A,B)$ follow from
\begin{equation}\label{eq:EDMDc:LS}
    \begin{split}
    [A \ \ B]
    &= \arg\min\nolimits_{A,B}
      \left\| \basismat_Y - A\,\basismat_X - B\,U \right\|_{\mathrm{Fr}}^2\\
    &= \basismat_Y W^\top (W W^\top)^{\dagger},
\end{split}
\end{equation}
with $W := [\basismat_X^\top \; U^\top]^\top$, recovering the autonomous
estimator of \eqref{eq:LS_dictionary} when $B = 0$.

When this approach is extended to bilinear Koopman models as in
\cite{BrudFu21,StraScha26:SafEDMD} and~\cite{strasser:schaller:berberich:worthmann:allgower:2025csl} for kernel EDMDc, EDMDc yields the predictor
\begin{equation}\label{eq:EDMDc:bilinear:predictor}
    z^+ = A \, z + B \, u + \sum\nolimits_{k=1}^{m} N_k \, z \, u_k,
\end{equation}
with matrices $A, N_k \in \mathbb{R}^{n \times n}$,
$B \in \mathbb{R}^{n \times m}$, where $u_k$ denotes the {$k$-th}~component of $u \in \mathbb{R}^m$. The bilinear
regressors are conveniently encoded via the column-wise Khatri--Rao
product $\odot$, defined for matrices $P \in \mathbb{R}^{a \times N}$
and $Q \in \mathbb{R}^{b \times N}$ as
$P \odot Q \in \mathbb{R}^{ab \times N}$ with column $j$ given by the
Kronecker product $p_j \otimes q_j$ of the corresponding columns. The
regressor matrix is then
\begin{equation*}
    \Theta := \begin{bmatrix} \basismat_X \\ U \\ U \odot \basismat_X
            \end{bmatrix} \in \mathbb{R}^{n + m + nm \times N},
\end{equation*}
and the parameter blocks $(A, B, N_1, \dots, N_m)$ follow from
\begin{equation}\label{eq:EDMDc:bilinear:LS}
    [A \ \ B \ \ N_1 \ \cdots \ N_m]
    = \basismat_Y \, \Theta^\top (\Theta \Theta^\top)^{\dagger},
\end{equation}
in direct analogy with \eqref{eq:EDMDc:LS}.
\miniskip

\noindent\textbf{EDMDc for Koopman control family}.
For the Koopman control family of \Cref{subsec:control:family},
EDMDc specializes to the universal \emph{input-state separable} form
\eqref{eq:composition-to-separation},
\begin{equation*}
    \basisfunvec(x^+) = \mathcal{A}(u)\,\basisfunvec(x),
    \qquad \mathcal{A} : \mathbb{U} \to \mathbb{R}^{\nD \times \nD},
\end{equation*}
in which the input enters through a general matrix-valued function
$\mathcal{A}(u)$ rather than affinely or bilinearly. A practical instance follows from the normal-form basis
\eqref{eq:normal}, 
$\Upsilon(x,u) = [I_{\nD}\; G(u)^\top]^\top \basisfunvec(x)$,
whose top block is the identity and whose lower block $G(u)$ carries
the input dependence. Applying EDMD to the augmented data matrices
$Z = [X^\top\; U^\top]^\top$, $Z^+ = [(X^+)^\top \; U^\top]^\top$ yields
$\tilde{A} = \Upsilon(Z^+)\Upsilon(Z)^{\dagger} \in \mathbb{R}^{s \times s}$, 
whose block decomposition along $(\nD, s-\nD)$ gives
\begin{equation}\label{eq:EDMDc:KCF:predictor}
    \basisfunvec(x^+) = \tilde{\mathcal{A}}(u)\,\basisfunvec(x),
    \qquad
    \tilde{\mathcal{A}}(u) = \tilde{A}_{11} + \tilde{A}_{12}\,G(u),
\end{equation}
with $\tilde{A}_{11} \in \mathbb{R}^{\nD \times \nD}$ and
$\tilde{A}_{12} \in \mathbb{R}^{\nD \times (s-\nD)}$, in agreement with
\Cref{subsec:control:family}. The input-dependent block $G(u)$, together with the state basis
$\basisfunvec$, may be parameterized as neural networks and jointly learned
by minimizing the invariance proximity of the augmented subspace, as
proposed in \cite{HaseCort26:control:family}. 

Alternatively, and as
pursued in \Cref{subsec:EDMDc:implementation}, in this paper we fix a priori both
the state basis and the coupling using kernel or
polynomial dictionaries, i.e., the state basis $\basisfunvec$ is a fixed dictionary
as in \eqref{eq:GeKo:EDMDc:data}, and $G(u)$ couples the input
dictionary to a \emph{subset} of the lifted state coordinates,
yielding a sparse input-state coupling that retains the
input-state separable structure \eqref{eq:EDMDc:KCF:predictor} while
avoiding the non-convex training of the learned parameterization.
For an extension of the error analysis discussed in \Cref{subsec:EDMD_error} and \Cref{subsec:EDMD_error_full}, we refer to the prototypical extensions~\cite{nuske2023finite,BoldPhil25} and~\cite{bevanda2026nonparametric}.
\miniskip

\noindent\textbf{EDMDc for the GeKo model}. The GeKo model
\eqref{eq:GeKo:bilinear} admits an analogous EDMDc
estimator. Given state and input dictionaries
$\{\basisfun_i^x\}_{i \in [n]}$ and $\{\basisfun_j^u\}_{j \in [p]}$ (e.g., RBFs on
$x$ and $u$), define the lifted state and input
$z = \basisfun_{[n]}^x(x) \in \mathbb{R}^n$,
$v = \basisfun_{[p]}^u(u) \in \mathbb{R}^p$, and the data matrices\looseness=-1
\begin{equation}\label{eq:GeKo:EDMDc:data}
    [Z_X]_{i,j} = \basisfun_i^x(x^{(j)}), ~
    [Z_Y]_{i,j} = \basisfun_i^x(y^{(j)}), ~
    [V]_{j,\ell} = \basisfun_j^u(u^{(\ell)}).
\end{equation}
Truncating \eqref{eq:GeKo:bilinear} to $n,p$ gives
$z^+ = K\,(z \otimes v)$ with $K \in \mathbb{R}^{n \times np}$,
which is linear in $K$ with regressor $z \otimes v$. Since the
$\ell$-th column of the Khatri--Rao product $Z_X \odot V$
introduced above is exactly $z^{(\ell)} \otimes v^{(\ell)}$, the
least-squares fit of the truncated model to the data is\looseness=-1
\begin{equation}\label{eq:GeKo:EDMDc:LS}
    K = Z_Y \, (Z_X \odot V)^\top
        \bigl( (Z_X \odot V)(Z_X \odot V)^\top \bigr)^{\dagger},
\end{equation}
in direct analogy with \eqref{eq:EDMDc:LS} and
\eqref{eq:EDMDc:bilinear:LS}, where the only difference comes from the 
regressor structure. 

\begin{remark}\longthmtitle{Multi-step Koopman learning overview}\label{rem:multistep:overview}
Data-driven learning of Koopman operators via EDMD typically solves the above-presented one-step regression problems. The advantage of these approaches is that an operator is learned which allows prediction over arbitrary horizons. Alternatively, in order to improve the multi-step prediction error, fitting a Koopman operator using more than a single time step has been
explored in \emph{linear-in-control} setting.
\cite{KordMezi18} uses a multi-step prediction loss to identify the input
matrix of a lifted linear predictor, after the autonomous part has been
fixed. In \cite{DeJongCDC2024} the observables are parameterized as a
neural network and the multi-step output prediction error is minimized
during training; the horizon-stacked prediction matrices of the resulting
linear multi-step Koopman model are then directly refitted by least
squares. Recently, \cite{WuTan2026:BraatzDrgona} develops an alternative
multi-step learning framework for linear Koopman models in which the
condensed horizon-dependent state--input mappings are fitted
\emph{independently} per prediction step, preserving a convex per-step
least-squares structure and supporting parallel computation and
$\ell_1$-regularized dictionary pruning. The common feature of these
approaches is that they trade the recursive Koopman operator structure for
a fixed-horizon stacked predictor: error compounding is avoided by design,
but the resulting model is no longer a single Koopman operator iterable at
arbitrary horizon. Beyond the linear-in-control setting, two approaches that optimize multi-step prediction capabilities for Koopman operators have been developed recently. \cite{lazar2026khatrirao} develops a multi-step optimized regression problem using time-sequenced data for learning a single generalized Koopman operator \cite{lazar2026product}, which is iterable over an arbitrary horizon and yields multi-step regression solutions for both linear and bilinear Koopman forms with non-lifted inputs. The recent paper~\cite{MH-JC-JWB:26-cdc} introduces the control consistency index and exploits its underlying geometric invariance properties to learn finite-dimensional representations with strong multi-step prediction capabilities.
\end{remark}

\mysubsec{Implementation and numerical comparison}{\secvcpages}
\label{subsec:EDMDc:implementation}
We compare the four EDMDc variants, linear, bilinear,
GeKo, and KCF, on the same pipeline. The complete Matlab implementation
is available at \url{https://github.com/KOT-tutorial/CDC26}.

\miniskip
\noindent\textbf{Common pipeline.}
For each example we fix a state dictionary
$\basisfun^x : \mathbb{R}^d \to \mathbb{R}^{n}$ and an input dictionary
$\basisfun^u : \mathbb{R} \to \mathbb{R}^{p}$. The state dictionary is
selected from Gaussian RBFs, rational quadratic, Mat{\'e}rn-$5/2$ or
multivariate polynomials; in the kernel cases, centers
$\{c_i\}_{i \in [n]}$ are placed by $k$-means on the state data and
the lifted features are scale-normalized as
$\basisfun_i^x(x) = \sqrt{2/n}\,k(x, c_i)$. The input dictionary is taken as the Chebyshev family
$\{T_0, T_1, \dots, T_{d_u}\}$ of degree $d_u$ on the pre-scaled input
$S_u u \in [-1, 1]$, the natural domain on which the Chebyshev
features stay well-conditioned independently of the physical input
range.
\miniskip\vspace{-0.03cm}

\noindent\textbf{Multi-step regression.}
The standard EDMDc estimators of
\eqref{eq:EDMDc:LS}, \eqref{eq:EDMDc:bilinear:LS},
\eqref{eq:GeKo:EDMDc:LS}, and \eqref{eq:EDMDc:KCF:predictor}
minimize the one-step-ahead prediction error, which can propagate
unfavorably when the identified models are deployed over long
horizons, e.g., within MPC. To mitigate multi-step error propagation
while still using kernel regression, we borrow the idea from
\cite{lazar2026khatrirao} and fit the four models on \emph{multi-step} training pairs: from each training trajectory we extract sliding windows of
horizon $H_t$, select a diverse subset of $M$ windows by $k$-means on the stacked window descriptor 
$[(z^{(0)})^\top\; (z^{(H_t/2)})^\top\; (z^{(H_t)})^\top]^\top$ 
(start,
midpoint, and end of the lifted window). Then we keep the window nearest
to each cluster centre, and unfold each window into $H_t$ one-step pairs
$\{(z^{(k)}, u^{(k)}, z^{(k+1)})\}_{k=0}^{H_t-1}$. The resulting $N = H_t M$ pairs are stacked into $(Z_X, U, Z_Y)$ and used to
solve \eqref{eq:EDMDc:LS}, \eqref{eq:EDMDc:bilinear:LS},
\eqref{eq:GeKo:EDMDc:LS}, and \eqref{eq:EDMDc:KCF:predictor} with
Tikhonov regularization $\gamma I$. For the comparison against
standard EDMD reported in Tables~\ref{tab:duffing:poly} and
\ref{tab:dcmotor:results}, the one-step baseline is fit on a random
subsample of the same size $N = H_t M$ drawn uniformly from all available
one-step transitions in the training trajectories, so the two
regimes differ only in the regression objective and not in the
amount of data. For the KCF estimator, the constant Chebyshev mode
$T_0 \equiv 1$ makes $z$ appear both as the $A_{11}$ block and as
the $T_0$ slice of $v \otimes z$; we drop $T_0$ from the bilinear
block so that $A_{11}$ uniquely carries the autonomous channel. A linear decoder
\begin{equation}
\label{eq:decoder}
D = X Z_X^\top (Z_X Z_X^\top + \gamma I)^{-1}
\end{equation}
maps predicted lifted states back to $\mathbb{R}^d$.

\begin{remark}\longthmtitle{Prepending the state}
\label{rem:prepend}
When the state dictionary does not already contain the coordinate
functions, one may prepend $[1, x^\top]^\top$ to the lift, so that
the original state occupies known rows of $z$. The decoder
\eqref{eq:decoder} then reduces to a fixed projection onto those
rows, shared by all four forms, rather than a regressed map. Although our KOT software toolbox implements both options, below we report the results that use the plain regressed decoder \eqref{eq:decoder} for brevity.
\end{remark}
\miniskip\vspace{-0.03cm}
\noindent\textbf{Representative system 1: forced Duffing oscillator}. We consider an 
oscillatory system with multiple fixed points, i.e., \looseness=-1
\begin{align*}
\dot{x}_1 = x_2,
&&\dot{x}_2 = -\delta x_2 - \alpha x_1 - \beta x_1^3 + u,
\end{align*}
with $(\delta, \alpha, \beta) \!=\! (0.5, -1, 1)$, $T_s \!=\! 0.05$\,s, and
$|u| \le 3$. 
Training: $L \!=\! 6$ multi-sine trajectories of length
$T_{\mathrm{long}} \!=\! 2000$, with initial conditions balanced across
both wells and the saddle; $M \!=\! 800$ multi-step windows of horizon
$H_t \!=\! 20$; $\gamma \!=\! 10^{-4}$. State dictionary: degree-$4$
multivariate polynomial ($n \!=\! 15$); input dictionary: Chebyshev
$T_0,\dots,T_3$ ($p \!=\! 4$).
Input is pre-scaled by $S_u = 1/3$.
\looseness=-1
\miniskip\vspace{-0.03cm}

\noindent\textbf{Representative system~2: DC motor with input nonlinearity}.
We further consider the 
example of \cite[Ex.~9.1]{HaseCort26:control:family}:
\begin{equation*}
\begin{split}
\dot{x}_1 &= -(R_a/L_a) x_1 - (k_m/L_a) x_2 f(u) + u_a/L_a,\\
\dot{x}_2 &= -(B/J) x_2 + (k_m/J) x_1 f(u) - \tau_l/J,
\end{split}
\end{equation*}
with $f(u) = 2\tanh(u\cos u)$ (non-monotone, hardest case in
\cite{HaseCort26:control:family}), $T_s = 5$\,ms, and $|u| \le 4$.
Training as for the representative system 1
with the same $(L, H_t, M, \gamma)$, except
$\deg\basisfun^x = 3$ ($n = 10$) and a richer input dictionary
$T_0,\dots,T_8$ ($p = 9$), required by the strong input
nonlinearity $f(u)$. State and input are pre-scaled by
$S_x = \mathrm{diag}(0.1, 0.004)$ and $S_u = 0.25$ to bring all
quantities to comparable magnitudes.

\noindent\textbf{Training and evaluation setting}.
We fit each form with one-step and multi-step regression: one-step
pairs are subsampled uniformly from all training transitions, and
multi-step pairs are extracted from the selected sliding windows
(\,$H_t = 20$, $\gamma = 10^{-4}$\,), matched in number so the two
regimes differ only in the regression objective. Over
$N_{\mathrm{test}} = 10$ test initial conditions we report the
median, mean, and worst per-trajectory Euclidean state error,
$\operatorname{mean}_t \|\hat{x}_t - x^*_t\|$. The total operator
Frobenius norm is $\sqrt{\|A\|^2 + \|B\|^2}$ (Linear),
$\sqrt{\|A\|^2 + \|B\|^2 + \|N\|^2}$ (Bilinear), $\|K\|$ (GeKo), and
$\sqrt{\|A_{11}\|^2 + \|A_{12}\|^2}$ (KCF).

\begin{figure}[tb]
\centering
\includegraphics[width=0.75\columnwidth]{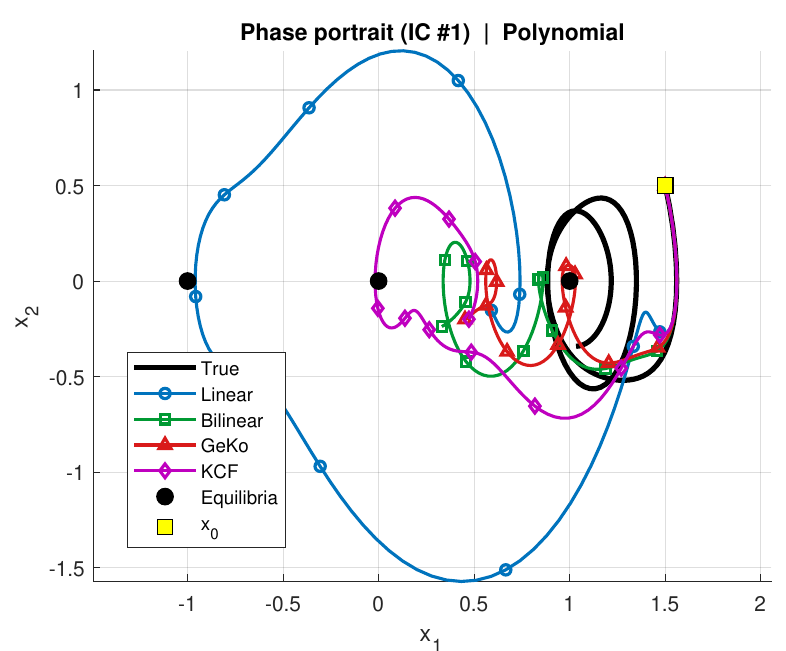}
\vspace{-0.35cm}
\caption{Phase trajectory plot, Duffing.}
\label{fig:duf:phase}
\end{figure}

\begin{figure}[tb]
\centering
\includegraphics[width=0.8\columnwidth]{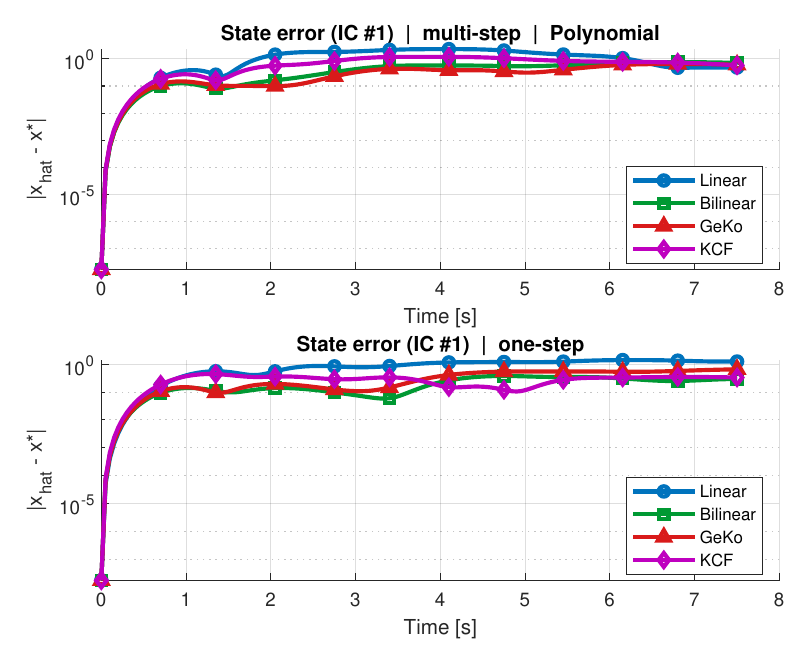}
\vspace{-0.35cm}
\caption{State error rollout, Duffing.}
\label{fig:duf:error}
\end{figure}

\begin{example}[Forced Duffing oscillator]
On the Duffing oscillator the input enters linearly, so a degree-$4$
polynomial state dictionary already captures the dynamics and the
typical-case accuracy is best for the forms with the least
input-coupling overhead: Bilinear and GeKo attain the lowest median
errors ($0.52$ and $0.55$). The ranking inverts on the worst case.
Initial condition~$9$ (see the GitHub repository) is hard for every form, i.e., it is the largest
error for all four, but the amount each degrades there tracks how
much input--state coupling it carries: KCF, which couples the input
to only $6$ of the $15$ lift features, is the most robust ($0.74$
median, $1.05$ worst), whereas the fully coupled GeKo is the least
($0.55$ median but $2.98$ worst), with Linear and Bilinear in
between. The effect is not removed by regularization or by changing
the number of training windows; it reflects the high variance of the
unconstrained coupling on an out-of-distribution initial condition.
KCF also retains the smallest operator of the input-lifted forms
($4.31$ vs.\ GeKo's $4.78$). A coupling sweep confirms an interior
optimum: coupling to the degree-$\le 2$ monomials ($q=2$,
$6$ features) outperforms both the leanest ($q=1$) and the
full-coupling ($q=4$) budgets.\looseness=-1 \hfill$\blacksquare$
\end{example}

\begin{table}[tb]
\centering
\caption{Forced Duffing oscillator.  KCF couples the input to $6$ of the $15$ lift features
($s = 33$). Multi-step open-loop state error over $N_{\mathrm{test}}
= 10$ ICs, $T_{\mathrm{test}} = 150$.}
\vspace{-0.2cm}
\label{tab:duffing:poly}
\setlength{\tabcolsep}{4pt}
\footnotesize
\begin{tabular}{@{}l ccc c@{}}
\toprule
& \multicolumn{3}{c}{MAE $\|\hat{x} - x^*\|$ (multi-step)}
& Op.\ norm \\
\cmidrule(lr){2-4}\cmidrule(lr){5-5}
Form & median & mean & worst & $\|\cdot\|_F$ \\
\midrule
Linear   & 1.06 & 1.16 & 2.11 & 4.09 \\
Bilinear & 0.52 & 0.66 & 1.62 & 4.26 \\
GeKo     & 0.55 & 0.83 & 2.98 & 4.78 \\
KCF      & 0.74 & 0.74 & 1.05 & 4.31 \\
\bottomrule
\end{tabular}
\end{table}

\begin{example}[DC motor with input nonlinearity]
On the DC motor the input enters through the non-monotone
nonlinearity $f(u) = 2\tanh(u\cos u)$, which violates the
control-affine assumption underlying the Linear and Bilinear forms.
Both fail to capture the dynamics (median error $\approx 0.27$ in
scaled coordinates), confirming the benchmark observation of
\cite[Fig.~3]{HaseCort26:control:family}: a larger operator does not
help Bilinear here, since the deficiency is structural. The
input-lifted forms recover the dynamics, reducing the error by more
than an order of magnitude, i.e., GeKo and KCF reach median errors of
$0.011$ and $0.012$ and, 
unlike in the Duffing oscillator example,
the errors are
uniform across initial conditions, with no tail (worst-case $0.020$
and $0.024$). Here GeKo and KCF are essentially tied, and the
coupling budget is immaterial: every nonzero budget in the sweep
reaches the same $\sim\!10^{-2}$ noise floor of the polynomial state
dictionary, so coupling the input to $6$ of the $10$ features is
already sufficient to capture the input nonlinearity. KCF attains
this with a slightly smaller operator norm than the affine Bilinear form
($2.84$ vs.\ $4.07$) despite modeling a genuinely input-nonlinear
map.\looseness=-1 \hfill$\blacksquare$
\end{example}

\begin{figure}[tb]
\centering
\includegraphics[width=0.75\columnwidth]{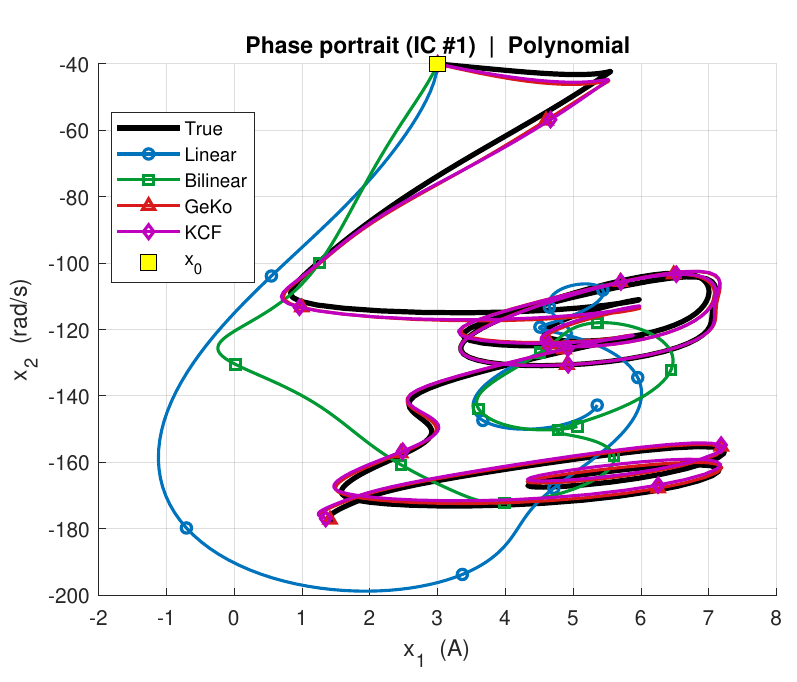}
\vspace{-0.35cm}
\caption{Phase trajectory plot, DC motor.}
\label{fig:dc:phase}
\end{figure}

\begin{figure}[tb]
\centering
\includegraphics[width=0.8\columnwidth]{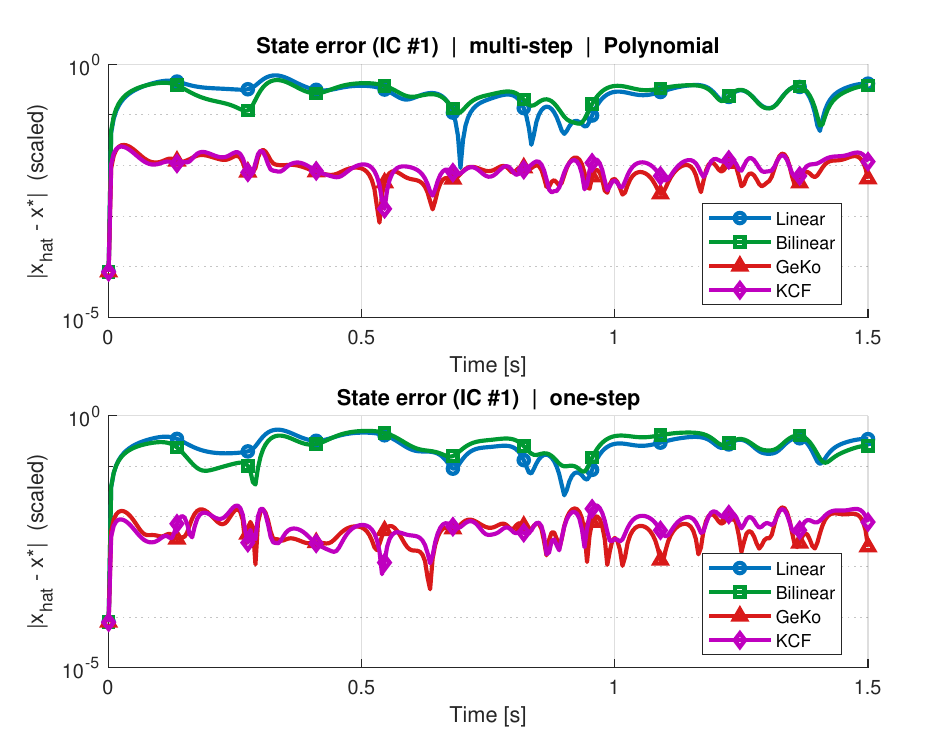}
\vspace{-0.35cm}
\caption{State error rollout, DC motor.}
\label{fig:dc:error}
\end{figure}

\begin{table}[htb]
\centering
\caption{DC motor with non-monotone input nonlinearity. 
KCF couples the input to $6$ of the $10$ lift features ($s = 58$).
Multi-step open-loop error in scaled state coordinates over
$N_{\mathrm{test}} = 10$ ICs, $T_{\mathrm{test}} = 300$.}
\vspace{-0.2cm}
\label{tab:dcmotor:results}
\setlength{\tabcolsep}{4pt}
\footnotesize
\begin{tabular}{@{}l ccc c@{}}
\toprule
& \multicolumn{3}{c}{MAE $\|\hat{x} - x^*\|$ (multi-step)}
& Op.\ norm \\
\cmidrule(lr){2-4}\cmidrule(lr){5-5}
Form & median & mean & worst & $\|\cdot\|_F$ \\
\midrule
Linear   & 0.276  & 0.292  & 0.425  & 3.20 \\
Bilinear & 0.272  & 0.288  & 0.316  & 4.07 \\
GeKo     & 0.0106 & 0.0114 & 0.0195 & 2.77 \\
KCF      & 0.0122 & 0.0137 & 0.0242 & 2.84 \\
\bottomrule
\end{tabular}
\end{table}

\begin{remark}[KCF versus GeKo subtleties]
\label{rem:kcf:vs:geko}
The Koopman control family (KCF) and the generalized Koopman (GeKo)
operator are structurally different routes to a Koopman description of
control systems, yet both arrive at a similar finite-dimensional
\emph{bilinear} form, i.e., \eqref{eq:EDMDc:KCF:predictor} and
\eqref{eq:GeKo:bilinear}, respectively,  despite originating from
(families of) operators acting on different function spaces. What separates them
theoretically is the underlying Koopman-invariance requirement: KCF
presumes a \emph{common} invariant subspace shared by the whole
input-parameterized family $\{\mathcal{K}_u\}_{u \in \mathbb{U}}$
\cite{HaseCort26:control:family}, whereas GeKo works with the single
composition operator on the product space $\mathcal{F}_x \otimes
\mathcal{F}_u$ under a relaxed invariance condition
\cite{lazar2026product}.

Once a finite-dimensional parameterization is fixed, learning either
operator from data is a regression problem that may be solved with
\emph{full} or \emph{sparse} input-state coupling; the sparsity pattern
can be guided by different scores, e.g.\ the multi-step roll-out
prediction error, the consistency index of \cite{HaseCort26:control:family} (see \Cref{rem:consistency_index}), or plain sparse least
squares. Throughout this tutorial paper we have adopted a fixed-order (tunable by the user) sparse KCF implementation to contrast with the full-coupling GeKo implementation and observe different performances. Via allowing full coupling  for KCF or imposing sparsity for GeKo, the two finite-dimensional Koopman forms can be rendered structurally identical. The GeKo operator yields a coordinate representation that is bounded on $\ell^2$ Riesz-basis,
which
enables Galerkin (finite-section) convergence analysis, a useful starting point for deriving tight approximation-error bounds \cite{lazar2026product}. Sharp, closed-form computable error bounds for finite-dimensional KCF models can be found in~\cite{MH-JC-JWB:26-cdc,HaseCort26:control:family}.
\looseness=-1
\end{remark}

\mysubsec{Koopman observers}{\secvdpages}
\label{subsec:observer}
While Sections~\ref{subsec:LTI}--\ref{subsec:EDMD:control} focus on
the use of Koopman models for prediction and control, the same
operator-theoretic viewpoint enables \emph{state estimation} for
nonlinear systems via linear infinite-dimensional observers. Starting
from the autonomous Koopman semigroup $(\mathcal{K}^t)_{t \ge 0}$ of
\Cref{sec:KOT} with output $y(t) = h(x(t; x_0))$, a natural idea
is to design a Luenberger-type observer on the lifted space and
recover the physical state by reading off the principal eigenfunction
coordinates.

This idea was made precise for discrete-time autonomous systems in
\cite{SuraBana16:observer, Sura16:observer} through the
\emph{Koopman observer form} (KOF). One selects a finite set of
Koopman eigenfunctions $\{\eigfun_i\}_{i\in[n]}$ whose span contains the
state and output observables, so that, by the eigenfunction property
$\mathcal{K}\eigfun_i = \lambda_i \eigfun_i$, the lifted coordinates
$\xi = (\eigfun_1, \dots, \eigfun_n)^\top$ evolve linearly,
\begin{equation}\label{eq:KOF}
    \xi^+ = \Lambda\, \xi, \qquad
    x = M_x\, \xi, \qquad
    y = M_y\, \xi,
\end{equation}
with $\Lambda = \mathrm{diag}(\lambda_1, \dots, \lambda_n)$ and
constant Koopman-mode matrices $M_x \in \mathbb{R}^{d \times n}$,
$M_y$ mapping the eigenfunction coordinates back to the state and
output. Equation~\eqref{eq:KOF} is a linear time-invariant system in
$\xi$, so a Luenberger/Kalman observer can be placed directly on it,
$\hat\xi^+ = \Lambda \hat\xi + L\,(y - M_y \hat\xi)$, and the state
estimate recovered as $\hat{x} = M_x \hat\xi$. The estimation error
$\xi - \hat\xi$ is governed by $\Lambda - L M_y$, so the gain $L$
follows from standard linear design; convergence holds under an
observability condition expressed directly in terms of the chosen
Koopman eigenvalues $\lambda_i$ and modes $M_y$, giving the first
link between the Koopman spectrum and observer synthesis. The
construction is exact wherever the selected eigenfunctions span the
state and output, often a large part of the basin of attraction or
the whole state space. Observability and detectability in this
framework have since been characterized in terms of eigenfunction
symmetries \cite{MesbBu21:observability}, of the Koopman Gramian
\cite{YeunLiu17:gramian}, and through data-driven constructions
\cite{DahdForb24:observer}.

A central limitation of these approaches is the standing assumption
that the output lies exactly in a finite-dimensional invariant
subspace, which is rarely satisfied and is in general only
approximately true. The recent work in
\cite{MoheMaurWinki25:dual} removes this assumption by reformulating
the estimation problem in the \emph{dual} Koopman system: rather than
propagating observables forward by $\mathcal{K}^t$, one propagates
reproducing-kernel functions $k_{x_0} \in \mathbb{H}$ backward by
$(\mathcal{K}^t)^\ast$, with $\mathbb{H}$ a reproducing kernel Hilbert space
(specifically the Hardy space on the polydisc $\mathbb{D}^d$ in the
setting of \cite{MoheMaurWinki25:dual}). The state estimate is then
recovered via the inner product with the coordinate functions
$p_k(x) = x_k$,\looseness=-1
\begin{equation*}
    \hat{x}_k(t) = \innerprod{p_k}{\hat{f}(t)},
    \qquad k = 1, \dots, d.
\end{equation*}
Under standard non-resonance and stability conditions at a
hyperbolic equilibrium, \cite{MoheMaurWinki25:dual} establishes that
(i) pointwise approximate observability of the dual Koopman system
is equivalent to observability of the underlying nonlinear system,
and (ii) a Luenberger observer placed on the $\beta$-unstable
spectral block of $\mathcal{L}^\ast$ delivers exponential
convergence of the estimation error at an arbitrarily prescribed
rate $\beta < 0$. The design reduces to a finite-dimensional pole
placement on the $N_\beta$ slowest eigenvalues of the truncated
operator matrix, exactly the setting where standard linear tools
apply.

\mysec{Controller design}{\secvipages}
\label{sec:controller}

In this section, we recap key concepts for controller design using Koopman operator theory with closed-loop guarantees, while referring to the recent overview article~\cite{StraWort26} and the references therein for a more detailed elaboration. 
Hereby, we focus on robust controller design and model predictive control, where data-driven surrogate models governed by
\begin{equation}\label{eq:controlled-dynamics:F-eps}
    x^+ = F^\varepsilon(x,u)
\end{equation}
are used, which satisfy a proportional error bound w.r.t.\ the original system dynamics $x^+ = F(x,u)$ given by~\eqref{eq:controlled-dynamics:DT} (\textit{ground truth}), i.e., 
\begin{equation*}
    \| F(x,u) - F^\varepsilon(x,u) \| \leq c_x \| x \| + c_u \| u \|
\end{equation*}
on convex and compact sets $\mathbb{X} \subseteq \mathbb{R}^d$ and $\mathbb{U} \subseteq \mathbb{R}^m$ containing the origins in their interiors, see, e.g., \cite{strasser:schaller:worthmann:berberich:allgower:2025tacon} and~\cite{schimperna2025data} for kernel-based Koopman approximants.
To this end, we briefly look at Lyapunov arguments as typically used in the stability analysis of autonomous dynamical systems, which will serve as a blueprint for the upcoming developments, see~\cite[Section~3]{BoldPhil25} for details. 
The decisive condition to infer asymptotic 
stability is a Lyapunov decrease: we consider the function~$V: \mathbb{R}^d \rightarrow \mathbb{R}$ satisfying
\begin{equation}\nonumber
    \alpha_1(\|x\|)
    \leq V(x) \leq 
    \alpha_2(\|x\|)
\end{equation}
for $\mathcal{K}_\infty$-functions $\alpha_1,\alpha_2$.\footnote{A continuous function $\alpha: [0,\infty) \rightarrow [0,\infty)$ is said to be of class $\mathcal{K}_\infty$, if $\alpha$ is strictly monotonically increasing, unbounded, and zero at zero.} 
Moreover, we assume that $V$ satisfies the decrease condition
\begin{equation}\label{eq:Lyapunov:F-eps}
    V(F^\varepsilon(x,\mu^\varepsilon(x)) \leq V(x) - \alpha_3(\| x \|)
\end{equation}
w.r.t.\ the closed-loop dynamics $x^+ \!=\! F^\varepsilon_{\mu^\varepsilon}(x) \!:=\! F^\varepsilon(x,\mu^\varepsilon(x))$
resulting from applying the static state feedback law $\mu\!:\! \mathbb{X} \!\rightarrow\! \mathbb{U}$ to the (data-driven) surrogate model~\eqref{eq:controlled-dynamics:F-eps}. 
However, we require a decrease w.r.t.\ the closed-loop dynamics
\begin{equation}\label{eq:cl-dynamics:DT}
    x^+ = F_{\mu^\varepsilon}(x) := F(x,\mu^\varepsilon(x)),
\end{equation}
i.e., the dynamics of the original nonlinear system~\eqref{eq:controlled-dynamics:DT} controlled by the feedback law $u = \mu^\varepsilon(x)$ designed for the surrogate model~\eqref{eq:controlled-dynamics:F-eps}.
Hence, making use of the Lyapunov decrease~\eqref{eq:Lyapunov:F-eps}, we estimate
\begin{align*}\nonumber 
    V(F_{\mu^\varepsilon}(x)) =\ & V(F^\varepsilon_{\mu^\varepsilon}(x)) + \Big(V(F_{\mu^\varepsilon}(x)) - V(F^\varepsilon_{\mu^\varepsilon}(x)) \Big)  \\ 
    \leq\ & V(x) - \alpha_3(\|x\|) + |V(F^\varepsilon_{\mu^\varepsilon}(x)) - V(F_{\mu^\varepsilon}(x))| \nonumber
\end{align*}
In conclusion, the ramification $|V(F^\varepsilon_{\mu^\varepsilon}(x)) - V(F_{\mu^\varepsilon}(x))|$ of the approximation error measured w.r.t.\ the Lyapunov function~$V$ has to be compensated by the Lyapunov decrease, i.e., $\alpha_3(\|x\|)$. Let, e.g., the inequality
\[
    \rho \cdot
    \alpha_3(\|x\|) \geq |V(F^\varepsilon_{\mu^\varepsilon}(x)) - V(F_{\mu^\varepsilon}(x))| \qquad\forall\,x \in \mathbb{X}
\]
with some factor $\rho \in (0,1)$ hold, then we may infer asymptotic stability of the origin w.r.t.\ the closed-loop dynamics~\eqref{eq:cl-dynamics:DT}, see e.g., \cite{BoldPhil25} 
for autonomous systems and~\cite{schimperna2025data} for controlled systems.\looseness=-1

\subsection{Closed-form control laws}
\label{subsec:closed-form_control}

There exists a plethora of Koopman-based control schemes, see, e.g., the recent overviews~\cite{duran2025control,StraWort26} and the references therein. 
Thus, we only discuss a few representative works in this subsection to illustrate the broad applicability of Koopman-based models for controller design, but focus our more detailed explanations on Koopman MPC in the subsequent subsections. LTI models as discussed in \Cref{subsec:LTI} lend themselves naturally to linear design techniques such as LQR~\cite{bevanda2022towards,gibson2022koopman}. 
Bilinear surrogate models, as discussed in Sections \ref{subsec:control:family} and \ref{subsec:product:Hilbert}, can be used in combination with robust control using linear matrix inequalities (LMIs) and semi-definite programming (SDP) for controller design~\cite{strasser:schaller:worthmann:berberich:allgower:2025tacon,StraScha26:SafEDMD,strasser:berberich:schaller:worthmann:allgower:2025at}. 
Leveraging extensions of the learning error analysis discussed in \Cref{subsec:EDMD_error_full}, these approaches enable the derivation of rigorous closed-loop guarantees, see also the recent overview~\cite{StraWort26}. 
Recent work highlights the impact of the dimensionality of bilinear surrogate models for control design underlining the significance of low-dimensional Koopman-invariant models \cite{hanna26:existence}. 
Alternative approaches rely on the integration of feedback linearization into learning a Koopman operator approximation \cite{gadginmath2024data}, 
contraction-theoretic approaches~\cite{yi2023equivalence} or approximately solving the stochastic Hamilton-Jacobi-Bellman equations~\cite{bevanda2025kernel,hoischen2024data}, see also~\cite{vaidya2025koopman} for results for deterministic systems linking HJB and Koopman theory.

\subsection{Model predictive control (MPC)}
\label{subsec:MPC:theory}

We focus on the prototypical problem of set-point stabilization, which we shift w.l.o.g.\ to the origin, i.e., $F(0,0) = 0$, where $(0,0)$ is assumed to be contained in the interior of the control and state constraint sets $\mathbb{U} \subseteq \mathbb{R}^m$ and $\mathbb{X} \subseteq \mathbb{R}^d$.
Then, using the (data-driven) surrogate model~\eqref{eq:controlled-dynamics:F-eps} in the optimization step, the MPC scheme is provided in \Cref{alg:MPC}.\looseness=-1 
\begin{algorithm}[htb]
	\caption{Model Predictive Control Scheme}
	\begin{algorithmic}[1]
		\label{alg:MPC}
        \STATE \textbf{Initialisation}. Set prediction horizon~$H$, $H \in \mathbb{N}_{\geq 2}$
        \FOR{$k=0,\ldots,\infty$}
        \STATE Measure (or estimate)  
        current state $\hat{x} := x(k)$
        \STATE Minimize 
        cost functional 
            \[
                J_H(\hat{x},\bar{\mathbf{u}}) = \sum\nolimits_{i=0}^{H-1} \ell(\bar{x}_{\mathbf{u}}(i;\hat{x}),\bar{u}(i)) + V_f(\bar{x}_{\mathbf{u}}(H;\hat{x}))
            \]
            w.r.t.\ $\bar{\mathbf{u}} = (\bar{u}(i))_{i=0}^{H\!-\!1}$ subject to
			\begin{itemize}
				\item [$\bullet$] \hspace{-0.4cm} 
                dynamics~\eqref{eq:controlled-dynamics:F-eps} with initial value $\bar{x}_{\mathbf{u}}(0;\hat{x}) = \hat{x}$
                    \[
                        \bar{x}_{\mathbf{u}}(i+1;\hat{x}) = F^\varepsilon(\bar{x}_{\mathbf{u}}(i;\hat{x}),\bar{u}(i))
                    \]
				\item [$\bullet$] \hspace{-0.4cm} state constraints $\bar{x}(i+1) \in \mathbb{X}$, $i \in \{1,\ldots,H\}$
                \item [$\bullet$] \hspace{-0.4cm} control constraints $\bar{u}(i) \in \mathbb{U}$, $i \in \{0,1,\ldots,H-1\}$
			\end{itemize}	
			to compute  
            optimal control sequence $\bar{\mathbf{u}}^\star = (\bar{u}^\star(i))_{i=0}^{H\!-\!1}$%
        \STATE Apply first input 
        $\bar{u}^\star(0) \in \mathbb{U}$ at
        plant governed by~\eqref{eq:controlled-dynamics:DT}
        \ENDFOR
	\end{algorithmic}
\end{algorithm}

\Cref{alg:MPC} yields an implicitly defined static \textit{state feedback law} $\mu = \mu_{H}^\varepsilon: \mathbb{X} \rightarrow \mathbb{U}$ by $\mu_H^\varepsilon(x(k)) := \bar{u}^\star(0)$, which depends on the horizon length~$H$ and the employed model~\eqref{eq:controlled-dynamics:F-eps}. The resulting MPC closed loop is governed by the dynamics
\begin{equation}\label{eq:dynamics:MPC}
    \!x_{\mu_H^\varepsilon}(k+1) = F(x_{\mu_H^\varepsilon}(k),\mu_H^\varepsilon(x_{\mu_H^\varepsilon}(k))),~~ x_{\mu_H^\varepsilon}(0) = \hat{x}.\! 
\end{equation}
Further, we define the (optimal) value function $V_H: \mathbb{X} \rightarrow [0,\infty]$ of the optimal control problem solved in each MPC iteration (Step~4 of \Cref{alg:MPC}) by
\begin{equation*}
    V_H^\varepsilon := \inf\nolimits_{\textbf{u} \in \mathcal{U}_H^\varepsilon(x)} J_H^\varepsilon(x,\mathbf{u}),
\end{equation*}
where the (state-dependent) set~$\mathcal{U}_H^\varepsilon(x)$ denotes the set of \textit{admissible} control sequences, i.e., control sequences $\mathbf{u} \in \mathbb{U}^{H}$ such that the state trajectory emanating from the initial value~$x$ and governed by the dynamics~\eqref{eq:controlled-dynamics:F-eps} satisfies the state constraints, i.e., $x_{\mathbf{u}}(k;x) \in \mathbb{X}$ for all $k \in \{0,1,\ldots,H\}$. 
Moreover, we define the infinite-horizon (closed-loop) costs
\[
    J_\infty^{\mu^\varepsilon_H}(\hat{x}) = \sum\nolimits_{k=0}^\infty \ell(x_{\mu_H^\varepsilon}(k;\hat{x}),\mu_H^\varepsilon(x_{\mu_H^\varepsilon}(k;\hat{x}))).
\]
Note that the optimal control problem in Step~4 of \Cref{alg:MPC} is solved based on the surrogate model~\eqref{eq:controlled-dynamics:F-eps}, while the computed input $\mu_H^\varepsilon(x(k))$ at time $k \in \mathbb{N}_0$ is applied at the plant, which is governed by the original dynamics~\eqref{eq:controlled-dynamics:DT}.
Hence, some robustness to model-plant mismatch is indispensable to successfully apply Koopman MPC, see~\cite{allan2017inherent} and \cite{BoldGrun25,schimperna2025data,moldenhauer2026discounted} as well as~\cite{DeJongCDC2024,  kuntz_beyond_2026,schimperna2025stability} for Koopman MPC with terminal conditions.

A prevalent choice are quadratic stage cost
\begin{equation}\label{eq:stage-cost:quadratic}
    \ell(x,u) := \| x \|_Q^2 + \| u \|_R^2 := x^\top Q x + u^\top R u
\end{equation}
with positive definite matrices $Q \in \mathbb{R}^{d \times d}$ and $R \in \mathbb{R}^{m \times m}$ in combination with linear surrogate models~$F^\varepsilon$ given by~\eqref{eq:controlled-dynamics:F-eps}, see, e.g., \cite{KordMezi18} and the many follow-up works. 
Using linear surrogate models is particularly attractive due to rendering the optimal control problem in Step~4 of \Cref{alg:MPC} computationally tractable even for large-scale models~$F^\varepsilon$ as shown, e.g., in~\cite{WuTan2026:BraatzDrgona}.
In addition, tailored parametric Koopman decompositions may be leveraged to increase the computational efficiency.

For the stability analysis of the set point w.r.t.\ the MPC closed-loop dynamics~\eqref{eq:dynamics:MPC}, the relaxed Lyapunov inequality
\begin{equation*}
    V_H^\varepsilon(F(x,\mu_H^\varepsilon(x))) \leq V_H^\varepsilon(x) - \alpha \ell(x,\mu_H^\varepsilon(x)) 
\end{equation*}
with $\alpha = \alpha_H^\varepsilon \in (0,1]$ provides a verifiable sufficient condition for asymptotic and/or exponential stability of the MPC closed loop, see, e.g., \cite{GrunRant08,grune_nonlinear_2017} and~\cite{lincoln2006relaxing} for relaxed dynamic programming.
To this end, either MPC with terminal conditions or MPC without (stabilizing) terminal conditions
are used. 

\textit{Terminal conditions} consist of a terminal region $\mathbb{X}_f$, $\mathbb{X}_f \subseteq \mathbb{X}$, and a suitably designed terminal cost $V_f:\mathbb{X}_f \rightarrow \mathbb{R}_{\geq 0}$ and controller $\mu_f: \mathbb{X}_f \rightarrow \mathbb{U}$ such that the following two properties hold for all $x \in \mathbb{X}_f$:
\begin{itemize}
    \item forward invariance of the terminal region, i.e., 
        $
            F(x,\mu_f(x)) \in \mathbb{X}_f$ and $\mu_f(x) \in \mathbb{U} 
        $
    \item decrease condition w.r.t.\ the stage cost, i.e., 
        $
            V_f(F(x,\mu_f(x))) \leq V_f(x) - \ell(x,\mu_f(x)) 
        $
\end{itemize}
Terminal conditions can be easily constructed if the linearization at the desired set point is controllable, see, e.g., the textbook~\cite{RawlMayn17}.
If the linearization is not stabilizable, non-quadratic stage or terminal costs may be necessary, see~\cite{MullWort17:quadratic}.

If no (stabilizing) terminal conditions are imposed, the combination of some controllability property and a sufficiently long horizon~$H$ is typically required.
Here, we resort to \textit{cost controllability} as proposed in~\cite{grimm2005model,GrunPann10}, see also the recent paper~\cite{CoroGrun20} and the references therein.
Cost controllability corresponds to the existence of a sequence $(B_k)_{k=2}^\infty$ satisfying\looseness=-1
\begin{equation}\label{eq:growth-bound}
    V_k(\hat{x}) :=\inf_{\mathbf{u} \in \mathcal{U}_k(\hat{x})} J_k(\hat{x},\mathbf{u}) \leq B_k \inf_{\mathbf{u} \in \mathcal{U}_1^\varepsilon(\hat{x})} \ell(\hat{x},u(0)) 
\end{equation}
for all $k \in \{2,\ldots,H
\}$ with degree of suboptimality 
\[
    \alpha_{H}
    := 1 - \tfrac {(B_{H}
    -1) \prod_{k=2}^{H
    } (B_k-1)}{\prod_{k=2}^{H} 
    B_k - \prod_{k=2}^{H}
    (B_k-1)}
\]
Note that the sufficient stability condition $\alpha_H \in (0,1]$ can always be ensured for a \textit{sufficiently long horizon~$H$
} if the system is cost controllable such that the condition $\limsup_{H
\rightarrow \infty} B_{H}
/H 
< 1$ holds, see~\cite{MullWort17:quadratic}. Moreover, it implies suboptimality of the MPC closed-loop cost $J_\infty^{\mu_H^\varepsilon}$ in comparison to infinite-horizon optimal costs, see also \cite{BoccGrun14} for the incorporation of state constraints.

In the following, we use quadratic stage cost~\eqref{eq:stage-cost:quadratic} for ease of exposition and refer to~\cite{grimm2005model,worthmann2011stability,CoroGrun20} for MPC results referring to more general stage costs. 
We provide conditions on the Koopman-based surrogate model, under which it can be rigorously shown that the origin is exponentially stable despite model-plant mismatch. 
Indeed, these conditions are compatible for MPC with and without (stabilizing) terminal conditions, see~\cite{schimperna2025data,schimperna2025stability}, see also~\cite{ShanCort25:MPC}, where similar results are shown under the assumption that the linear Koopman model is either exact or satisfies a proportional error bound.
So far, one advantage of bilinear Koopman models as proposed in~\cite{strasser:schaller:berberich:worthmann:allgower:2025csl} using approximation results derived in~\cite{BoldPhil25} is that such a proportional error bound can be ensured using kernel EDMD approximants using the concept of Koopman control family as elaborated in Subsection~\ref{subsec:control:family}. 
Then, stability of the MPC closed loop can be proved by, e.g., first showing that the growth bound~\eqref{eq:growth-bound} is preserved for all $k \in \{2,3,\ldots,\bar{H
}\}$, where $\bar{H
}$ can be chosen arbitrary, but finite (for a sufficiently high approximation accuracy). 
Then, a relaxed Lyapunov inequality can be established as rigorously shown in~\cite{schimperna2025data}, see also~\cite{schimperna2025stability} for comparable results w.r.t.\ terminal conditions. Finally, leveraging stability, cost controllability can be fully recovered, i.e., the growth bound~\eqref{eq:growth-bound} also holds for $k \in \mathbb{N}_{\geq \bar{H
}}$, see~\cite{schimperna2025data,moldenhauer2026discounted}.

A topic of particular importance in Koopman-based MPC is the learning of suitable subspaces~\cite{mamakoukas2023learning} leveraging, e.g., Koopman eigenfunctions~\cite{KordMezi20} and the concept of invariance proximity discussed in \Cref{subsec:Koopman-invariance}. 
MPC using bilinear surrogate models is, e.g., successfully used in robotics~\cite{folkestad2021koopman}, see also~\cite{zhao2024deep} for recent results on deep Koopman MPC.

\begin{remark}\longthmtitle{Input-output data}
    Recently, extensions to Koopman models using only input-output data were proposed, see, e.g., the recent preprints~\cite{iacob:szecsi:mate:beintema:schoukens:toth:2025} on Koopman modeling and~\cite{StraBerb26:Koopman:output} for Koopman control with closed-loop guarantees. Such input-output models may also be leveraged within the presented MPC theory to infer exponential stability of the closed loop, see, e.g., \cite{BoldSchi26}. 
    We also want to point out recent work~\cite{deutscher2024data,deutscher2025koopman} on data-driven robust output regulation for systems governed by partial differential equations.\looseness=-1
\end{remark}

Furthermore, we clearly acknowledge that this subsection is not close to being complete and does not provide a full overview about the current state of the art, e.g., tube-based MPC or Koopman-based approaches related to the various notions of input-to-state stability (ISS) are completely missing.  
While these research directions are highly interesting, they are also out of scope within our tutorial-style exposition and should be treated (in detail) in a potential extended version in the hopefully not-so-far future.

\subsection{A hands-on tutorial on Koopman MPC}
\label{subsec:MPC:tutorial}
We deploy the four EDMDc models of \Cref{subsec:EDMDc:implementation}
as prediction models within an MPC controller. The MPC prediction
horizon is denoted $H\!\leq\! H_t$ in this subsection, to make a distinction  with the training horizon $H_t$ in \Cref{subsec:EDMDc:implementation}.
At each sampling instant $t$,
given the current state $x_t$ and the previous applied input
$u_{t-1}$, we lift $z_t \!=\! \basisfun^x(x_t)$ and solve \looseness=-1
\begin{subequations}\label{eq:MPC:ocp}
\begin{align}
\!\min_{\mathbf{u}_t} \;\; & J_t(\mathbf{u}_t) =
    \sum_{k=0}^{H-1} \ell\bigl(z_{k|t}, u_{k|t}, u_{k-1|t}\bigr)
    + \ell_f\bigl(z_{H|t}\bigr)\! \label{eq:MPC:cost} \\
\text{s.t.}\;\;
& z_{0|t} = \basisfun^x(x_t), \label{eq:MPC:ic} \\
& z_{k+1|t} = F_{\mathrm{Koop}}(z_{k|t}, u_{k|t}),
    ~~ k = 0, \dots, H\!-\!1,\!\! \label{eq:MPC:dyn} \\
& \|u_{k|t}\|_\infty \le u_{\max}, \quad k = 0, \dots, H-1,
    \label{eq:MPC:box}
\end{align}
\end{subequations}
with $\mathbf{u}_t := (u_{0|t}, \dots, u_{H-1|t})$, stage and
terminal costs
\begin{align}
\ell(z, u, u^-) &= (z - z_{\mathrm{ref}})^{\!\top} Q_z (z - z_{\mathrm{ref}})
    + u^\top R_u u \nonumber\\&+ (u - u^-)^\top R_{\Delta u} (u - u^-), \label{eq:MPC:stage} \\
\ell_f(z) &= (z - z_{\mathrm{ref}})^{\!\top} Q_z (z - z_{\mathrm{ref}}),
    \label{eq:MPC:terminal}
\end{align}
and $F_{\mathrm{Koop}}$ instantiated as one of
\eqref{eq:GeKo:bilinear},
\eqref{eq:EDMDc:predictor},
\eqref{eq:EDMDc:bilinear:predictor},
\eqref{eq:EDMDc:KCF:predictor}. The lifted-state weight is pulled
back from a physical-coordinate penalty $Q_x$ via the linear decoder
$D$ in \eqref{eq:decoder} of \Cref{subsec:EDMDc:implementation},
\begin{equation*}
    Q_z = D^{\!\top} Q_x D,
    \qquad
    z_{\mathrm{ref}} = \basisfun^x(x_{\mathrm{ref}}),
\end{equation*}
so that $(z - z_{\mathrm{ref}})^{\!\top} Q_z (z - z_{\mathrm{ref}})
= \|D(z - z_{\mathrm{ref}})\|_{Q_x}^2$ approximates the physical
tracking error in the lifted space. The closed-loop law applies
$u_t = u_{0|t}^\star$ to the plant, advances to $t+1$, and re-solves
\eqref{eq:MPC:ocp}.

\noindent\textbf{Numerical implementation.}
We use a \emph{single-shooting} formulation: the only decision
variables are the input sequence $\mathbf{u}_t := (u_{0|t}, \dots,
u_{H-1|t}) \in \mathbb{R}^{m H}$; the lifted-state sequence is
generated on the fly by rolling the predictor forward. The dynamics
constraint \eqref{eq:MPC:dyn} thus disappears from the optimizer and
\eqref{eq:MPC:ocp} reduces to the dense box-constrained program\looseness=-1
\begin{equation*}
    \mathbf{u}_t^\star
    = \arg\min_{\mathbf{u}_t \in \mathbb{U}^{H}} \; J_t(\mathbf{u}_t),
    \,\,
    \mathbb{U}:= \{\,u \in \mathbb{R}^m : \|u\|_\infty \le u_{\max}\,\},
\end{equation*}
which we solve with \texttt{fmincon}/SQP, warm-started by the shifted
previous solution $(u_{1|t-1}^\star, \dots, u_{H-1|t-1}^\star, 0)$.
Each cost evaluation requires the predicted trajectory and, for SQP
to converge quickly, an accurate gradient $\nabla_{\mathbf{u}_t} J_t
\in \mathbb{R}^{m H}$. We compute the gradient analytically by
\emph{adjoint backpropagation} in four steps.

\emph{Step 1 -- Forward pass.} Starting from $z_{0|t} = \basisfun^x(x_t)$,
propagate the predictor for $k = 0, \dots, H - 1$:
\begin{equation*}
    z_{k+1|t} = F_{\mathrm{Koop}}(z_{k|t}, u_{k|t}),
\end{equation*}
storing each $z_{k|t}$ and the Jacobians
\begin{equation*}
\begin{split}
    A_k &:= \tfrac{\partial F_{\mathrm{Koop}}}{\partial z}\big|_{(z_{k|t}, u_{k|t})}
    \in \mathbb{R}^{n \times n}, \\
    B_k &:= \tfrac{\partial F_{\mathrm{Koop}}}{\partial u}\big|_{(z_{k|t}, u_{k|t})}
    \in \mathbb{R}^{n \times m},
    \end{split}
\end{equation*}
both available in closed form for all four predictors. Accumulate
$J_t$ from \eqref{eq:MPC:cost} during the same sweep. 

\emph{Step 2 -- Adjoint backward pass.} Define the adjoint vector
$\lambda_k := \partial J_t / \partial z_{k|t} \in \mathbb{R}^{n}$.
Initialize at the terminal stage and propagate backwards, storing
each $\lambda_k$:
\begin{equation*}
\begin{aligned}
    \lambda_{H} &= 2 Q_z (z_{H|t} - z_{\mathrm{ref}}), \\
    \lambda_k &= 2 Q_z (z_{k|t} - z_{\mathrm{ref}})
        + A_k^{\!\top} \lambda_{k+1},
    \qquad k = H - 1, \dots, 1.
\end{aligned}
\end{equation*}

\emph{Step 3 -- Input gradient.} With all $\lambda_k$ now available,
assemble the gradient block for each $u_{k|t}$:
\begin{equation*}
    \tfrac{\partial J_t}{\partial u_{k|t}}
    = 2 R_u u_{k|t} + 2 R_{\Delta u}\, \delta_k    + B_k^{\!\top} \lambda_{k+1},
\end{equation*}
where $\delta_k \in \mathbb{R}^m$ captures the $\Delta u$-coupling
between $u_{k|t}$ and its neighbours,
\begin{equation*}
    \delta_k = (u_{k|t} - u_{k-1|t}) - (u_{k+1|t} - u_{k|t}),
\end{equation*}
with the convention $u_{-1|t} = u_{t-1}$ (previous applied input)
and the second term omitted when $k = H - 1$.

\emph{Step 4 -- Receding-horizon step.} Pass $J_t$ and
$\nabla_{\mathbf{u}_t} J_t = (\partial J_t / \partial u_{0|t}, \dots,
\partial J_t / \partial u_{H-1|t})$ to \texttt{fmincon}, retrieve
$\mathbf{u}_t^\star$, apply $u_t = u_{0|t}^\star$ to the plant,
shift the optimizer's warm-start, and advance to $t + 1$.
\miniskip

\noindent\textbf{Output tracking variant.}
For the DC motor we track a scalar reference $r_t$ on the angular
velocity via the decoder $y_t = D z_t$. Steps 1--4 carry over with
the substitutions $Q_z \leftarrow D^{\!\top} Q_y D$ in both
\eqref{eq:MPC:stage}--\eqref{eq:MPC:terminal} and $z_{\mathrm{ref}}
\leftarrow D^{\dagger} r_k$ (time-varying through $k$). We report
the integral squared error $\mathrm{ISE} = T_s \sum_t \|x_t -
x_t^\star\|^2$ and the CPU time of the solution per-iteration (mean and
worst) in \Cref{tab:MPC:duffing,tab:MPC:dcmotor}.

\begin{table}[t]
\centering
\caption{Closed-loop MPC, forced Duffing oscillator (regulation to
origin, $H = 15$, $T_{\mathrm{mpc}} = 100$ steps): CPU time and tracking ISE for the 4 EDMDc forms
(multi-step training). (Runs 100/100).}
\vspace{-0.2cm}
\label{tab:MPC:duffing}
\setlength{\tabcolsep}{6pt}
\footnotesize
\begin{tabular}{@{}l ccc@{}}
\toprule
Form & ISE & mean CPU [ms] & worst CPU [ms] \\
\midrule
Linear   & 3.291 & 2.2 & 20.7 \\
Bilinear & 3.287 & 3.1 &  6.3 \\
GeKo     & 3.286 & 4.7 & 11.0 \\
KCF      & 3.286 & 8.2 & 29.4 \\
\bottomrule
\end{tabular}
\end{table}

\begin{figure}[t]
\centering
\includegraphics[width=0.8\columnwidth]{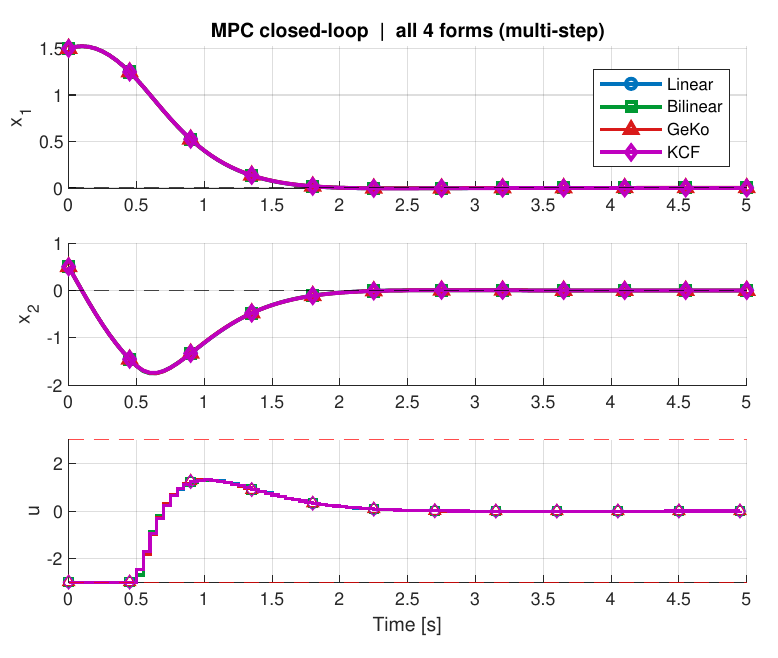}
\vspace{-0.35cm}
\caption{Closed-loop Koopman MPC results, Duffing oscillator.}
\label{fig:duf:MPC}
\end{figure}
\begin{example}[forced Duffing oscillator]
For the Duffing oscillator, all four EDMDc forms achieve nearly
identical tracking performance (ISE within $0.2\%$) -- expected,
since the plant is control-affine and any of the four predictors
captures the local dynamics well enough inside a horizon-$H$
receding window. The forms separate only in per-step computational
cost, which grows with the lifted dimension: Linear is cheapest
($2.2$\,ms mean), followed by Bilinear and GeKo ($3$--$5$\,ms),
while the input-lifted KCF is the most expensive ($8.2$\,ms). All
four remain well within the $T_s = 50$\,ms sampling rate, with
worst-case times under $30$\,ms. The closed-loop trajectories are
given in Figure~\ref{fig:duf:MPC}. \hfill$\blacksquare$
\end{example}

\begin{table}[t]
\centering
\caption{Closed-loop MPC, DC motor (tracking on $x_2$, $H = 20$, $T_{\mathrm{mpc}} = 600$ steps,
$T_s = 5$\,ms): CPU time and tracking ISE for
the 4 EDMDc forms (multi-step training). (Runs 600/600).}
\vspace{-0.2cm}
\label{tab:MPC:dcmotor}
\setlength{\tabcolsep}{6pt}
\footnotesize
\begin{tabular}{@{}l r r r@{}}
\toprule
Form & ISE & mean CPU [ms] & worst CPU [ms] \\
\midrule
Linear   & $58.7$        &  2.5 &   8.0 \\
Bilinear & $28{,}087$    &  2.6 &   6.1 \\
GeKo     & $32.4$        & 57.9 & 140.4 \\
KCF      & $32.2$        & 62.5 & 141.5 \\
\bottomrule
\end{tabular}
\end{table}

\begin{figure}[t]
\centering
\includegraphics[width=0.8\columnwidth]{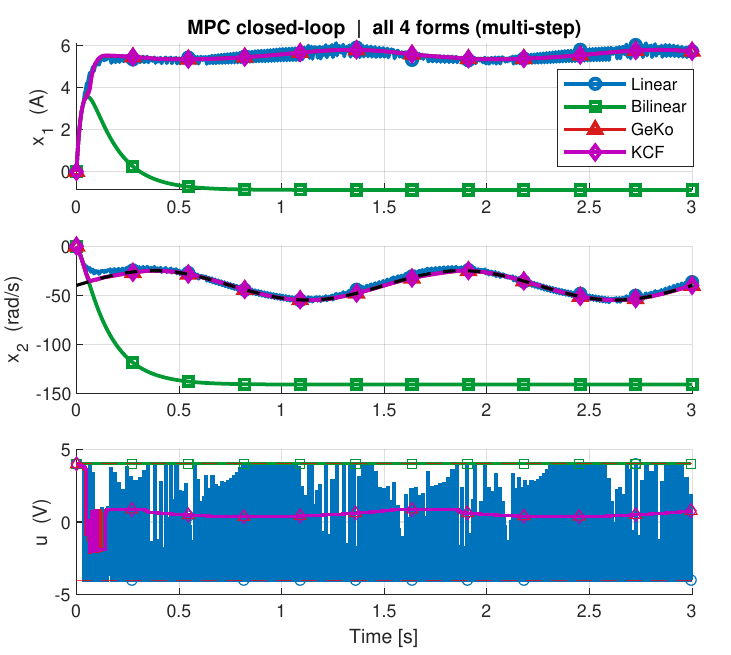}
\vspace{-0.35cm}
\caption{Closed-loop Koopman MPC results, DC motor.}
\label{fig:dc:MPC}
\end{figure}
\begin{example}[DC motor with input nonlinearity]
The picture inverts on the DC motor. Bilinear MPC blows up by nearly
three orders of magnitude in ISE ($28{,}087$ vs.\ $58.7$ for Linear)
because the non-monotone $f(u) = 2\tanh(u\cos u)$ 
breaks the control-affine
assumption underlying \eqref{eq:EDMDc:bilinear:predictor}, and the
receding-horizon loop amplifies that misfit. The lifted-input forms
GeKo and KCF make no such assumption and achieve the lowest tracking
error ($32.4$ and $32.2$, against $58.7$ for Linear). The binding
constraint is now computation: at $T_s = 5$\,ms, GeKo and KCF
average $58$\,ms and $63$\,ms per step, with worst cases near
$140$\,ms, i.e., $10$--$30\times$ over the sampling period, whereas
the affine forms stay in the few-millisecond range. The take-away
from \Cref{tab:MPC:duffing,tab:MPC:dcmotor} is that input lifting is
indispensable for closed-loop tracking when control-affineness
fails, but at the present \texttt{fmincon}/SQP cost precludes
hard real-time deployment at millisecond sampling rates; specialized
solvers would be
needed to close this gap. ~ \phantom{n} \hfill$\blacksquare$
\end{example}

\mysec{Koopman meets machine learning}{\secviipages}
\label{sec:ML}

The fields of Koopman operator methods and machine learning have had a growing mutual impact on each other in recent years. This section sheds light on some of the most significant of developments of these interactions. \Cref{subsec:stat_learn} provides a statistical learning perspective on Koopman operator learning. An alternative Bayesian view is presented in \Cref{subsec:Bayesian}. The fundamental idea behind deep learning-based approaches for learning Koopman operators is explained in \Cref{subsec:deep_koopman}. Finally, examples of the exploitation of Koopman operator theory in modern machine learning approaches such as generative models, model pruning and reinforcement learning are discussed in \cref{subsec:ml_beyond}.

\mysubsec{Statistical learning theory for Koopman models}{\secviiapages}
\label{subsec:stat_learn}

While the previous sections focused mainly on system representations and methods together with their usage in control, they provided only limited insight into how these algorithms behave given increasing amounts of data. To understand this aspect of learning Koopman operator models, statistical learning theory provides us with effective techniques to analyze the decay of estimation errors for a given learning algorithm \cite{vapnik99:overview}. 

At the core of statistical learning theory is the assumption that data is generated by an arbitrary, unknown but fixed distribution \cite{von2011statistical}, i.e., the distribution does not need to satisfy a certain parametric form or coverage conditions. Based on this assumption, the central quantity that is investigated is the risk\looseness=-1
\begin{align}\label{eq:risk}
    \mathcal{R}(M) = \mathbb{E}_{x,y}[l(x,y,M(x))]
\end{align}
of a model $M$ under a loss function $l$. For example, we have the model $M(x) = K\basisfun_{[n]}(x)$ and the loss function 
\begin{align*}
    l(x,y,M(x))=\|\basisfun_{[n]}(y)-M(x)\|^2
\end{align*}
for the EDMD algorithm introduced in \Cref{sec:EDMD} \cite{kostic22:learning}. Given a model class and a loss, we can then define an optimal model $M^*$ as the minimizer of the risk \eqref{eq:risk}. However, we typically cannot compute the expectation in \eqref{eq:risk} analytically since we do not have direct access to the necessary probability distribution generating the data. Therefore, we approximate the expectation using training samples and use the empirical risk
\begin{align}\label{eq:emp_risk}
    \hat{\mathcal{R}}(M)=\sum\nolimits_{i=1}^N \|\basisfun_{[n]}(y^{(i)})\!-\!M(x^{(i)})\|^2
\end{align}
as proxy in practice, giving rise to the name of this approach: empirical risk minimization \cite{von2011statistical}. Hence, it can be immediately observed by comparing \eqref{eq:LS_dictionary} and \eqref{eq:emp_risk} that EDMD is a form of empirical risk minimization. It is noteworthy that statistical learning approaches often minimize a regularized empirical risk, i.e., an objective of the form $\hat{\mathcal{R}}(M)+c\|M\|$ with parameter $c\in\mathbb{R}_{>0}$ and suitable norm, such that overfitting to noise is mitigated.

As empirical risk minimization employs \eqref{eq:emp_risk} merely as a proxy, it is important for a learning algorithm that the risk of the learning algorithm converges in probability to the risk of the best model, which is commonly referred to as consistency \cite{von2011statistical}. 
A more quantifiable, commonly used requirement in modern literature is that the generalization gap/error $|\hat{\mathcal{R}}(M)-\mathcal{R}(M)|$ of a learning algorithm must vanish asymptotically. The generalization gap is an informative quantity since it provides us with insight on how well we can predict the test error, which is the sum of generalization gap and empirical risk. On the other hand, we can analyze the convergence \textit{rate} of the generalization gap, such that we obtain additional information on how fast our algorithm allows us to learn. 
Note that it is crucial in general to restrict the considered function class for analyzing the consistency and decay rate of the generalization gap of an algorithm \cite{vapnik1971uniform}. Thus, kernel methods have become a popular tool in statistical learning as their induced RKHSs are flexible yet well-specified function classes \cite{vapnik2013nature}. 
Given a restricted function class, it is necessary to additionally quantify its expressivity to enable the derivation of decay rates for the generalization gap. Common metrics for this quantification are Rademacher and Gaussian complexities \cite{bartlett02:rademacher}.

Using techniques from statistical learning theory, \cite{kostic22:learning} shows that generalization gaps of regularized EDMD satisfy\looseness=-1
\begin{align*}
    \!|\hat{\mathcal{R}}(M)\!-\!\mathcal{R}(M)| \leq \mathcal{O}\Big(\! \tfrac{\sqrt{n\log(\sfrac{N^2}{\delta})}}{\sqrt{N}}\!+\!\tfrac{\sqrt{n}\log(\sfrac{N^2}{\delta})}{N}\!\Big)\!
\end{align*}
with probability at least $1-\delta$ under the assumption that training data is independently identically distributed (iid). Similar learning rates can be obtained for estimated eigenvalues \cite{kostic22:learning, kostic23:sharp}. Moreover, extensions for the Koopman generator exist \cite{kostic24:learning}. Note that the  iid data assumption directly stems from classical results in statistical learning theory \cite{von2011statistical}, but it conflicts with the common availability of trajectory data for learning Koopman operator models. Rigorously handling non-iid data requires modified concentration inequalities \cite{kostic22:learning}.\looseness=-1

A benefit of the empirical risk minimization approach lies in its flexibility and adaptability. For example, it provides a natural framework to integrate constraints, e.g., on the sparsity of the Koopman operator approximation \cite{kostic22:learning}. Furthermore, the objective can be easily modified to consider further requirements such as Koopman invariance on data. This is demonstrated in \cite{bevanda21:koopman} by 
considering the modified objective \looseness=-1
\begin{align*}
    l(\hat{x},y,M)=\mathbb{E}_{\tau}\left\| \observable((x(t,\hat{x})) \!-\! \left( \mathcal{I}_{\lambda,t}^{[\tau_s,\tau_e]} M(x)\right) (\hat{x}) \right\|^2
\end{align*}
with $\tau\sim\mathcal{U}([\tau_s,\tau_e])$, a RKHS as function class, i.e., $M\in\mathcal{H}$, and the symmetrization operator 
\begin{align*}\nonumber
    (\mathcal{I}_{\lambda,t}^{[\tau_s,\tau_e]} \basisfun)(\hat{x})=\int\nolimits_{\tau=\tau_s}^{\tau_e} \exp(-\lambda(\tau-t)) \basisfun(x(\tau,\hat{x}))\,\mathrm{d}\tau.
\end{align*}
Notably, this symmetrization operator maps a function to its closest eigenfunction \cite{bevanda25:koopman}, such that $\mathcal{I}_{\lambda,t}^{[\tau_s,\tau_e]}\eigfun= \eigfun$. Moreover, the choice of a RKHS as a function class allows us to formulate the solution in terms of kernels. Thereby, this approach allows us to learn Koopman invariant features via kernels of the form
\begin{align}\label{eq:Koopman_kernel}
    k_{\lambda}([\hat{x};t],[\hat{x}';t'])= \mathcal{I}_{\lambda,t}^{[\tau_s,\tau_e]} k(\hat{x},\hat{x}')\overline{\mathcal{I}_{\lambda,t'}^{[\tau_s,\tau_e]}}.
\end{align}

\mysubsec{Bayesian perspective on kernel methods}{\secviibpages}
\label{subsec:Bayesian}

A common alternative to the frequentist perspective of statistical machine learning is the probabilistic view of Bayesian methods \cite{ghahramani2015probabilistic}. While Bayesian  techniques have attracted only limited attention for learning Koopman operator models in the past, many kernelized approaches can be equivalently formulated in a Bayesian setting due to the relationship between kernel and Gaussian process regression \cite{kanagawa18:gaussian}. 

A Gaussian process (GP) is the generalization of the normal distribution to function space, which defines a joint Gaussian distribution for any finite number of test points. This is compactly expressed by specifying the model $M$ as $M\sim \mathcal{GP}(0,k)$, whereby the kernel $k$ defines a covariance between points. 
Given a GP prior $M\sim \mathcal{GP}(0,k)$, regression is formulated as Bayesian inference problem. That is, we condition the prior on training pairs $(x,y)$, which results in a posterior distribution that is also Gaussian. Assuming Gaussian noise on the data $y$, the point-wise posterior distribution can be specified in closed-form via the mean and variance function
\begin{align*}
    m_M(x) &= k_{\mathcal{X}}^\top (k_{X}+\sigma_n^2 I)^{-1}y_{[N]}\\
    \sigma_M^2(x)&= k(x,x)-k_{\mathcal{X}}^\top (k_{X}+\sigma_n^2 I)^{-1}k_{\mathcal{X}},
\end{align*}
where $[y_{[N]}]=y_i$, $\sigma_n^2$ corresponds to the noise variance and $k_{\mathcal{X}}$, $k_{X}$ are defined in \Cref{subsec:kEDMD}. While we often do not know the noise variance $\sigma_n^2$ in practice, the probabilistic approach provides us with a straightforward approach to automatically tune this parameter: maximum likelihood estimation. In the case of GP regression, the logarithm of the likelihood admits the closed-form expression
\begin{align*}
    \log p(y|X) =& -\tfrac{1}{2} y_{[N]}^\top (k_{X}+\sigma_n^2 I)^{-1}y_{[N]}\\ &- \tfrac{1}{2}\log(\det( k_{X}+\sigma_n^2 I)) - \tfrac{N}{2}\log(2\pi),\nonumber
\end{align*}
whose gradients can be described by closed-form expressions as well. Therefore, it lends itself to gradient-based optimization. Notably, this log-likelihood maximization is not restricted to the noise variance $\sigma_n^2$, but can be equally employed to tune the parameters of kernels such as lengthscales, c.f., \Cref{tab:common_kernels}. Thereby, Gaussian process regression and Bayesian methods more generally provide a significant advantage over frequentist approaches, which do not offer such a structured approach for parameter optimization.

The beneficial properties of the probabilistic perspective for machine learning extend far beyond structured hyperparameter optimization. Variational approximations provide a theoretically founded approach toward scalable kernelized learning by allowing us to optimize lower bounds of the log-likelihood \cite{hensman13:gaussian}. More broadly, the wide application of GP models has led to a large variety of scalable GP approximations \cite{liu20:gaussian} and software packages that make them ready to use \cite{gardner18:gpytorch, virtanen20:scipy, rasmussen10:gaussian, lederer2021gaussian}. Furthermore, the availability of regression error bounds \cite{srinivas12:information,chowdhury2017kernelized,lederer2019uniform, fiedler2021practical} and information-theoretic sample complexity guarantees \cite{srinivas12:information} make GP regression appealing. Note that these guarantees immediately extend to the parametric setting since linear Bayesian regression is a special case of GP regression.

Using dedicated Koopman kernel designs such as \eqref{eq:Koopman_kernel}, these beneficial properties can be straightforwardly exploited~\cite{bevanda25:koopman}. However, it mostly remains an open research question how the Bayesian perspective can be beneficially employed more widely for Koopman operator learning.

\mysubsec{Deep learning-based Koopman models}{\secviicpages}
\label{subsec:deep_koopman}

The end of the 20th and the beginning of the 21st century has seen a revolutionary 
increase in the availability of data. Indeed, we are in the middle of the {\it sensing} revolution,
where sensing is used in the broadest meaning of data acquisition. Most of this data goes unprocessed,
unanalyzed, and consequently, unused. This causes missed opportunities, in 
domains of vast societal importance - health, commerce, technology, security, just to mention some.

A variety of mathematical methods have emerged out of that need. Perhaps the most popular methodology, Deep Neural Networks, has as the underlying learning elements ``neuron functions", that
are modeled after biological neurons. The algorithms based on deep learning have
achieved substantial success in image recognition, speech recognition and natural language processing, deploying the ``supervised" machine learning philosophy.
Convolutional neural networks \cite{lecun2018power} provided a superstructure to the deep neural network architecture   that resembles the organization of the animal visual cortex.
This led to an enormous success in image recognition, and even in realistic image generation, via Generative Adversarial Networks (GAN's) \cite{goodfellow2014generative} and Flow Matching \cite{lipman2023flow}. 
These examples motivate the exploitation of the flexibility and scalability of deep learning techniques for Koopman operator learning approaches.

``Deep Koopman" methods \cite{lusch18:deep,yeung2019learning,otto2019linearly} seek a finite-dimensional approximation of the
Koopman operator by combining  operator theory with deep
learning. The objective of the Deep Koopman architectures is to learn observables spanning a
finite-dimensional invariant subspace.
To this end, a deep encoder network $\basisfunvec:\mathcal{X}\to\mathbb{R}^n$
and decoder $\decoderfunvec:\mathbb{R}^n\to\mathcal{X}$ are trained jointly with
a matrix $K\in\mathbb{R}^{n\times n}$ so that the lifted coordinates
$z=\basisfunvec(x)$ evolve linearly, $z^+ = Kz$. 
The training objective balances three losses. The \emph{prediction loss}\looseness=-1
\[
  \mathcal{L}_{\mathrm{pred}}
  = \sum\nolimits_{i=1}^{N}\bigl\|\basisfunvec((x^{(i)})^+)-K\basisfunvec(x^{(i)})\bigr\|^2
\]
enforces linear evolution in the latent space. The \emph{reconstruction
loss}
\[
  \mathcal{L}_{\mathrm{rec}}
  = \sum\nolimits_{i=1}^{N}\bigl\|\decoderfunvec(\basisfunvec(x^{(i)}))-x^{(i)}\bigr\|^2
\]
ensures the embedding is invertible. The \emph{multi-step linearity
loss}
\[
  \mathcal{L}_{\mathrm{lin}}
  = \sum\nolimits_{i=1}^{N}\sum\nolimits_{j=1}^{H}
    \bigl\|\basisfunvec((x^{(i)})^{(+j)})-K^j\basisfunvec(x^{(i)})\bigr\|^2
\]
with $(x^{(i)})^{(+j)}$ denoting the $j$ step successor state
suppresses error accumulation over long rollouts. The full objective is $\mathcal{L}
  = \mathcal{L}_{\mathrm{pred}}
  + \alpha\,\mathcal{L}_{\mathrm{rec}}
  + \beta\,\mathcal{L}_{\mathrm{lin}}$, with weights $\alpha,\beta>0$. 

This set of architectures is remarkably similar to the modern AI ``world model" architectures such as JEPA \cite{lecun2022path} as has been observed in \cite{ruiz2026koopman}.

\mysubsec{Exploiting Koopmanism in machine learning}{\secviicpages}
\label{subsec:ml_beyond}

Even though deep learning techniques have been successfully employed in many applications as discussed in \Cref{subsec:deep_koopman}, their success is essentially limited to 
static pattern recognition or generation tasks. In particular, deep learning methodologies are less successful in dynamically changing contexts \cite{miller2024survey}, present, for example, in autonomous driving.
It is interesting that, in contrast to biological modules responsible for vision,  even the basic issue of finding specific brain structures that are responsible for perception of time, and thus understanding of dynamics, are still being investigated~\cite{roseboom2019activity}.\looseness=-1

As indicated in the prior text, Koopman operator theory has recently emerged as the main candidate for machine learning of dynamical processes \cite{mezi05,budisicetal:2012,Mauroy20Susu20:book}. The connections between the machine learning and Koopman operator communities are growing rapidly, in many directions, and Koopman operator methods are increasingly used in machine learning itself.\looseness=-1
\miniskip

\noindent \textbf{Flow matching and diffusion models.}
Flow matching and diffusion model technologies are generative machine learning methods that employ an iterative sampling procedure during inference to sample from a learned probability distribution. By interpreting the iterative sampling process as a dynamical system, Koopman operator techniques immediately become applicable. In \cite{berman25:one}, this insight is used to design a Koopman-based offline diffusion distillation algorithm, which compresses noisy-to-clean endpoint pairs into a single latent linear evolution to reduce the computational complexity of sampling.
Subsequent work addresses the trajectory information that this endpoint view leaves implicit: \cite{bai2026hierarchical} recovers intermediate diffusion states through hierarchical Koopman subspaces, while \cite{turan2025unfolding} enforces consistency with the teacher vector field in conditional flow matching, preserving the generative path rather than only its endpoints. Extending this line beyond image generation, \cite{zheng2026dynamic} proposes a dynamic Koopman distillation approach for real-time closed-loop robot control, where a state-dependent factorized transition enables low-latency diffusion-policy inference under receding-horizon execution. Together, these works show how Koopman operator methods benefit modern generative machine learning models.\looseness=-1
\miniskip

\noindent \textbf{Neural network pruning.} The removal of parameters to yield sparsity, which is commonly referred to as pruning, has attracted sustained 
interest since the
discovery of the \emph{lottery ticket hypothesis}~\cite{frankle2019}.
A central open question is why the na\"{i}ve strategy of removing
weights by magnitude is as competitive as far more elaborate
algorithms, and how the training trajectory itself encodes information
about which parameters are expendable.
Redman et al.~\cite{redman22:an} address this directly by recasting
the gradient-descent training trajectory as a \emph{dynamical system}
and applying Koopman operator theory to its analysis.
Viewing the sequence of parameter vectors $\{\mathbf{w}_n\}$ produced
during training as an orbit in weight space, the Koopman mode
decomposition 
of this trajectory extracts spectral information
that naturally identifies which weights are dynamically inert.
The key result is that the resulting Koopman-theoretic pruning
algorithms \emph{unify} magnitude and gradient-based pruning: both
emerge as special cases of projecting the weight trajectory onto
dominant Koopman modes, thereby providing the first operator-theoretic
justification for magnitude pruning's robustness, particularly in the
pre-convergence regime where empirical success had previously lacked
theoretical explanation.
Building on this spectral view of training, Redman et
al.~\cite{redman24:identifying} develop a framework for identifying
\emph{equivalent training dynamics} across architectures by comparing
Koopman eigenvalues extracted from their respective weight
trajectories.
Two networks are declared dynamically equivalent if their Koopman
spectra are conjugate, a condition that can be checked from data alone
without knowledge of the loss landscape geometry.
The framework is validated on convolutional architectures and on
Transformers that do and do not undergo \emph{grokking} (delayed
generalization), demonstrating that 
a comparison of Koopman eigenvalues
correctly identifies the known equivalence between online mirror
descent and online gradient descent.
Taken together, these works establish a principled dynamical-systems
language for two central phenomena in deep learning -- why sparse
networks can be found efficiently, and when two apparently different
training procedures are in fact equivalent -- and suggest that the
spectral geometry of the Koopman operator is a natural invariant for
characterizing neural network training dynamics.\looseness=-1
\miniskip

\noindent \textbf{Reinforcement learning.}
The intersection of reinforcement learning (RL) and Koopman operator
theory has emerged as a productive research direction, motivated by the
observation that the Bellman equation governing optimal control is
itself a dynamical object amenable to linearization. Recall that in a
Markov Decision Process (MDP) with state $s$, action $a$, reward $r$,
and discount factor $\gamma\in(0,1)$, the optimal value function
$V^*(s)$ satisfies the Bellman equation
\[
  V^*(s)
  = \max_{a}\Bigl[r(s,a)
    + \gamma\,\mathbb{E}_{s'\sim P(\cdot|s,a)}\bigl[V^*(s')\bigr]\Bigr],
\]
which is intractable for nonlinear, high-dimensional systems.
The key insight of \cite{rozwood23:koopmanassisted} is that the expectation operator
on the right-hand side is precisely the action of a
\emph{controlled Koopman operator} $\mathcal{K}^a$ on the observable
$V^*$:
\[
  \mathcal{K}^a V^*(s)
  = \mathbb{E}_{s'\sim P(\cdot|s,a)}\bigl[V^*(s')\bigr].
\]
By parameterizing $\mathcal{K}^a$ with the control action one obtains
a \emph{Koopman tensor} $\mathcal{T}$ such that
\[
  \mathcal{K}^a \basisfunvec^s(s)
  \approx \mathcal{T}\bigl[\basisfunvec^s(s)\otimes\basisfunvec^a
  (a)\bigr],
\]
where $\basisfunvec^s:\mathcal{S}\to\mathbb{R}^d$ and $\basisfunvec^a:\mathcal{A}\to\mathbb{R}^m$
are learned observable
dictionaries for the state and action spaces respectively, and
$\otimes$ denotes the Kronecker product.
Lifting the Bellman equation into this linear framework renders
Hamilton--Jacobi--Bellman-based methods substantially more tractable.\looseness=-1

A complementary direction exploits the spectral structure of the
Koopman operator to address the out-of-distribution problem endemic to
offline RL~\cite{weissenbacher2022}.
The central observation is that the Koopman eigenfunctions
$\varphi_i$ satisfying $\mathcal{K}\eigfun_i = \lambda_i\eigfun_i$ encode the \emph{symmetries} of the system dynamics: if
$\eigfun_i(s) = \eigfun_i(s')$, then $s$ and $s'$ are dynamically
equivalent and can be used to augment one another in the training
dataset.
By learning a Koopman latent representation of the environment and
reading off its symmetry group from the eigenstructure of $\mathcal{K}$,
\cite{weissenbacher2022} construct a principled data augmentation
scheme, referred to as Koopman Forward Conservative Q-learning (KFC), that
extends the static offline dataset along symmetry trajectories
parametrized by $\varepsilon$:
\[
  (s,a,r,s') \;\longmapsto\;
  \bigl(s+\varepsilon,\,a,\,r,\,s'+\varepsilon\bigr),
  \quad \varepsilon \in \mathcal{E}_{\mathrm{sym}},
\]
and is shown to consistently improve over state-of-the-art model-free
Q-learning baselines.
Further threads apply Koopman-linearized models directly to policy
optimization in robotics~\cite{retchin2023,sinha2022}.

What emerges is a powerful framework for unsupervised learning from small amounts of data, enabling self-supervised learning \cite{lecun2018power} that is much more in line with the theory of human learning than the machine learning methods of the second wave \cite{prabhakar2017powerful}.

\section{Outlook}
\label{sec:outlook}

Several promising directions and open research questions naturally emerge from this tutorial.
On the applied side, an important next phase for the field of Koopman operator methods is to connect the recent theoretical and methodological progress more directly with high-impact applications. This includes deploying Koopman-based models in robotics settings marked by contact-rich, dexterous, non-prehensile, soft, and hybrid dynamics, and exploiting the theoretical insights provided by the Koopman operator framework in modern, large-scale AI models for language, vision, and planning. 
From a methodological perspective, an important question concerns the refinement of dictionaries to reduce modeling error
while simultaneously capturing relevant dynamical behaviors. 
Here, subspace pruning methods with theoretical guarantees offer a promising direction for the future. Additionally, the probabilistic perspective on Koopman operator learning is severely underexplored despite its potential to enable scalable training methods with rigorous theoretical guarantees amenable to uncertainty quantification. Hybrid extensions that unify discrete-time and continuous-time systems are a further interesting research direction given the separate handling of these system classes in existing methods.
Finally, various important research questions about theoretical aspects of Koopman operator methods remain. For example, the benefits of (approximate) Koopman invariance on the complexity of learning problems is not well understood. Similarly, deep learning-based Koopman models are lacking a suitable extension of the existing theoretical framework for analyzing Koopman operator learning. Even though this outlook only covers exemplary research problems, we believe that it illustrates the exciting future research opportunities in the field of Koopman operator methods.

\section*{acknowledgements}

I.M.\ was supported by  AFOSR award number FA9550-22-1-0531 and Defense Advanced Research Projects Agency (DARPA) under Agreement No. HR00112590152.
A.L.\ and K.W.\ are thankful for support by the  Deutsche Forschungsgemeinschaft (DFG, German
Research Foundation) within the research unit ALeSCo – Active Learning
in Systems and Control, project-ID 535860958. Moreover, A.L.\ gratefully acknowledges support by a start-up grant of the National University of Singapore. J.C.\ was supported by ONR Award N00014-23-1-2353.

\label{sec:references}
\bibliographystyle{ieeetr}
\bibliography{references_clean.bib}

\end{document}